\documentclass[acmsmall,screen]{acmart}

\setcopyright{rightsretained}
\acmPrice{}
\acmDOI{10.1145/3290382}
\acmYear{2019}
\copyrightyear{2019}
\acmJournal{PACMPL}
\acmNumber{POPL} 
\acmArticle{69}
\acmYear{2019}
\acmMonth{1}

\usepackage{booktabs}
\usepackage{subcaption}
\usepackage{microtype}
\usepackage{amsmath}
\usepackage{amssymb,amsfonts}
\usepackage{stmaryrd}
\usepackage{hyperref}
\usepackage{multicol}
\usepackage{color}
\usepackage{algorithm}
\usepackage[noend]{algpseudocode}
\usepackage{graphicx}
\usepackage{tikz}
\usepackage{thm-restate}
\usepackage{mathpartir}
\usepackage{xspace}
\usepackage{wrapfig}
\usepackage{ifthen}
\usepackage{xcolor}
\usepackage{stackengine}

\usepackage[capitalise]{cleveref}

\definecolor{StringRed}{rgb}{.637,0.082,0.082}
\definecolor{CommentGreen}{rgb}{0.0,0.55,0.3}
\definecolor{KeywordBlue}{rgb}{0.0,0.3,0.55}
\definecolor{LinkColor}{rgb}{0.55,0.0,0.3}
\definecolor{CiteColor}{rgb}{0.55,0.0,0.3}
\definecolor{HighlightColor}{rgb}{0.0,0.0,0.0}

\definecolor{grey}{rgb}{0.5,0.5,0.5}
\definecolor{red}{rgb}{1,0,0}
\definecolor{darkgreen}{rgb}{0.0,0.7,0.0}

\hypersetup{%
  linktocpage=true, pdfstartview=FitV,
  breaklinks=true, pageanchor=true, pdfpagemode=UseOutlines,
  plainpages=false, bookmarksnumbered, bookmarksopen=true, bookmarksopenlevel=3,
  hypertexnames=true, pdfhighlight=/O,
  colorlinks=true,linkcolor=LinkColor,citecolor=CiteColor,
  urlcolor=LinkColor
}

\usetikzlibrary{calc,shapes,shapes.misc,shadows,arrows,decorations.pathreplacing,decorations.markings,decorations.pathmorphing,positioning,automata,fadings,fit,arrows}
\usetikzlibrary{shapes.multipart,matrix}
\usetikzlibrary{shapes.callouts} 
\usetikzlibrary{patterns,backgrounds}

\newlength{\hatchspread}
\newlength{\hatchthickness}
\newlength{\hatchshift}
\newcommand{\hatchcolor}{}
\tikzset{hatchspread/.code={\setlength{\hatchspread}{#1}},
         hatchthickness/.code={\setlength{\hatchthickness}{#1}},
         hatchshift/.code={\setlength{\hatchshift}{#1}},
         hatchcolor/.code={\renewcommand{\hatchcolor}{#1}}}
\tikzset{hatchspread=10pt,
         hatchthickness=4pt,
         hatchshift=0pt,
         hatchcolor=black}
\pgfdeclarepatternformonly[\hatchspread,\hatchthickness,\hatchshift,\hatchcolor]
   {custom north west lines}
   {\pgfqpoint{\dimexpr-2\hatchthickness}{\dimexpr-2\hatchthickness}}
   {\pgfqpoint{\dimexpr\hatchspread+2\hatchthickness}{\dimexpr\hatchspread+2\hatchthickness}}
   {\pgfqpoint{\dimexpr\hatchspread}{\dimexpr\hatchspread}}
   {
    \pgfsetlinewidth{\hatchthickness}
    \pgfpathmoveto{\pgfqpoint{0pt}{\dimexpr\hatchspread+\hatchshift}}
    \pgfpathlineto{\pgfqpoint{\dimexpr\hatchspread+0.15pt+\hatchshift}{-0.15pt}}
    \ifdim \hatchshift > 0pt
      \pgfpathmoveto{\pgfqpoint{0pt}{\hatchshift}}
      \pgfpathlineto{\pgfqpoint{\dimexpr0.15pt+\hatchshift}{-0.15pt}}
    \fi
    \pgfsetstrokecolor{\hatchcolor}
    \pgfusepath{stroke}
   }

\pgfdeclarepatternformonly[\hatchspread,\hatchthickness,\hatchshift,\hatchcolor]
   {custom north east lines}
   {\pgfqpoint{\dimexpr-2\hatchthickness}{\dimexpr-2\hatchthickness}}
   {\pgfqpoint{\dimexpr\hatchspread+2\hatchthickness}{\dimexpr\hatchspread+2\hatchthickness}}
   {\pgfqpoint{\dimexpr\hatchspread}{\dimexpr\hatchspread}}
   {
    \pgfsetlinewidth{\hatchthickness}
    \pgfpathmoveto{\pgfqpoint{\dimexpr\hatchshift-0.15pt}{-0.15pt}}
    \pgfpathlineto{\pgfqpoint{\dimexpr\hatchspread+0.15pt}{\dimexpr\hatchspread-\hatchshift+0.15pt}}
    \ifdim \hatchshift > 0pt
      \pgfpathmoveto{\pgfqpoint{-0.15pt}{\dimexpr\hatchspread-\hatchshift-0.15pt}}
      \pgfpathlineto{\pgfqpoint{\dimexpr\hatchshift+0.15pt}{\dimexpr\hatchspread+0.15pt}}
    \fi
    \pgfsetstrokecolor{\hatchcolor}
    \pgfusepath{stroke}
   }

\tikzstyle{issuedStyle}=[pattern=custom north east lines, hatchcolor=colorISS, rounded corners]
\tikzstyle{coveredStyle}=[]

\newcommand{\setBox}[2]{
    \draw[#1] ($(#2)  + (-1.2,-0.4)$) rectangle ++(2.4,0.8);
}
\newcommand{\bigSetBox}[2]{
    \draw[#1] ($(#2)  + (-1.3,-0.5)$) rectangle ++(2.6,1.0);
}

\newcommand{\coveredBox}[1]{
  \setBox{coveredStyle}{#1}
}
\newcommand{\issuedBox}[1]{
  \setBox{issuedStyle}{#1}
}
\newcommand{\issuedCoveredBox}[1]{
  \bigSetBox{coveredStyle}{#1}
  \issuedBox{#1}
}

\AtEndPreamble{
\theoremstyle{acmplain}

\newtheorem*{notation*}{Notation}
}
\theoremstyle{acmdefinition}
\newtheorem{remark}{Remark}

\newcommand{\citeapp}[2]{\cref{#1}}

\crefformat{section}{#2\S{}#1#3}
\Crefname{section}{Section}{Sections}
\Crefformat{section}{Section #2#1#3}

\crefformat{subsection}{#2\S{}#1#3}  
\Crefname{subsection}{Section}{Sections}
\Crefformat{subsection}{Section #2#1#3}

\crefname{figure}{\text{Fig.}}{\text{Figures}}
\Crefname{figure}{\text{Figure}}{\text{Figures}}
\crefname{corollary}{\text{Corollary}}{\text{corollaries}}
\Crefname{corollary}{\text{Corollary}}{\text{Corollaries}}
\crefname{lemma}{\text{Lemma}}{\text{Lemmas}}
\Crefname{lemma}{\text{Lemma}}{\text{Lemmas}}
\crefname{proposition}{\text{Prop.}}{\text{Prop.}}
\Crefname{proposition}{\text{Proposition}}{\text{Propositions}}
\crefname{definition}{\text{Def.}}{\text{Definitions}}
\Crefname{definition}{\text{Definition}}{\text{Definitions}}
\crefname{notation}{\text{Notation}}{\text{Notations}}
\Crefname{notation}{\text{Notation}}{\text{Notations}}
\crefname{theorem}{\text{Theorem}}{\text{Theorems}}
\Crefname{theorem}{\text{Theorem}}{\text{Theorems}}
\crefname{conjecture}{\text{Conj.}}{\text{Conjectures}}
\Crefname{conjecture}{\text{Conjecture}}{\text{Conjectures}}

\newcommand{\textcode}[1]{\texorpdfstring{\texttt{#1}}{#1}}
\newcommand{\kw}[1]{\textbf{\textcode{#1}}}

\newcommand{\ALT}{\;\;|\;\;}

\newcommand{\ie}{\emph{i.e.,} }
\newcommand{\eg}{\emph{e.g.,} }

\newcommand{\wrt}{w.r.t.~}
\newcommand{\aka}{a.k.a.~}

\newcommand{\inarrC}[1]{\begin{array}{@{}c@{}}#1\end{array}}

\newcommand{\inarr}[1]{\begin{array}{@{}l@{}}#1\end{array}}
\newcommand{\inarrII}[2]{\begin{array}{@{}l@{~~}||@{~~}l@{}}\inarr{#1}&\inarr{#2}\end{array}}
\newcommand{\inarrIII}[3]{\begin{array}{@{}l@{~~}||@{~~}l@{~~}||@{~~}l@{}}\inarr{#1}&\inarr{#2}&\inarr{#3}\end{array}}
\newcommand{\inarrIV}[4]{\begin{array}{@{}l@{~~}||@{~~}l@{~~}||@{~~}l@{~~}||@{~~}l@{}}\inarr{#1}&\inarr{#2}&\inarr{#3}&\inarr{#4}\end{array}}

\renewcommand{\comment}[1]{\color{teal}{~~\texttt{/\!\!/}\textit{#1}}}

\newcommand{\set}[1]{\{{#1}\}}

\newcommand{\fpfn}{\mathrel{\stackrel{\mathsf{fin}}{\rightharpoonup}}}
\newcommand{\st}{\; | \;}
\newcommand{\N}{{\mathbb{N}}}
\newcommand{\Q}{{\mathbb{Q}}}
\newcommand{\dom}[1]{\textit{dom}{({#1})}}
\newcommand{\codom}[1]{\textit{codom}{({#1})}}

\newcommand{\tup}[1]{{\langle{#1}\rangle}}
\newcommand{\nin}{\not\in}
\newcommand{\suq}{\subseteq}
\newcommand{\sqsuq}{\sqsubseteq}

\newcommand{\size}[1]{|{#1}|}

\newcommand{\maketil}[1]{{#1}\ldots{#1}}
\newcommand{\til}{\maketil{,}}

\newcommand{\rst}[1]{|_{#1}}
\newcommand{\imm}[1]{{#1}{\rst{\text{imm}}}}

\newcommand{\defeq}{\triangleq}
\newcommand{\powerset}[1]{\mathcal{P}({#1})}

\renewcommand{\implies}{\Rightarrow}

\colorlet{colorPO}{gray!60!black}
\colorlet{colorRF}{green!60!black}
\colorlet{colorMO}{orange}
\colorlet{colorFR}{purple}
\colorlet{colorECO}{red!80!black}
\colorlet{colorSYN}{green!40!black}
\colorlet{colorHB}{blue}
\colorlet{colorPPO}{magenta}
\colorlet{colorPB}{olive}
\colorlet{colorSBRF}{olive}
\colorlet{colorRMW}{olive!70!black}
\colorlet{colorRS}{blue}
\colorlet{colorRELEASE}{blue!70!black}
\colorlet{colorSC}{olive!40!black}
\colorlet{colorPSC}{olive!40!black}
\colorlet{colorREL}{olive}
\colorlet{colorCONFLICT}{olive}
\colorlet{colorRACE}{olive}
\colorlet{colorWB}{orange!70!black}
\colorlet{colorSCB}{violet}
\colorlet{colorDETOUR}{teal}
\colorlet{colorDEPS}{violet}
\colorlet{colorFENCE}{olive}
\colorlet{colorCOV}{magenta!20}
\colorlet{colorISS}{blue!10!white}

\tikzset{
   every path/.style={>=stealth},
   po/.style={->,color=colorPO,shorten >=-0.5mm,shorten <=-0.5mm},
   sw/.style={->,color=colorSYN,shorten >=-0.5mm,shorten <=-0.5mm},
   sc/.style={->,color=colorSC,dotted,thick,shorten >=-0.5mm,shorten <=-0.5mm},
   rf/.style={->,color=colorRF,dashed,,shorten >=-0.5mm,shorten <=-0.5mm},
   hb/.style={->,color=colorHB,thick,shorten >=-0.5mm,shorten <=-0.5mm},
   mo/.style={->,color=colorMO,dotted,very thick,shorten >=-0.5mm,shorten <=-0.5mm},
   no/.style={->,dotted,thick,shorten >=-0.5mm,shorten <=-0.5mm},
   fr/.style={->,color=colorFR,dotted,thick,shorten >=-0.5mm,shorten <=-0.5mm},
   deps/.style={->,color=colorDEPS,dotted,thick,shorten >=-0.5mm,shorten <=-0.5mm},
   ppo/.style={->,color=colorPPO,shorten >=-0.5mm,shorten <=-0.5mm},
   rmw/.style={->,color=colorRMW,thick,shorten >=-0.5mm,shorten <=-0.5mm},
   detour/.style={->,color=colorDETOUR,shorten >=-0.5mm,shorten <=-0.5mm},
}

\newcommand{\rlx}{\mathtt{rlx}}
\newcommand{\rel}{{\mathtt{rel}}}
\newcommand{\acq}{{\mathtt{acq}}}
\newcommand{\acqrel}{{\mathtt{acqrel}}}
\newcommand{\sco}{{\mathtt{sc}}}

\newcommand{\full}{{\mathtt{sy}}}
\newcommand{\ld}{{\mathtt{ld}}}

\newcommand{\isync}{{\mathtt{isync}}}
\newcommand{\lwsync}{{\mathtt{lwsync}}}
\newcommand{\sync}{{\mathtt{sync}}}

\newcommand{\normal}{{\mathtt{normal}}}
\newcommand{\strong}{{\mathtt{strong}}}

\newcommand{\isex}{{\mathtt{ex}}}
\newcommand{\isnotex}{{\operatorname{\mathtt{not-ex}}}}

\newcommand{\initev}[1]{\tup{\texttt{init}~{#1}}}

\newcommand{\rlab}[3]{{\lR}^{#1}({#2},{#3})}
\newcommand{\erlab}[4]{
\ifthenelse{\equal{#2}{}}{{\lR}^{#1}_{#4}}
{\ifthenelse{\equal{#3}{}\and\equal{#1}{}\and\equal{#2}{}}{{\lR}^({#2})}
{\ifthenelse{\equal{#3}{}}{{\lR}^{#1}_{#4}({#2})}
{{\lR}^{#1}_{#4}({#2},{#3})}}}}
\newcommand{\ewlab}[4]{
\ifthenelse{\equal{#2}{}}{{\lW}^{#1}_{#4}}
{\ifthenelse{\equal{#3}{}\and\equal{#1}{}\and\equal{#2}{}}{{\lW}^({#2})}
{\ifthenelse{\equal{#3}{}}{{\lW}^{#1}_{#4}({#2})}
{{\lW}^{#1}_{#4}({#2},{#3})}}}}

\newcommand{\prlab}[2]{{\lR}({#1},{#2})}
\newcommand{\wlab}[3]{{\lW}^{#1}({#2},{#3})}
\newcommand{\pwlab}[2]{{\lW}({#1},{#2})}
\newcommand{\flab}[1]{{\lF}^{#1}}

\newcommand{\lE}{{\mathtt{E}}}

\newcommand{\lR}{{\mathtt{R}}}
\newcommand{\lW}{{\mathtt{W}}}

\newcommand{\lQ}{{\mathtt{Q}}}
\newcommand{\lL}{{\mathtt{L}}}
\newcommand{\lU}{{\mathtt{RMW}}}
\newcommand{\lF}{{\mathtt{F}}}

\newcommand{\lAT}{{\mathtt{At}}}

\newcommand{\lWstrong}{\ewlab{}{}{}{\strong}}

\newcommand{\lLAB}{{\mathtt{lab}}}
\newcommand{\lTID}{{\mathtt{tid}}}
\newcommand{\lSN}{{\mathtt{sn}}}

\newcommand{\lLOC}{{\mathtt{loc}}}
\newcommand{\lMOD}{{\mathtt{mod}}}
\newcommand{\lVAL}{{\mathtt{val}}}

\newcommand{\lX}{\mathtt{x}}

\newcommand{\lAR}{\mathtt{ar}}

\newcommand{\lPO}{{\color{colorPO}\mathtt{po}}}
\newcommand{\lRF}{{\color{colorRF} \mathtt{rf}}}
\newcommand{\lRMW}{{\color{colorRMW} \mathtt{rmw}}}
\newcommand{\lCO}{{\color{colorMO} \mathtt{co}}}

\newcommand{\lFR}{{\color{colorFR} \mathtt{fr}}}

\newcommand{\lECO}{{\color{colorECO} \mathtt{eco}}}

\newcommand{\lRS}{{\color{colorRS}\mathtt{rs}}}
\newcommand{\lRELEASE}{{\color{colorRELEASE}\mathtt{release}}}
\newcommand{\lSW}{{\color{colorSYN}\mathtt{sw}}}
\newcommand{\lHB}{{\color{colorHB}\mathtt{hb}}}
\newcommand{\lDOB}{{\mathtt{dob}}}
\newcommand{\lFWBOB}{{\mathtt{fwbob}}}
\newcommand{\lBOB}{{\mathtt{bob}}}
\newcommand{\lAOB}{{\mathtt{aob}}}
\newcommand{\lOBS}{{\mathtt{obs}}}

\newcommand{\lSC}{{\color{colorSC}\mathtt{sc}}}

\newcommand{\lRSs}{{\color{colorRS}\mathtt{rs}}\smodel}
\newcommand{\lRELEASEs}{{\color{colorRELEASE}\mathtt{release}}\smodel}
\newcommand{\lSWs}{{\color{colorSYN}\mathtt{sw}}\smodel}
\newcommand{\lHBs}{{\color{colorHB}\mathtt{hb}}\smodel}

\newcommand{\lPSC}{{\color{colorPSC} \mathtt{psc}}}

\newcommand{\lDETOUR}{{{\color{colorDETOUR}\mathtt{detour}}}}
\newcommand{\lDEPS}{{{\color{colorDEPS}\mathtt{deps}}}}
\newcommand{\lCTRL}{{{\color{colorDEPS}\mathtt{ctrl}}}}
\newcommand{\lDATA}{{{\color{colorDEPS}\mathtt{data}}}}
\newcommand{\lADDR}{{{\color{colorDEPS}\mathtt{addr}}}}
\newcommand{\lPPO}{{{\color{colorPPO}\mathtt{ppo}}}}
\newcommand{\lRMWDEP}{{{\color{colorDEPS}\mathtt{casdep}}}}

\newcommand{\lmakeE}[1]{#1\mathtt{e}}
\newcommand{\lRFE}{\lmakeE{\lRF}}
\newcommand{\lCOE}{\lmakeE{\lCO}}
\newcommand{\lFRE}{\lmakeE{\lFR}}
\newcommand{\lmakeI}[1]{#1\mathtt{i}}
\newcommand{\lRFI}{\lmakeI{\lRF}}
\newcommand{\lCOI}{\lmakeI{\lCO}}
\newcommand{\lFRI}{\lmakeI{\lFR}}

\newcommand{\lSYNC}{\mathtt{sync}}
\newcommand{\lLWSYNC}{\mathtt{lwsync}}
\newcommand{\lCTRLISYNC}{\operatorname{\mathtt{ctrl-isync}}}
\newcommand{\lFENCE}{\mathtt{fence}}

\newcommand{\lPROP}{\mathtt{prop}}
\newcommand{\lRDW}{{\mathtt{rdw}}}
\newcommand{\lHBP}{{\color{colorHB}\mathtt{hb}}_\mathtt{p}}
\newcommand{\lPPOP}{{{\color{colorPPO}\mathtt{ppo}}_\mathtt{p}}}
\newcommand{\lii}{{\mathtt{ii}}}
\newcommand{\lic}{{\mathtt{ic}}}
\newcommand{\lci}{{\mathtt{ci}}}
\newcommand{\lcc}{{\mathtt{cc}}}

\newcommand{\Tid}{\mathsf{Tid}}
\newcommand{\Loc}{\mathsf{Loc}}
\newcommand{\Val}{\mathsf{Val}}
\newcommand{\Lab}{\mathsf{Lab}}

\newcommand{\Event}{\mathsf{Event}}
\newcommand{\Init}{\mathsf{Init}}

\newcommand{\Inst}{\mathsf{Inst}}
\newcommand{\Exp}{\mathsf{Exp}}
\newcommand{\Reg}{\mathsf{Reg}}
\newcommand{\Sprog}{\mathsf{SProg}}

\newcommand{\smodel}{_{\text{RC11}}}
\newcommand{\POWER}{\ensuremath{\mathsf{POWER}}\xspace}

\newcommand{\IMM}{\ensuremath{\mathsf{IMM}}\xspace}

\newcommand{\ARM}{\ensuremath{\mathsf{ARM}}\xspace}

\newcommand{\prog}{prog}
\newcommand{\sprog}{sprog}

\newcommand{\compile}[1]{{(\!|} #1 {|\!)}}

\newcommand{\lDMBSY}{\mathtt{F^{\full}}}
\newcommand{\lDMBLD}{\mathtt{F^{\ld}}}

\newcommand{\msg}[5]{\tup{#1:#2@(#3,#4],#5}}
\newcommand{\msgRlx}[3]{\tup{#1:#2@#3}}
\newcommand{\Promise}{\ensuremath{\mathsf{Promise}}\xspace}

\newcommand{\view}{view}
\newcommand{\viewCur}{{\sf cur}}
\newcommand{\viewAcq}{{\sf acq}}
\newcommand{\viewRel}{{\sf rel}}

\newcommand{\loc}{x}
\newcommand{\locy}{y}
\newcommand{\locz}{z}

\newcommand{\tid}{i}
\newcommand{\val}{v}
\newcommand{\tidmod}[1]{_{#1}}

\newcommand{\readInst }[3]{#2 \;{:=}\;[#3]^{#1}}
\newcommand{\fenceInst}[1]{\kw{fence}^{#1}}

\newcommand{\ifGotoInst}[2]{\kw{if} \; #1 \; \kw{goto} \; #2}
\newcommand{\writeInst}[3]{[#2]^{#1}\;{:=}\;#3}
\newcommand{\assignInst}[2]{#1\;{:=}\;#2}
\newcommand{\incInst}[6]{#3 \;{:=}\;\kw{FADD}_{#6}^{#1,#2}({#4},{#5})}
\newcommand{\casInst}[7]{#3 \;{:=}\;\kw{CAS}_{#7}^{#1,#2}({#4},{#5},{#6})}

\newcounter{mylabelcounter}

\makeatletter
\newcommand{\labelAxiom}[2]{%
\hfill{\normalfont\textsc{(#1)}}\refstepcounter{mylabelcounter}
\immediate\write\@auxout{%
  \string\newlabel{#2}{{\unexpanded{\normalfont\textsc{#1}}}{\thepage}{{\unexpanded{\normalfont\textsc{#1}}}}{mylabelcounter.\number\value{mylabelcounter}}{}}
}%
}
\makeatother

\newcommand{\squishlist}[1][$\bullet$]{
 \begin{list}{#1}
  { \setlength{\itemsep}{0pt}
     \setlength{\parsep}{0pt}
     \setlength{\topsep}{1pt}
     \setlength{\partopsep}{0pt}
     \setlength{\leftmargin}{1.2em}
     \setlength{\labelwidth}{0.5em}
     \setlength{\labelsep}{0.4em} } }
\newcommand{\squishend}{
  \end{list}  }

\newcommand{\coverable}{{\sf Coverable}}
\newcommand{\issuable}{{\sf Issuable}}

\newcommand{\IssuedSet}{I}
\newcommand{\CoveredSet}{C}

\newcommand{\travConfigP}{{\rm partial\text{-}trav\text{-}config}}

\newcommand{\myrightarrow}[1]{
\ifthenelse{\equal{#1}{}}{\xrightarrow{}}{\mathrel{\raisebox{-2pt}{$\xrightarrow{#1}$}}}}

\newcommand{\travConfigStep}{\myrightarrow{\rm STC}}
\newcommand{\etravStep}{\myrightarrow{}}

\newcommand{\View}{\ensuremath{\mathsf{View}}\xspace}

\newcommand{\letdef}[2]{\kw{let} \; #1 = #2 \; \kw{in}}

\newcommand{\lURR}{\mathtt{vf}}
\newcommand{\lBVF}{\mathtt{bvf}}
\newcommand{\lBVFrlx}{\lBVF^{\rlx}}

\newcommand{\Tto}{T}
\newcommand{\Tfrom}{F}

\newcommand{\PhiD}{\Psi}

\allowdisplaybreaks

\newcommand{\txtsub}[2]{{#1}_{\text{#2}}}

\newcommand{\tcom}{{\mathcal{V}}}
\newcommand{\gsco}{{\mathcal{S}}}
\newcommand{\mem}{M}
\newcommand{\gts}{\mathcal{T\!S}}
\newcommand{\lts}{\mathit{TS}}

\newcommand{\lprmem}{\texttt{prm}}
\newcommand{\lprom}{P}
\newcommand{\myfuture}{{fut}}
\newcommand{\ogsco}{\txtsub{\gsco}{\myfuture}}
\newcommand{\omem}{\txtsub{\mem}{\myfuture}}

\newcommand{\simfull}{{\mathcal{J}}}
\newcommand{\simthread}{{\mathcal{I}}}
\newcommand{\msghelper}{{\rm view}}
\newcommand{\astep}[1]{\myrightarrow{#1}}

\newcommand{\D}{{D}}
\newcommand{\ee}[1]{e_{#1}}
\newcommand{\PConf}{\Sigma}

\newcommand{\futuresimmem}{{\rm up}\text{-}{\rm mem}}
\newcommand{\fsim}{\mathcal{L}}

\newcommand{\mems}{M_{\rm S}}

\newcommand{\gscos}{\gsco_{\rm S}}
\newcommand{\memt}{M_{\rm T}}

\newcommand{\gscot}{\gsco_{\rm T}}

\newcommand{\viewCurS}{\viewCur_{\rm S}}
\newcommand{\viewAcqS}{\viewAcq_{\rm S}}
\newcommand{\viewRelS}{\viewRel_{\rm S}}
\newcommand{\viewCurT}{\viewCur_{\rm T}}
\newcommand{\viewAcqT}{\viewAcq_{\rm T}}
\newcommand{\viewRelT}{\viewRel_{\rm T}}

\newcommand{\sigmaS}{\sigma_{\rm S}}
\newcommand{\sigmaT}{\sigma_{\rm T}}
\newcommand{\lpromS}{\lprom_{\rm S}}
\newcommand{\lpromT}{\lprom_{\rm T}}

\newcommand{\addG}[1]{\mathsf{add}_{#1}}
\newcommand{\nextG}[1]{\mathsf{next}_{#1}}
\newcommand{\lG}{\mathtt{G}}
\newcommand{\simcertification}{{\rm crt}}

\newcommand{\TCcert}{\TC^{\simcertification}}

\newcommand{\Gcert}{G^{\simcertification}}
\newcommand{\CoveredSetcert}{\CoveredSet^{\simcertification}}
\newcommand{\IssuedSetcert}{\IssuedSet^{\simcertification}}

\newcommand{\simthreadcert}{\mathcal{I}^{\simcertification}}
\newcommand{\TC}{TC}
\newcommand\doubleplus{+\kern-1.3ex+\kern0.8ex}

\begin{document}

\title{Bridging the Gap between Programming Languages and Hardware Weak Memory Models}



\author{Anton Podkopaev}
\affiliation{
  \institution{St. Petersburg University}            
  \city{St. Petersburg}
  \state{Russia}
}
\affiliation{
  \institution{JetBrains Research}            
  \city{St. Petersburg}
  \country{Russia}
}
\affiliation{
  \institution{MPI-SWS}            
  \country{Germany}
}
\email{anton.podkopaev@jetbrains.com}          

\author{Ori Lahav}
\affiliation{
  \institution{Tel Aviv University}           
  \country{Israel}
}
\email{orilahav@tau.ac.il}         

\author{Viktor Vafeiadis}
\orcid{0000-0001-8436-0334}             
\affiliation{
  \institution{MPI-SWS}           
  \streetaddress{Saarland Informatics Campus}
  \country{Germany}
}
\email{viktor@mpi-sws.org}         


\begin{abstract}
We develop a new intermediate weak memory model, \IMM, as a way of modularizing the
proofs of correctness of compilation from concurrent programming languages
with weak memory consistency semantics to mainstream multi-core architectures,
such as POWER and ARM.
We use \IMM to prove the correctness of compilation from the
promising semantics of Kang et al.\ to POWER (thereby correcting and improving their result) and ARMv7,
as well as to the recently revised ARMv8 model.
Our results are mechanized in Coq, and to the best of our knowledge, these are the first machine-verified
compilation correctness results for models that are weaker than x86-TSO.

\end{abstract}

\ccsdesc[500]{Theory of computation~Concurrency}
\ccsdesc[500]{Software and its engineering~Semantics}
\ccsdesc[500]{Software and its engineering~Compilers}
\ccsdesc[500]{Software and its engineering~Correctness}

\keywords{Weak memory consistency, IMM, promising semantics, C11 memory model}

\maketitle

\section{Introduction}
\label{sec:intro}

To support platform-independent concurrent programming, 
languages like C/C++11 and Java9 provide several types of memory accesses
and high-level fence commands.
Compilers of these languages are required to map the high-level primitives to instructions of mainstream architectures: 
in particular, x86-TSO~\cite{x86-tso}, ARMv7 and POWER~\cite{herding-cats}, and ARMv8~\cite{Pulte-al:POPL18}.
In this paper, we focus on proving the correctness of such mappings.
Correctness amounts to showing that for every source program $P$, 
the set of behaviors allowed by the target architecture for the mapped program $\compile{P}$
(the program obtained by pointwise mapping the instructions in $P$)
is contained in the set of behaviors allowed by the language-level model for $P$.
Establishing such claim is a major part of a compiler correctness proof,
and it is required for demonstrating the implementability of concurrency semantics.%
\footnote{In the rest of this paper
we refer to these mappings as ``compilation'',
leaving compiler optimizations out of our scope.}

Accordingly, it has been an active research topic.
In the case of C/C++11, \citet{Batty-al:POPL11} established the correctness of a mapping to x86-TSO, 
while \citet{Batty-al:POPL12} addressed the mapping to POWER and ARMv7.
However, the correctness claims of \citet{Batty-al:POPL12} were subsequently found to be incorrect \cite{Manerkar-al:CoRR16,scfix},
as they mishandled the combination of sequentially consistent accesses with weaker accesses.
\citet{scfix} developed RC11, a repaired version of C/C++11, and established (by pen-and-paper proof) 
the correctness of the suggested compilation schemes to x86-TSO, POWER and ARMv7.
Beyond (R)C11, however, there are a number of other proposed higher-level semantics, 
such as JMM~\cite{Manson-al:POPL05},
OCaml~\cite{OCAMLraces},
Promise~\cite{Kang-al:POPL17}, 
LLVM~\cite{Chakraborty-al:CGO17},
Linux kernel memory model~\cite{LK},
AE-justification~\cite{Jeffrey-Riely:LICS16},
Bubbly~\cite{Pichon-al:POPL16},
and WeakestMO~\cite{Chakraborty-Vafeiadis:POPL19},
for which only a handful of compilation correctness results have been developed.

As witnessed by a number of known incorrect claims and proofs,
these correctness results may be very difficult to establish.
The difficulty stems from the typical large gap
between the high-level programming language concurrency features and semantics, and the architecture ones.
In addition, since hardware models differ in their strength (\eg which dependencies are preserved)
and the primitives they support (barriers and atomic accesses), each hardware model may require a new challenging proof.


To address this problem, 
we propose to modularize the compilation correctness proof to go via an intermediate model, which we call \IMM (for Intermediate Memory Model).
\IMM contains features akin to a language-level model (such as relaxed and release/acquire accesses as well as compare-and-swap primitives), 
but gives them a hardware-style declarative (\aka axiomatic)  semantics referring to explicit syntactic dependencies.%
\footnote{Being defined on a per-execution basis, \IMM is not suitable as language-level semantics (see~\cite{Batty-al:ESOP15}).
Indeed, it disallows various compiler optimizations that remove syntactic dependencies.}
 \IMM is very useful for structuring the compilation proofs and for enabling proof reuse:
for $N$ language semantics and $M$ architectures,
using \IMM, we can reduce the number of required results from $N\times M$ to $N+M$,
and moreover each of these $N+M$ proofs is typically easier than a corresponding end-to-end proof
because of a smaller semantic gap between \IMM and another model than between a given language-level and hardware-level model.
The formal definition of \IMM contains a number of subtle points
as it has to be weaker than existing hardware models, 
and yet strong enough to support compilation from language-level models.
(We discuss these points in \cref{sec:model}.)

\begin{wrapfigure}{r}{0.38\textwidth}\centering
\begin{tikzpicture}[auto,nd/.style={inner sep=2pt},lab/.style={inner sep=0pt}]
\node[nd] (ph) {\IMM};
\node[nd,right=8mm of ph] (arm7) {ARMv7};
\node[nd,above right=2.5mm and 8mm of ph] (pow)  {POWER};
\node[nd,above=2.5mm of pow ] (tso)  {x86-TSO};
\node[nd,below=2.5mm of arm7] (arm8) {ARMv8};
\node[nd,below=2.5mm of arm8] (riscv) {RISC-V};
\node[nd,above left=2.5mm and 8mm of ph] (prom) {Promise};
\node[nd,below left=2.5mm and 8mm of ph] (rc11) {(R)C11$^*$};
\draw[->,out=  0,in=170] (prom) edge (ph);
\draw[->,out=  0,in=190] (rc11) edge (ph);
\draw[->,out= 60,in=180] (ph) edge (tso);
\draw[->,out= 10,in=180] (ph) edge  (pow);
\draw[->] (ph) edge (arm7);
\draw[->,out=-10,in=180] (ph) edge (arm8);
\draw[->,out=-60,in=180] (ph) edge (riscv);
\end{tikzpicture}
\caption{Results proved in this paper.}
\label{fig:res}
\end{wrapfigure}
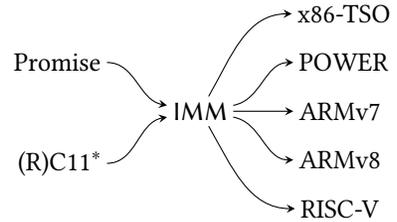

As summarized in \cref{fig:res},
besides introducing \IMM and proving that it is a sound abstraction over a range of hardware memory models,
we prove the correctness of compilation from fragments of C11 and RC11 without non-atomic and SC accesses (denoted by (R)C11$^*$)
and from the language-level memory model of the ``promising semantics'' of \citet{Kang-al:POPL17} to \IMM.

The latter proof is the most challenging.
The promising semantics is a recent prominent attempt to solve the infamous ``out-of-thin-air'' problem in 
programming language concurrency semantics~\cite{Batty-al:ESOP15,Boehm14} without sacrificing performance.
To allow efficient implementation on modern hardware platforms,
the promising semantics allows threads to execute instructions out of order
by having them ``\emph{promise}'' (\ie pre-execute) future stores.
To avoid out-of-thin-air values, every step in the promising semantics is subject to a \emph{certification} condition.
Roughly speaking, this means that thread $\tid$ may take a step to a state $\sigma$, 
only if there exists a sequence of steps of thread $\tid$ starting from $\sigma$
to a state $\sigma'$ in which $\tid$ indeed performed (fulfilled) all its pre-executed writes (promises). 
Thus, the validity of a certain trace in the promising semantics depends on existence of other traces.

In mapping the promising semantics to \IMM, we therefore have the largest gap to bridge:
a non-standard operational semantics on the one side
versus a hardware-like declarative semantics on the other side.
To relate the two semantics, we carefully construct a traversal strategy on \IMM execution graphs,
which gives us the order in which we can execute the promising semantics machine, keep satisfying its certification condition,
and finally arrive at the same outcome.

The end-to-end result is the correctness of an efficient mapping
from the promising semantics of \citet{Kang-al:POPL17} to the main hardware architectures.
While there are two prior compilation correctness results from promising semantics to POWER and ARMv8~\cite{Podkopaev-al:ECOOP17,Kang-al:POPL17},
neither result is adequate.
The POWER result~\cite{Kang-al:POPL17} considered a simplified (suboptimal) compilation scheme
and, in fact, we found out that its proof is \emph{incorrect} in its handling of SC fences (see \cref{sec:related} for more details).
In addition, its proof strategy, which is based on program transformations account for weak behaviors~\cite{trns}, cannot be applied to ARM.
The ARMv8 result~\cite{Podkopaev-al:ECOOP17} handled only a small restricted subset of the concurrency features of the promising semantics
and an operational hardware model (ARMv8-POP) that was later abandoned by ARM in favor of a rather different declarative model~\cite{Pulte-al:POPL18}.

By encompassing all features of the promising semantics,
our proof uncovered a subtle correctness problem in the conjectured compilation scheme of its read-modify-write (RMW) operations to ARMv8
and to the closely related RISC-V model.
We found out that exclusive load and store operations in ARMv8 and RISC-V are weaker than those of POWER and ARMv7, 
following their models by \citet{herding-cats}, so that the intended compilation of RMWs is broken (see \cref{ex:strong-rmw}).
Thus, the mapping to ARMv8 that we proved correct places a weak barrier (specifically ARM's ``ld fence'') after every RMW.%
\footnote{Recall that RMWs are relatively rare.
The performance cost of this fixed compilation scheme is beyond the scope of this paper,
and so is the improvement of the promising semantics to recover the correctness of the barrier-free compilation.}
To keep \IMM as a sound abstraction of ARMv8 and allow reuse of \IMM in a future improvement of the promising semantics, 
we equip \IMM with two types of RMWs: usual ones that are compiled to ARMv8 without the extra barrier,
and stronger ones that require the extra barrier.
To establish the correctness of the mapping from the (existing) promising semantics to \IMM,
we require that RMW instructions of the promising semantics are mapped to \IMM's strong RMWs.

Finally,  to ensure correctness of such subtle proofs, our results are all mechanized in Coq ($\sim$33K LOC).
To the best of our knowledge, this constitutes the first mechanized correctness of compilation result from a high-level programming language
concurrency model to a model weaker than x86-TSO.
We believe that the existence of Coq proof scripts relating the different models may facilitate the development and investigation
of weak memory models in the future, as well as the possible modifications of \IMM to accommodate new and revised 
hardware and/or programming languages concurrency semantics.

The rest of this paper is organized as follows.
In \cref{sec:pre} we present \IMM's program syntax and its mapping to execution graphs.
In \cref{sec:model} we define \IMM's consistency predicate.
In \cref{sec:hardware} we present the mapping of \IMM to main hardware and establish its correctness.
In \cref{sec:c11} we present the mappings from C11 and RC11 to \IMM and establish their correctness.
\Cref{sec:rlx-prom-compilation,sec:prom-compilation} concern the mapping of 
the promising semantics of \citet{Kang-al:POPL17} to \IMM.
To assist the reader, we discuss first (\cref{sec:rlx-prom-compilation}) a restricted fragment (with only relaxed accesses),
and later (\cref{sec:prom-compilation}) extend our results and proof outline to the full promising model.
Finally, we discuss related work in \cref{sec:related} and conclude in \cref{sec:conclusion}.

\smallskip
Supplementary material for this paper, including the Coq development,  
is publicly available at \url{http://plv.mpi-sws.org/imm/}.



\section{Preliminaries: from programs to execution graphs}
\label{sec:pre}

Following the standard declarative (\aka axiomatic) approach of defining memory consistency models~\cite{herding-cats},
the semantics of \IMM programs is given in terms of \emph{execution graphs} which partially
order \emph{events}.
This is done in two steps. First, the program is mapped to a large set of execution graphs in
which the read values are completely arbitrary.
Then, this set is filtered by a consistency predicate, and only \emph{\IMM-consistent}
execution graphs determine the possible outcomes of the program under \IMM.
Next, we define \IMM's programming language (\cref{sec:language}),
define \IMM's execution graphs (\cref{sec:executions}),
and present the construction of execution graphs from programs (\cref{sec:prog-to-exec}).
The next section (\cref{sec:model}) is devoted to present \IMM's consistency predicate. 

Before we start we introduce some notation for relations and functions.
Given a binary relation $R$,
we write $R^?$, $R^+$, and $R^*$ respectively to denote its reflexive, transitive, and reflexive-transitive closures.
The inverse relation is denoted by $R^{-1}$, and
$\dom{R}$ and $\codom{R}$ denote $R$'s domain and codomain.
We denote by $R_1\mathbin{;}R_2$ the left composition of two relations $R_1,R_2$,
and assume that $;$ binds tighter than $\cup$ and $\setminus$.
We write $\imm{R}$ for the set of all \emph{immediate $R$ edges}: $\imm{R}\defeq R \setminus R\mathbin{;}R$.
We denote by $[A]$ the identity relation on a set $A$.
In particular, $[A]\mathbin{;}R\mathbin{;}[B] = R\cap (A\times B)$.
For finite sets $\set{a_1\til a_n}$, we omit the set parentheses and write $[a_1 \til a_n]$.
Finally, for a function $f: A \to B$ and a set $X\suq A$, we write $f[X]$ to denote the set $\set{f(x) \st x\in X}$.

\subsection{Programming language}
\label{sec:language}

\begin{figure}[t]
\centering
$\inarr{
\begin{array}[t]{@{} l @{\hspace{20pt}} l @{}}
	\begin{array}{@{} l r @{}}
		\textbf{Domains} \\
			\quad n \in \N & \text{Natural numbers} \\
			\quad v \in \Val \defeq \N & \text{Values} \\
			\quad \loc \in \Loc  \defeq \N & \text{Locations} \\
			\quad r \in \Reg & \text{Registers}	 \\
			\quad \tid \in \Tid & \text{Thread identifiers} 
	\end{array}		&
	\begin{array}{@{} r@{\;}l r @{}}
		\textbf{Modes} \\
			\quad o_\lR ::=  & \rlx \ALT \acq  & \text{Read modes} \\
			\quad o_\lW ::=  & \rlx \ALT \rel & \text{Write modes} \\
			\quad o_\lF ::=  & \acq \ALT \rel \ALT \acqrel \ALT \sco  & \text{Fence modes} \\
			\quad o_\lU::=  & \normal \ALT \strong & \text{RMW modes} \\ 
            \mbox{}
	\end{array}	
		\end{array}	
\\[8ex]
\begin{array}{r@{\;}l}
\Exp \ni e ::=  & r \ALT n \ALT e_1 + e_2 \ALT e_1 - e_2 \ALT \ldots
\\[1ex]
\Inst \ni \mathit{inst} ::=  & 
\assignInst{r}{e}
\ALT \ifGotoInst{e}{n}
\ALT \writeInst{o_\lW}{e}{e} 
\ALT \readInst{o_\lR}{r}{e} 
\ALT \\&
\incInst{o_\lR}{o_\lW}{r}{e}{e}{o_\lU}
\ALT \casInst{o_\lR}{o_\lW}{r}{e}{e}{e}{o_\lU}
\ALT \fenceInst{o_\lF}
\end{array} 
\\[4ex]
\begin{array}{@{} l r @{}}
\sprog \in \Sprog \defeq \N \fpfn \Inst & \text{Sequential programs} \\
\prog : \Tid \to \Sprog  & \text{Programs}
\end{array}}$
\caption{Programming language syntax.}
\label{fig:lang}
\end{figure}

\IMM is formulated over the language defined in \cref{fig:lang}
with C/C++11-like concurrency features.
Expressions are constructed from registers (local variables) and integers,
and represent values and locations.
Instructions include assignments and conditional branching,
as well as memory operations.
Intuitively speaking, 
an assignment $\assignInst{r}{e}$ assigns the value of $e$ to register $r$ (involving no memory access);
$\ifGotoInst{e}{n}$ jumps to line $n$ of the program iff the value of $e$ is not $0$;
the write $\writeInst{o_\lW}{e_1}{e_2}$ stores the value of $e_2$ in the address given by $e_1$;
the read $\readInst{o_\lR}{r}{e} $ loads the value in address $e$ to register $r$;
$\incInst{o_\lR}{o_\lW}{r}{e_1}{e_2}{o_\lU}$ atomically increments the value in address $e_1$
by the value of $e_2$ and loads the old value to $r$;
$\casInst{o_\lR}{o_\lW}{r}{e}{e_\lR}{e_\lW}{o_\lU}$ atomically compares
the value stored in address $e$ to the value of $e_\lR$, and if the two values are the same,
it replaces the value stored in $e$ by the value of $e_\lW$; and fence instructions $\fenceInst{o_\lF}$
are used to place global barriers. 

The memory operations are annotated with \emph{modes} that are ordered as follows: 
$$\sqsuq ~\defeq \set{\tup{\rlx,\acq},\tup{\rlx,\rel}, \tup{\acq,\acqrel},\tup{\rel,\acqrel},\tup{\acqrel,\sco}}^*$$
Whenever $o_1 \sqsuq o_2$, we say that $o_2$ is stronger than $o_1$: 
it provides more consistency guarantees but is more costly to implement.
RMWs include two modes---$o_\lR$ for the read part and
$o_\lW$ for the write part---as well as a third (binary) mode $o_\lU$ used to denote certain RMWs
as stronger ones.

In turn, sequential programs are finite maps from $\N$ to instructions, 
and (concurrent) programs are top-level parallel composition of 
sequential programs,
defined as mappings from a finite set $\Tid$ of thread identifiers to sequential programs.
In our examples, we write sequential programs as sequences of instructions delimited by `$;$' (or line breaks) 
and use `$\parallel$' for parallel composition.

\begin{remark}
C/C++11 sequentially consistent (SC) accesses are not included in \IMM.
They can be simulated, nevertheless, using SC fences
following the compilation scheme of C/C++11 (see~\cite{scfix}).
We note that SC accesses are also not supported by the promising semantics.
\end{remark}

\subsection{Execution graphs}
\label{sec:executions}

\begin{definition}
\label{def:event}
An \emph{event}, $e\in\Event$, takes one of the following forms:
\begin{itemize}
\item Non-initialization event: $\tup{\tid,n}$ where 
$\tid\in\Tid$ is a thread identifier,
and $n\in\Q$ is a serial number inside each thread.
\item Initialization event: $\initev{\loc}$ where $\loc\in\Loc$ is the location being initialized.
\end{itemize}
We denote by $\Init$ the set of all initialization events.
The functions $\lTID$ and $\lSN$ return the (non-initialization) event's thread identifier and serial number.
\end{definition}

Our representation of events induces a \emph{sequenced-before} partial order on events given by:
$$e_1 < e_2 \Leftrightarrow (e_1 \in \Init \land e_2 \nin \Init) \lor 
(e_1 \nin \Init \land e_2 \nin \Init \land \lTID(e_1) =\lTID(e_2) \land \lSN(e_1) < \lSN(e_2))$$
Initialization events precede all non-initialization events,
while events of the same thread are ordered according to their serial numbers.
We use rational numbers as serial numbers to be able to easily add an event between any two events.

\begin{definition}
\label{def:label}
A \emph{label}, $l\in\Lab$, takes one of the following forms:
\begin{itemize}
\item Read label: $\erlab{o_\lR}{\loc}{v}{s}$ where $\loc\in\Loc$, $v\in\Val$, $o_\lR\in\set{\rlx,\acq}$, 
and $s\in\set{\isnotex,\isex}$.
\item Write label: $\ewlab{o_\lW}{\loc}{v}{o_\lU}$ where $\loc\in\Loc$, $v\,{\in}\,\Val$, $o_\lW\,{\in}\,\set{\rlx,\rel}$,
and $o_\lU\,{\in}\,\set{\normal,\strong}$.
\item Fence label: $\flab{o_\lF}$ where $o_\lF\in\set{\acq,\rel,\acqrel,\sco}$.
\end{itemize}
\end{definition}

Read labels include a location, a value, and a mode, as well as an ``is exclusive'' flag $s$.
Exclusive reads stem from an RMW and are usually followed by a corresponding write. 
An exception is the case of a ``failing'' CAS (when the read value is not the expected one), 
where the exclusive read is not followed by a corresponding write.
Write labels include a location, a value, and a mode, as well as a flag marking certain writes
as strong. This will be used to differentiate the strong RMWs from the normal ones.
Finally, a fence label includes just a mode.

\begin{definition}
\label{def:execution}
An \emph{execution} $G$ consists of:
\begin{enumerate}
\item a finite set $G.\lE$ of events.
Using $G.\lE$ and the partial order $<$ on events, 
we derive the \emph{program order} (\aka \emph{sequenced-before}) relation in $G$:
$G.\lPO \defeq [G.\lE]; < ;[G.\lE]$.
For $\tid\in\Tid$, we denote by $G.\lE\tidmod{\tid}$ the set $\set{a\in G.\lE \st \lTID(a)=\tid}$,
and by  $G.\lE\tidmod{\neq \tid}$ the set $\set{a\in G.\lE \st \lTID(a)\neq\tid}$.
\item a labeling function $G.\lLAB: G.\lE \to \Lab$.
The labeling function naturally induces functions 
$G.\lMOD$, $G.\lLOC$, and $G.\lVAL$
that return (when applicable) an event's label 
mode, location, and value.
We use $G.\lR,G.\lW,G.\lF$ to denote the subsets of $G.\lE$ of events labeled with the respective type.
We use obvious notations to further restrict the different modifiers of the event
(\eg $G.\ewlab{}{\loc}{}{}=\set{w\in G.\lW \st G.\lLOC(w)=\loc }$
and $G.\lF^{\sqsupseteq o}=\set{f\in G.\lF \st G.\lMOD(f)\sqsupseteq o}$).
We assume that $G.\lLAB(\initev{\loc})=\ewlab{\rlx}{\loc}{0}{\normal}$
for every $\initev{\loc}\in G.\lE$.
\item a relation $G.\lRMW\suq \bigcup_{\loc\in\Loc} [G.\erlab{}{\loc}{}{\isex}];\imm{G.\lPO};[G.\ewlab{}{\loc}{}{}]$, called \emph{RMW pairs}.
We require that $G.\ewlab{}{}{}{\strong} \suq \codom{G.\lRMW}$.
\item a relation $G.\lDATA\suq [G.\lR];G.\lPO;[G.\lW]$, called \emph{data dependency}.
\item a relation $G.\lADDR\suq [G.\lR];G.\lPO;[G.\lR\cup G.\lW]$, called \emph{address dependency}.
\item a relation $G.\lCTRL\suq [G.\lR];G.\lPO$, called \emph{control dependency},
that is forwards-closed under the program order: $G.\lCTRL;G.\lPO \suq G.\lCTRL$.
\item a relation $G.\lRMWDEP\suq [G.\lR];G.\lPO;[G.\erlab{}{}{}{\isex}]$, called \emph{CAS dependency}.
\item a relation $G.\lRF\suq \bigcup_{\loc\in\Loc} G.\ewlab{}{\loc}{}{} \times G.\erlab{}{\loc}{}{}$, called \emph{reads-from},
and satisfying:
$G.\lVAL(w)=G.\lVAL(r)$ for every $\tup{w,r}\in G.\lRF$; and
$w_1=w_2$ whenever $\tup{w_1,r},\tup{w_2,r}\in G.\lRF$ (that is, $G.\lRF^{-1}$ is functional).
\item a strict partial order $G.\lCO \suq \bigcup_{\loc\in\Loc} G.\ewlab{}{\loc}{}{} \times G.\ewlab{}{\loc}{}{}$, 
called \emph{coherence order} (\aka \emph{modification order}). 
\end{enumerate}
\end{definition}

\subsection{Mapping programs to executions}
\label{sec:prog-to-exec}

Sequential programs are mapped to execution graphs by means of an operational semantics.
Its states have the form $\sigma=\tup{\sprog,pc,\Phi,G,\PhiD,S}$,
where 
$\sprog$ is the thread's sequential program;
$pc\in\N$ points to the next instruction in $\sprog$ to be executed;
$\Phi:\Reg\to\Val$ maps register names to the values they store (extended to expressions in the obvious way);
$G$ is an execution graph (denoted by $\sigma.\lG$);
$\PhiD: \Reg \to \powerset{G.\lR}$ maps each register name to the set of events that were used to compute the register's value;
and $S \suq G.\lR$ maintains the set of events having a control dependency to the current program point.
The $\PhiD$ and $S$ components are used to calculate the dependency edges in $G$.
$\PhiD$ is extended to expressions in the obvious way (\eg
 $\PhiD(n) \defeq \emptyset$ and $\PhiD(e_1 + e_2) \defeq \PhiD(e_1) \cup \PhiD(e_2)$).
Note that the executions graphs produced by this semantics represent traces of one thread,
and as such, they are quite degenerate: $G.\lPO$ totally orders $G.\lE$ and $G.\lRF = G.\lCO=\emptyset$.

\begin{figure}[t]
\small\centering$
\begin{array}{|c|l|}\hline
\text{When } \sprog(pc)=\ldots   & \multicolumn{1}{c|}{\text{we have the following constraints relating $pc,pc',\Phi,\Phi',G,G',\PhiD,\PhiD',S,S'$:}}
\\[1pt]\hline
\assignInst{r}{e} & \inarr{
  pc'=pc+1 \land \Phi' = \Phi[r:=\Phi(e)] \land G'=G \land \PhiD' = \PhiD[r:=\PhiD(e)] \land
  S'=S 
  }
\\[1pt]\hline
\ifGotoInst{e}{n} & \inarr{ 
(\Phi(e)\neq 0 \implies pc' =n) \land (\Phi(e)= 0 \implies pc' =pc +1) \land {}\\
G=G' \land \Phi=\Phi' \land \PhiD' = \PhiD \land S'=S \cup \PhiD(e) 
}
\\[1pt]\hline
\writeInst{o_\lW}{e_1}{e_2}  & \inarr{ 
  G'=\addG{G}(\tid,\ewlab{o_\lW}{\Phi(e_1)}{\Phi(e_2)}{\normal},\emptyset,\PhiD(e_2),\PhiD(e_1),S,\emptyset)
\land{}\\
pc'=pc+1 \land \Phi'=\Phi \land \PhiD' = \PhiD \land S'=S 
}
\\[1pt]\hline
\readInst{o_\lR}{r}{e} & \inarr{ \exists v.~ 
  G'=\addG{G}(\tid,\erlab{o_\lR}{\Phi(e)}{v}{\isnotex},\emptyset,\emptyset,\PhiD(e),S,\emptyset)
\land{}\\
pc'=pc+1 \land \Phi' =\Phi[r:=v] \land \PhiD' = \PhiD[r:=\set{\tup{\tid,\nextG{G}}}] \land S'=S 
}
\\[1pt]\hline
\incInst{o_\lR}{o_\lW}{r}{e_1}{e_2}{o_\lU} & \inarr{
  \exists v.~
  \letdef{a_\lR,G_\lR}{\tup{\tid,\nextG{G}},
    \addG{G}(\tid,\erlab{o_\lR}{\Phi(e_1)}{v}{\isex},\emptyset,\emptyset,\PhiD(e_1),S,\emptyset)} \\
  G'=\addG{G_\lR}(\tid,\ewlab{o_\lW}{\Phi(e_1)}{v+\Phi(e_2)}{o_\lU},\set{a_\lR},\set{a_\lR} \cup \PhiD(e_2),\PhiD(e_1),S,\emptyset) 
\land{}\\
pc'=pc+1  \land \Phi' =\Phi[r:=v] \land \PhiD' = \PhiD[r:=\set{a_\lR}] \land S'=S 
}
\\[1pt]\hline
\casInst{o_\lR}{o_\lW}{r}{e}{e_\lR}{e_\lW}{o_\lU} & \inarr{ \exists v.~
  \letdef{a_\lR, G_\lR}{
    \tup{\tid,\nextG{G}},
    \addG{G}(\tid,\erlab{o_\lR}{\Phi(e)}{v}{\isex},\emptyset,\emptyset,\PhiD(e),S,\PhiD(e_\lR)) 
  } \\
  pc'=pc+1 \land \Phi' =\Phi[r:=v] \land \PhiD' = \PhiD[r:=\set{a_\lR}] \land S'=S \land {} \\
  (v \neq \Phi(e_\lR) \implies G' = G_\lR)  \land {} \\
  (v = \Phi(e_\lR) \implies 
    G'=\addG{G_\lR}(\tid,\ewlab{o_\lW}{\Phi(e)}{\Phi(e_\lW)}{o_\lU},\set{a_\lR},\PhiD(e_\lW),\PhiD(e),S,\emptyset) )
}
\\[1pt]\hline
\fenceInst{o_\lF}   & \inarr{
  G'=\addG{G}(\tid,\flab{o_\lF},\emptyset,\emptyset,\emptyset,S,\emptyset) \land
 pc'=pc+1 \land \Phi' =\Phi \land \PhiD'=\PhiD \land
S'=S 
}
\\[1pt]\hline
\end{array}$
\vspace{-6pt}
\caption{The relation $\tup{\sprog,pc,\Phi,G,\PhiD,S} \rightarrow_{i} \tup{\sprog,pc',\Phi',G',\PhiD',S'}$
representing a step of thread $i$.}
\label{fig:prog-to-exec}
\end{figure}

The \emph{initial state} is $\sigma_0(\sprog)\defeq \tup{\sprog,0,\lambda r.\;0,G_\emptyset,\lambda r.\;\emptyset,\emptyset}$
($G_\emptyset$  denotes the empty execution),
\emph{terminal} states are those in which $pc \nin \dom{\sprog}$,
and the transition relation is given in \cref{fig:prog-to-exec}.
It uses the notations $\nextG{G}$ to obtain the next serial number in a thread execution graph $G$
($\nextG{G} \defeq \size{G.\lE}$)
and $\addG{G}$ to append an event with thread identifier $\tid$ and label $l$ to $G$:

\begin{definition}
For an execution graph $G$, $\tid\in\Tid$, $l\in\Lab$, and 
$E_\lRMW,E_\lDATA,E_\lADDR,E_\lCTRL,E_\lRMWDEP\suq G.\lR$,
$\addG{G}(\tid,l,E_\lRMW,E_\lDATA,E_\lADDR,E_\lCTRL,E_\lRMWDEP)$
denotes the execution graph $G'$ given by:
$$
\begin{array}{@{}r@{\;}l@{\quad}@{}r@{\;}l@{}}
G'.\lE &= G.\lE \uplus \set{\tup{\tid,\nextG{G}}}  &
G'.\lLAB &= G.\lLAB \uplus \set{\tup{\tid,\nextG{G}}\mapsto l} \\
G'.\lRMW &= G.\lRMW \uplus (E_\lRMW \times \set{\tup{\tid,\nextG{G}}}) &
G'.\lDATA &= G.\lDATA \uplus (E_\lDATA \times \set{\tup{\tid,\nextG{G}}}) \\
G'.\lADDR &= G.\lADDR \uplus (E_\lADDR \times \set{\tup{\tid,\nextG{G}}}) &
G'.\lCTRL &= G.\lCTRL \uplus (E_\lCTRL \times \set{\tup{\tid,\nextG{G}}}) \\
G'.\lRMWDEP &= G.\lRMWDEP \uplus (E_\lRMWDEP \times \set{\tup{\tid,\nextG{G}}}) &
G'.\lRF &= G.\lRF \qquad
G'.\lCO = G.\lCO 
\end{array}$$%
\end{definition}


Besides the explicit calculation of dependencies, the operational semantics is standard.

\begin{example}
The only novel ingredient is the \emph{CAS dependency} relation, which tracks 
reads that affect the success of a CAS instruction.
As an example, consider the following program.

{\makeatletter
\let\par\@@par
\par\parshape0
\everypar{}\begin{wrapfigure}{r}{0.58\textwidth}\centering
\small$
\inarr{
  \readInst{\rlx}{a}{x} \\
  \casInst{\rlx}{\rlx}{b}{y}{a}{1}{\normal} \\
  \writeInst{\rlx}{z}{2} \\
}
\;\;\vrule\;\;
\inarr{\begin{tikzpicture}[yscale=0.8,xscale=1]
  \node (i11) at (0,  0) {$\erlab{\rlx}{x}{0}{\isnotex}$};
  \node (i12) at (0, -1) {$\erlab{\rlx}{y}{0}{\isex}$};
  \node (i13) at (0, -2) {$\ewlab{\rlx}{y}{1}{\normal}$};
  \node (i14) at (0, -3) {$\ewlab{\rlx}{z}{2}{}$};
  \draw[po] (i11) edge node[left] {\smaller$\lPO$} (i12);
  \draw[po] (i12) edge node[left] {\smaller$\lPO$} (i13);
  \draw[po] (i13) edge node[left] {\smaller$\lPO$} (i14);
  \draw[deps,bend left=15] (i11) edge node[right] {\smaller$\lRMWDEP$} (i12);
  \draw[rmw,bend left=15] (i12) edge node[right] {\smaller$\lRMW$} (i13);
\end{tikzpicture}}
\;\;\vrule\;\;
\inarr{\begin{tikzpicture}[yscale=0.8,xscale=1]
  \node (i11) at (0,  0) {$\erlab{\rlx}{x}{1}{\isnotex}$};
  \node (i12) at (0, -1) {$\erlab{\rlx}{y}{0}{\isex}$};
  \node (i14) at (0, -3) {$\ewlab{\rlx}{z}{2}{}$};
  \draw[po] (i11) edge node[left] {\smaller$\lPO$} (i12);
  \draw[deps,bend left=15] (i11) edge node[right] {\smaller$\lRMWDEP$} (i12);
  \draw[po] (i12) edge node[left] {\smaller$\lPO$} (i14);
\end{tikzpicture}}$
 \end{wrapfigure}

\noindent
The CAS instruction may produce a write event or not, 
depending on the value read from $y$ and the value of register $a$, which is assigned at the read instruction from $x$.
The $\lRMWDEP$ edge reflects the latter dependency in both representative execution graphs.
The mapping of \IMM's CAS instructions to POWER and ARM ensures that the $\lRMWDEP$ on the source execution graph
implies a control dependency to all $\lPO$-later events in the target graph (see \cref{sec:hardware}).
\qed
  \par}
\end{example}




Next, we define \emph{program executions}.

\begin{definition}
For an execution graph $G$ and $\tid\in\Tid$,
$G\rst{\tid}$ denotes the execution graph given by:
$$
\begin{array}{@{}r@{\;}l@{}@{}r@{\;}l@{}}
G\rst{\tid}.\lE &= G.\lE\tidmod{\tid}  &
G\rst{\tid}.\lLAB &= G.\lLAB\rst{G.\lE\tidmod{\tid}} \\
G\rst{\tid}.\lRMW &= [G.\lE\tidmod{\tid}] ; G.\lRMW ; [G.\lE\tidmod{\tid}] &
G\rst{\tid}.\lDATA &= [G.\lE\tidmod{\tid}] ;  G.\lDATA ; [G.\lE\tidmod{\tid}] \\
G\rst{\tid}.\lADDR &= [G.\lE\tidmod{\tid}] ;  G.\lADDR ; [G.\lE\tidmod{\tid}] &
G\rst{\tid}.\lCTRL &= [G.\lE\tidmod{\tid}] ;  G.\lCTRL ; [G.\lE\tidmod{\tid}]  \\
G\rst{\tid}.\lRMWDEP &= [G.\lE\tidmod{\tid}] ;  G.\lRMWDEP ; [G.\lE\tidmod{\tid}] \qquad &
G\rst{\tid}.\lRF & = G\rst{\tid}.\lCO = \emptyset
\end{array}$$%
\end{definition}

\begin{definition}[Program executions]
An execution graph $G$ is a \emph{(full) execution graph of a program $\prog$} if
 for every $\tid\in\Tid$, there exists a (terminal) state $\sigma$ such that
$\sigma.\lG= G\rst{\tid}$ and $\sigma_0(\prog(\tid)) \to_{\tid}^* \sigma$.
\end{definition}

Now, given the \IMM-consistency predicate presented in the next section,
we define the set of allowed outcomes.

\begin{definition}
\label{def:initialized}
$G$ is \emph{initialized} if
$\initev{\loc} \in G.\lE$ 
for every $\loc\in G.\lLOC[G.\lE]$.
\end{definition}

\begin{definition}
\label{def:outcome}
\label{def:allowed}
A function $O:\Loc \to \Val$ is:
\begin{itemize}
\item an \emph{outcome of an execution graph} $G$
if for every $\loc\in\Loc$, either $O(\loc)=G.\lVAL(w)$ for some $G.\lCO$-maximal event $w\in G.\ewlab{}{\loc}{}{}$,
or $O(\loc)=0$ and $G.\ewlab{}{\loc}{}{}=\emptyset$.
\item an \emph{outcome of a program} $\prog$ under \IMM
if $O$ is an outcome of some \IMM-consistent initialized full execution graph of $\prog$.
\end{itemize} 
\end{definition}

\section{\IMM: The intermediate model}
\label{sec:model}

In this section, we introduce the consistency predicate of \IMM.
The first (standard) conditions require that every read reads from some write
($\codom{G.\lRF} = G.\lR$), 
and that the coherence order totally orders the writes to each location
($G.\lCO$ totally orders $G.\ewlab{}{\loc}{}{}$ for every $\loc\in\Loc$).
In addition, we require (1) coherence, 
(2) atomicity of RMWs, and 
(3) global ordering,
which are formulated in the rest of this section, 
with the help of several derived relations on events.

The rest of this section is described in the context of a given execution graph $G$,
and the `$G.$' prefix is omitted.
In addition, we employ the following notational conventions:
for every relation $\lX \suq \lE \times \lE$, 
we denote by $\lmakeE{\lX}$ its thread external restriction ($\lmakeE{\lX}\defeq\lX \setminus \lPO$),
while $\lmakeI{\lX}$ denotes its thread internal restriction ($\lmakeI{\lX} \defeq\lX \cap \lPO$).
We denote by $\lX\rst\lLOC$ its restriction to accesses to the same location 
($\lX\rst \lLOC \defeq \textstyle\bigcup_{\loc\in \Loc} [\erlab{}{\loc}{}{} \cup \ewlab{}{\loc}{}{}]\mathbin{;} \lX \mathbin{;} [\erlab{}{\loc}{}{} \cup \ewlab{}{\loc}{}{}]$).

\subsection{Coherence}
\label{sec:model-coherence}

\emph{Coherence} is a basic property of memory models that implies that 
programs with only one shared location behave as if they were running under sequential consistency.
Hardware memory models typically enforce coherence by requiring that
$\lPO\rst\lLOC \cup \lRF \cup \lCO \cup \lRF^{-1}\mathbin{;}\lCO$ is acyclic
(\aka \emph{SC-per-location}).
Language models, however, strengthen the coherence requirement by replacing $\lPO$ with 
a ``happens before'' relation $\lHB$ that includes $\lPO$ as well as inter-thread synchronization.
Since \IMM's purpose is to verify the implementability of language-level models,
we take its coherence axiom to be close to those of language-level models.
Following \cite{scfix}, we therefore define the following relations:
\begin{align*}
 \lRS &\defeq [\lW]\mathbin{;}\lPO\rst\lLOC \mathbin{;}[\lW] \cup [\lW]\mathbin{;}(\lPO\rst\lLOC^? \mathbin{;} \lRF\mathbin{;}\lRMW)^*
 \tag{\emph{release sequence}} \\
 \lRELEASE &\defeq ([\ewlab{\rel}{}{}{}] \cup [\flab{\sqsupseteq \rel}]\mathbin{;}\lPO) \mathbin{;}\lRS
 \tag{\emph{release prefix}} \\
 \lSW & \defeq  \lRELEASE\mathbin{;} (\lRFI \cup \lPO\rst\lLOC^?\mathbin{;}\lRFE)\mathbin{;} ([\erlab{\acq}{}{}{}] \cup \lPO\mathbin{;}[\flab{\sqsupseteq\acq}])
 \tag{\emph{synchronizes with}} \\
 \lHB &\defeq (\lPO \cup \lSW)^+ \tag{\emph{happens-before}}  \\
 \lFR &\defeq \lRF^{-1}\mathbin{;}\lCO \tag{\emph{from-read/read-before}}  \\
 \lECO &\defeq \lRF \cup \lCO\mathbin{;}\lRF^? \cup \lFR\mathbin{;}\lRF^? \tag{\emph{extended coherence order}}
\end{align*}
We say that $G$ is coherent if $\lHB\mathbin{;}\lECO^?$ is irreflexive,
or equivalently $\lHB\rst\lLOC \cup \lRF \cup \lCO \cup \lFR$ is acyclic.

\begin{example}[Message passing]
Coherence disallows the weak behavior of the MP litmus test:
\[
\inarrII{
  \writeInst{\rlx}{x}{1} \\
  \writeInst{\rel}{y}{1} 
}{
  \readInst{\acq}{a}{y} \comment{1} \\
  \readInst{\rlx}{b}{x} \comment{0} \\
}
\quad\vrule\quad
\inarr{\begin{tikzpicture}[yscale=0.8,xscale=1]
  \node (i11) at (0,  0) {$\ewlab{\rlx}{x}{1}{}$};
  \node (i12) at (0, -1) {$\ewlab{\rel}{y}{1}{}$};
  \node (i21) at (3,  0) {$\erlab{\acq}{y}{1}{\isnotex}$};
  \node (i22) at (3, -1) {$\erlab{\rlx}{x}{0}{\isnotex}$};

  \draw[po] (i11) edge (i12);
  \draw[po] (i21) edge (i22);
  \draw[rf] (i12) edge node[below] {\small$\lRF$} (i21);
  \draw[fr] (i22) edge node[above] {\small$\lFR$} (i11);
\end{tikzpicture}}
\]
To the right, we present the execution yielding the annotated weak outcome.%
\footnote{
We use program comments notation to refer to the read values in the behavior we discuss.
These can be formally expressed as program outcomes (\cref{def:outcome}) by storing the
read values in distinguished memory locations.
In addition, for conciseness, we do not show the implicit initialization events
and the $\lRF$ and $\lCO$ edges from them,
and include the $o_\lU$ subscript only for writes in $\codom{G.\lRMW}$
(recall that $G.\ewlab{}{}{}{\strong} \suq \codom{G.\lRMW}$).}
The $\lRF$-edges and the induced $\lFR$-edge are determined by the annotated outcome.
The displayed execution is inconsistent
because the $\lRF$-edge between the release write and the acquire read constitutes an $\lSW$-edge,
and hence there is an $\lHB \mathbin{;} \lFR$ cycle.
\qed
\end{example}

\begin{remark}
\label{rem:rs}
Adept readers may notice that our definition of $\lSW$ is stronger (namely, our $\lSW$ is larger)
than the one of RC11~\cite{scfix}, which (following the fixes of \citet{c11comp} to C/C++11's original definition)
employs the following definitions:
\begin{align*}
 \lRSs &\defeq [\lW]\mathbin{;}\lPO\rst\lLOC^?\mathbin{;}(\lRF\mathbin{;}\lRMW)^* &
 \lRELEASEs &\defeq ([\ewlab{\rel}{}{}{}] \cup [\flab{\sqsupseteq \rel}]\mathbin{;}\lPO) \mathbin{;}\lRSs \\
 \lSWs &\defeq  \lRELEASE\mathbin{;} \lRF\mathbin{;} ([\erlab{\acq}{}{}{}] \cup \lPO\mathbin{;}[\flab{\sqsupseteq\acq}]) &
  \lHBs &\defeq  (\lPO \cup \lSWs)^+ 
\end{align*}
The reason for this discrepancy is our aim to allow the splitting of release writes and RMWs
into release fences followed by relaxed operations. Indeed, as explained in \cref{sec:power},
the soundness of this transformation allows us to simplify our proofs.
In RC11~\cite{scfix}, as well as in C/C++11~\cite{Batty-al:POPL11},
this rather intuitive transformation, as we found out, is actually \emph{unsound}.
To see this consider the following example:
$$\inarrIII{
\writeInst{\rlx}{y}{1} \\
\writeInst{\rel}{x}{1} 
}{
\incInst{\acq}{\rel}{a}{x}{1}{} \comment{1} \\
\writeInst{\rlx}{x}{3} 
}{
\readInst{\acq}{b}{x} \comment{3} \\
\readInst{\rlx}{c}{y}  \comment{0}
}$$
(R)C11 disallows the annotated behavior, due in particular to the release sequence formed
from the release exclusive write to $x$ in the second thread to its subsequent relaxed write.
However, if we split the increment to 
$\fenceInst{\rel} ; \incInst{\acq}{\rlx}{a}{x}{1}{}$ (which intuitively may seem stronger),
the release sequence will no longer exist, and the annotated behavior will be allowed.
\IMM overcomes this problem by strengthening $\lSW$ in a way that ensures a synchronization edge
for the transformed program as well.
In \cref{sec:power}, we establish the soundness of this splitting transformation in general.
In addition, note that, as we show in \cref{sec:hardware}, existing hardware support \IMM's
stronger synchronization without strengthening the intended compilation schemes.
On the other hand, in our proof concerning the promising semantics in \cref{sec:prom-compilation}, it is more convenient to use
RC11's definition of $\lSW$, which results in a (provably) stronger (namely, allowing less behaviors) model
that still accounts for all the behaviors of the promising semantics.%
\footnote{The C++ committee is currently revising the release sequence definition
aiming to simplify it and relate it to its actual uses. 
The analysis here may provide further input to that discussion.}
\end{remark}

\subsection{RMW atomicity}
\label{sec:model-atomicity}

\emph{Atomicity of RMWs} simply states that the load of a successful RMW reads from the 
immediate $\lCO$-preceding write before the RMW's store.
Formally, $\lRMW \cap (\lFRE\mathbin{;}\lCOE) = \emptyset$, which says that there is no other
write ordered between the load and the store of an RMW.

\begin{example}[Violation of RMW atomicity]
The following behavior violates the fetch-and-add atomicity 
and is disallowed by all known weak memory models.
\[
\inarrII{
  \incInst{\rlx}{\rlx}{a}{x}{1}{\normal} \comment{0}
}{
  \writeInst{\rlx}{x}{2}  \\
  \readInst{\rlx}{b}{x} \comment{1} \\
}
\quad\vrule\quad
\inarr{\begin{tikzpicture}[yscale=0.8,xscale=1]
  \node (i11) at (0,  0) {$\erlab{\rlx}{x}{0}{\isex}$};
  \node (i12) at (0, -1) {$\ewlab{\rlx}{x}{1}{\normal}$};
  \node (i21) at (3,  0) {$\ewlab{\rlx}{x}{2}{}$};
  \node (i22) at (3, -1) {$\erlab{\rlx}{x}{1}{\isnotex}$};

  \draw[rmw] (i11) edge node[right] {\small$\lRMW$} (i12);
  \draw[po] (i21) edge (i22);
  \draw[fr] (i11) edge node[above] {\small$\lFRE$} (i21);
  \draw[mo,bend left=8] (i21) edge node[above] {\small$\lCOE$} (i12);
  \draw[rf] (i12) edge node[below] {\small$\lRF$} (i22);
\end{tikzpicture}}
\]
To the right, we present an inconsistent execution corresponding to the outcome 
omitting the initialization event for conciseness.
The $\lRF$ edges and the induced $\lFRE$ edge are forced by the annotated outcome,
while the $\lCOE$ edge is forced because of coherence: 
\ie ordering the writes in the reverse order yields a coherence violation.
The atomicity violation is thus evident.
\qed
\end{example}

\subsection{Global Ordering Constraint}
\label{sec:global}
The third condition---the \emph{global ordering} constraint---is the most complicated 
and is used to rule out out-of-thin-air behaviors.
We will incrementally define a relation $\lAR$ that we require to be acyclic. 

First of all, $\lAR$ includes the external reads-from relation, $\lRFE$,
and the ordering guarantees induced by memory fences and release/acquire accesses.
Specifically, release writes enforce an ordering to any previous event of the same thread,
acquire reads enforce the ordering to subsequent events of the same thread,
while fences are ordered with respect to both prior and subsequent events.
As a final condition, release writes are ordered before any subsequent writes to the same location: 
this is needed for maintaining release sequences.
\begin{align*}
 \lBOB & \defeq \inarr{
 	\lPO\mathbin{;}[\ewlab{\rel}{}{}{}] \cup [\erlab{\acq}{}{}{}]\mathbin{;} \lPO \cup \lPO\mathbin{;}[\lF] \cup [\lF]\mathbin{;}\lPO 
    \cup [\ewlab{\rel}{}{}{}]\mathbin{;}\lPO\rst\lLOC\mathbin{;}[\lW] 
        \tag{\emph{barrier order}}
   }\\
 \lAR &\defeq \lRFE \cup \lBOB \cup \ldots     \tag{\emph{acyclicity relation, more cases to be added}}
\end{align*}
Release/acquire accesses and fences in \IMM play a double role: 
they induce synchronization similar to RC11 as discussed in \cref{sec:model-coherence} 
and also enforce intra-thread instruction ordering as in hardware models. 
The latter role ensures the absence of `load buffering' behaviors in the following examples.

\begin{example}[Load buffering with release writes]
Consider the following program, 
whose annotated outcome disallowed by ARM, POWER,
and the promising semantics.\footnote{In this and other examples, when saying whether a behavior of a program is allowed by ARM/POWER,
we implicitly mean the intended mapping of the program's primitive accesses to ARM/POWER.
See \cref{sec:hardware} for details.}
\[
\inarrII{
  \readInst{\rlx}{a}{x} \comment{1} \\
  \writeInst{\rel}{y}{1} \\
}{\readInst{\rlx}{b}{y} \comment{1} \\
  \writeInst{\rel}{x}{1}  \\
}
\quad\vrule\quad
\inarr{\begin{tikzpicture}[yscale=0.8,xscale=1]
  \node (i11) at (0,  0) {$\erlab{\rlx}{x}{1}{\isnotex}$};
  \node (i12) at (0, -1) {$\ewlab{\rel}{y}{1}{}$};
  \node (i21) at (3,  0) {$\erlab{\rlx}{y}{1}{\isnotex}$};
  \node (i22) at (3, -1) {$\ewlab{\rel}{x}{1}{}$};

  \draw[po] (i11) edge node[right] {\small$\lBOB$} (i12);
  \draw[po] (i21) edge node[right] {\small$\lBOB$} (i22);
  \draw[rf] (i12) edge (i21);
  \draw[rf] (i22) edge node[pos=0.6, below] {\small$\lRFE$} (i11);
\end{tikzpicture}}
\]
\IMM disallows the outcome because of the $\lBOB\cup\lRFE$ cycle.
\qed
\end{example}

\begin{example}[Load buffering with acquire reads]
Consider a variant of the previous program with acquire loads and relaxed stores:
\[
\inarrII{
  \readInst{\acq}{a}{x} \comment{1} \\
  \writeInst{\rlx}{y}{1} \\
}{\readInst{\acq}{b}{y} \comment{1} \\
  \writeInst{\rlx}{x}{1}  \\
}
\quad\vrule\quad
\inarr{\begin{tikzpicture}[yscale=0.8,xscale=1]
  \node (i11) at (0,  0) {$\erlab{\acq}{x}{1}{\isnotex}$};
  \node (i12) at (0, -1) {$\ewlab{\rlx}{y}{1}{}$};
  \node (i21) at (3,  0) {$\erlab{\acq}{y}{1}{\isnotex}$};
  \node (i22) at (3, -1) {$\ewlab{\rlx}{x}{1}{}$};

  \draw[po] (i11) edge node[right] {\small$\lBOB$} (i12);
  \draw[po] (i21) edge node[right] {\small$\lBOB$} (i22);
  \draw[rf] (i12) edge  (i21);
  \draw[rf] (i22) edge node[pos=0.6, below] {\small$\lRFE$} (i11);
\end{tikzpicture}}
\]
\IMM again declares the presented execution as inconsistent
following both ARM and POWER, which forbid the annotated outcome.
The promising semantics, in contrast, allows this outcome
to support a higher-level optimization (namely, elimination of redundant acquire reads).
\qed
\end{example}

Besides orderings due to fences, hardware preserves certain orderings due to syntactic code dependencies.
Specifically, whenever a write depends on some earlier read by a chain of syntactic dependencies or internal reads-from edges 
(which are essentially dependencies through memory),
then the hardware cannot execute the write until it has finished executing the read, 
and so the ordering between them is preserved.
We call such preserved dependency sequences the \emph{preserved program order} ($\lPPO$) and include it in $\lAR$.
In contrast, dependencies between read events are not always preserved,
and so we do not incorporate them in the $\lAR$ relation.
\begin{align*}
 \lDEPS &\defeq \lDATA \cup \lCTRL \cup \lADDR\mathbin{;}\lPO^? \cup \lRMWDEP \cup [\erlab{}{}{}{\isex}]\mathbin{;} \lPO \tag{\emph{syntactic dependencies}}  \\
 \lPPO &\defeq [\lR] \mathbin{;} (\lDEPS \cup \lRFI )^{+} \mathbin{;} [\lW] \tag{\emph{preserved program order}}
\\
 \lAR &\defeq \lRFE \cup \lBOB \cup \lPPO \cup \ldots 
\end{align*}
The extended constraint rules out the weak behaviors of variants of the load buffering example
that use syntactic dependencies to enforce an ordering.

\begin{example}[Load buffering with an address dependency]
Consider a variant of the previous program with an address-dependent read instruction 
in the middle of the first thread:
\[
\inarrII{
  \readInst{\rlx}{a}{x} \comment{1} \\
  \readInst{\rlx}{b}{y + a} \\
  \writeInst{\rlx}{y}{1} \\
}{\readInst{\rlx}{c}{y} \comment{1} \\
  \writeInst{\rel}{x}{1}  \\
}
\quad\vrule\quad
\inarr{\begin{tikzpicture}[yscale=0.5,xscale=1]
  \node (i11) at (1,  0) {$\erlab{\rlx}{x}{1}{\isnotex}$};
  \node (i12) at (0, -1) {$\erlab{\rlx}{y+1}{0}{\isnotex}$};
  \node (i13) at (1, -2) {$\ewlab{\rlx}{y}{1}{}$};
  \node (i21) at (4,  0) {$\erlab{\rlx}{y}{1}{\isnotex}$};
  \node (i22) at (4, -2) {$\ewlab{\rel}{x}{1}{}$};

  \draw[deps,out=180,in=90] (i11) edge node[left] {\small$\lADDR$} (i12);
  \draw[po,out=270,in=180] (i12) edge node[left]  {\small$\lPO$} (i13);
  \draw[po] (i21) edge node[right] {\small$\lBOB$} (i22);
  \draw[rf] (i13) edge  (i21);
  \draw[rf] (i22) edge node[pos=0.6, below] {\small$\lRFE$} (i11);
\end{tikzpicture}}
\]
The displayed execution is \IMM-inconsistent because of the $\lADDR\mathbin{;}\lPO\mathbin{;}\lRFE\mathbin{;}\lBOB\mathbin{;}\lRFE$ cycle.
Hardware implementations cannot produce the annotated behavior because the write to $y$ cannot be issued until 
it has been determined that its address does not alias with $y + a$, 
which cannot be determined until the value of $x$ has been read.
\qed
\end{example}

Similar to syntactic dependencies, 
$\lRFI$ edges are guaranteed to be preserved only on dependency paths from a read to a write,
not otherwise.

\begin{example}[$\lRFI$ is not always preserved]
  \label{ex:rfi-not-preserved}
Consider the following program, whose annotated outcome is allowed by ARMv8.
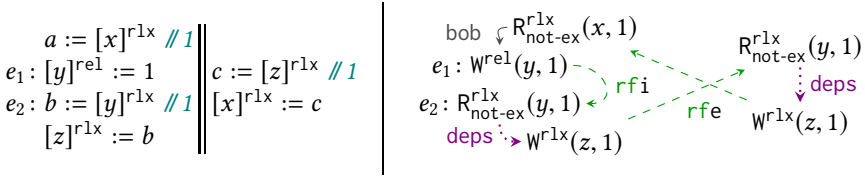
\[
\inarrII{
  \phantom{e_1\colon} \readInst{\rlx}{a}{x} \comment{1} \\
  e_1\colon \writeInst{\rel}{y}{1} \\
  e_2\colon \readInst{\rlx}{b}{y} \comment{1} \\
  \phantom{e_2\colon} \writeInst{\rlx}{z}{b} \\
}{\readInst{\rlx}{c}{z} \comment{1} \\
  \writeInst{\rlx}{x}{c}
}
\quad\vrule\quad
\inarr{\begin{tikzpicture}[yscale=0.5,xscale=1]
  \node (i11) at (1,  0) {$\erlab{\rlx}{x}{1}{\isnotex}$};
  \node (i12) at (0, -1) {$e_1\colon \ewlab{\rel}{y}{1}{}$};
  \node (i13) at (0, -2) {$e_2\colon \erlab{\rlx}{y}{1}{\isnotex}$};
  \node (i14) at (1, -3) {$\ewlab{\rlx}{z}{1}{}$};

  \node (i21) at (4, -0.5) {$\erlab{\rlx}{y}{1}{\isnotex}$};
  \node (i22) at (4, -2.5) {$\ewlab{\rlx}{z}{1}{}$};

  \draw[po,out=180,in=90] (i11) edge node[left] {\small$\lBOB$\;} (i12);
  \draw[rf,out=0,in=0] (i12) edge node[right] {\small $\lRFI$} (i13);
  \draw[deps,out=270,in=180] (i13) edge node[left] {\small $\lDEPS$} (i14);
  \draw[deps] (i21) edge node[right] {\small $\lDEPS$} (i22);
  \draw[rf] (i22) edge (i11);
  \draw[rf] (i14) edge node[pos=.66,below=2pt] {\small $\lRFE$} (i21);
\end{tikzpicture}}
\]
To the right, we show the corresponding execution (the $\lRF$ edges are forced because of the outcome).
Had we included $\lRFI$ unconditionally as part of $\lAR$, we would have disallowed the behavior,
because it would have introduced an $\lAR$ edge between events $e_1$ and $e_2$, and therefore an $\lAR$ cycle.
\qed
\end{example}

Note that we do not include $\lFRI$ in $\lPPO$ since it is not preserved in ARMv7~\cite{herding-cats}
(unlike in x86-TSO, POWER, and ARMv8).
Thus, as ARMv7 (as well as the Flowing and POP models of ARM in~\cite{arm8-model}),
\IMM allows the weak behavior from \cite[§6]{trns}.

Next, we include $\lDETOUR \defeq (\lCOE\mathbin{;} \lRFE) \cap \lPO$ in $\lAR$.
It captures the case when a read $r$ does not read from an earlier write $w$ to the same location
but from a write $w'$ of a different thread. 
In this case, both ARM and POWER enforce an ordering between $w$ and $r$.
Since the promising semantics also enforces such orderings
(due to the certification requirement in every future memory, see~\cref{sec:prom-compilation}),
\IMM also enforces the ordering by including $\lDETOUR$ in $\lAR$.
\begin{example}[Enforcing $\lDETOUR$]
The annotated behavior of the following program is disallowed by POWER, ARM, and the promising semantics, 
and so it must be disallowed by \IMM.
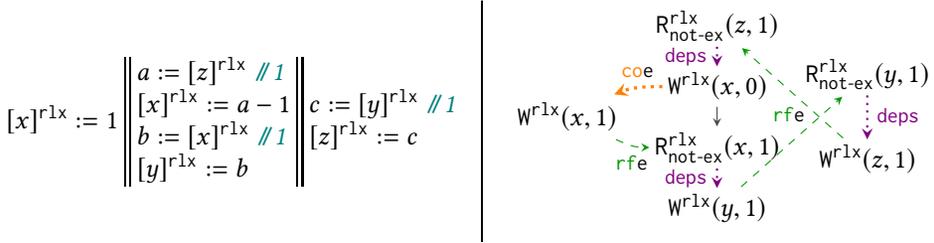
\[
\inarrIII{
  \writeInst{\rlx}{x}{1}
}{\readInst{\rlx}{a}{z}  \comment{1} \\
  \writeInst{\rlx}{x}{a -1} \\
  \readInst{\rlx}{b}{x} \comment{1} \\
  \writeInst{\rlx}{y}{b} \\
}{\readInst{\rlx}{c}{y}  \comment{1} \\
  \writeInst{\rlx}{z}{c}
}
\quad\vrule\quad
\inarr{\begin{tikzpicture}[yscale=0.8]
  \node (i01) at (0, -1.5) {$\ewlab{\rlx}{x}{1}{}$};

  \node (i11) at (2,  0) {$\erlab{\rlx}{z}{1}{\isnotex}$};
  \node (i12) at (2, -1) {$\ewlab{\rlx}{x}{0}{}$};
  \node (i13) at (2, -2) {$\erlab{\rlx}{x}{1}{\isnotex}$};
  \node (i14) at (2, -3) {$\ewlab{\rlx}{y}{1}{}$};

  \node (i21) at (4, -0.8) {$\erlab{\rlx}{y}{1}{\isnotex}$};
  \node (i22) at (4, -2.2) {$\ewlab{\rlx}{z}{1}{}$};

  \draw[mo,out=180,in=25] (i12) edge node[above] {\smaller\smaller$\lCOE$} (i01);
  \draw[deps] (i11) edge node[left] {\smaller\smaller$\lDEPS$} (i12);
  \draw[rf,out=-25,in=180] (i01) edge node[below] {\smaller\smaller$\lRFE$} (i13);
  \draw[po] (i12) edge node[left] {} (i13);
  \draw[deps] (i13) edge node[left] {\smaller\smaller$\lDEPS$} (i14);

  \draw[deps] (i21) edge node[right] {\smaller\smaller$\lDEPS$} (i22);
  \draw[rf] (i22) edge node[pos=.54,below=3pt] {\smaller\smaller$\lRFE$} (i11);
  \draw[rf] (i14) edge (i21);
\end{tikzpicture}}
\]
If we were to exclude $\lDETOUR$ from the acyclicity condition, 
the execution of the program shown above to the right would have been allowed by \IMM.
\qed
%
\end{example}

We move on to a constraint about SC fences.
Besides constraining the ordering of events from the same thread, 
SC fences induce inter-thread orderings whenever there is a coherence path between them.
Following the RC11 model \cite{scfix}, 
we call this relation $\lPSC$ and include it in $\lAR$.
\begin{align*}
 \lPSC &\defeq [\lF^\sco] \mathbin{;}\lHB\mathbin{;} \lECO\mathbin{;}\lHB\mathbin{;} [\lF^\sco] \tag{\emph{partial SC fence order}} \\
 \lAR &\defeq \lRFE \cup \lBOB \cup \lPPO \cup \lDETOUR \cup \lPSC \cup \ldots 
\end{align*}

\begin{example}[Independent reads of independent writes]
\label{ex:iriw}
Similar to POWER, \IMM is not ``multi-copy atomic''~\cite{tutorial_arm_power} (or ``memory atomic''~\cite{Zhang18}).
In particular, it allows the weak behavior of the IRIW litmus test even with release-acquire accesses. 
To forbid the weak behavior, one has to use SC fences:
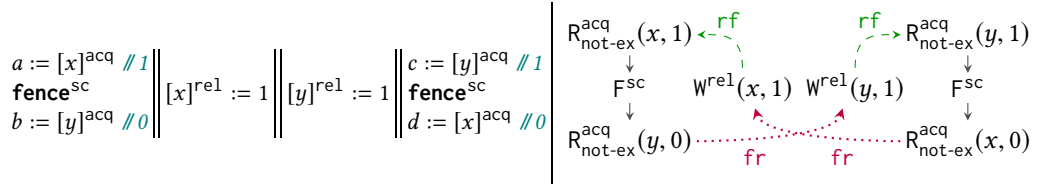
\[\mbox{\small
$\inarrIV{
  \readInst{\acq}{a}{x} \comment{1} \\
  \fenceInst{\sco} \\ 
  \readInst{\acq}{b}{y} \comment{0} 
}{
  \writeInst{\rel}{x}{1}
}{
  \writeInst{\rel}{y}{1}
}{  
  \readInst{\acq}{c}{y} \comment{1} \\
  \fenceInst{\sco} \\ 
  \readInst{\acq}{d}{x} \comment{0} \\
}$}
~~\vrule~
\inarr{\begin{tikzpicture}[xscale=0.75,yscale=1.4]
  \node (11)  at (-1,0.5) {$\ewlab{\rel}{x}{1}{}$ };
  \node (21)  at (-3,1) {$\erlab{\acq}{x}{1}{\isnotex}$ };
  \node (215) at (-3,0.5) {$\flab{\sco}$ };
  \node (22)  at (-3,0) {$\erlab{\acq}{y}{0}{\isnotex}$ };
  \node (31)  at (3,1) {$\erlab{\acq}{y}{1}{\isnotex}$ };
  \node (315) at (3,0.5) {$\flab{\sco}$ };
  \node (32)  at (3,0) {$\erlab{\acq}{x}{0}{\isnotex}$ };
  \node (41)  at (1,0.5) {$\ewlab{\rel}{y}{1}{}$ };
  \draw[po] (21) edge (215);
  \draw[po] (215) edge (22);
  \draw[po] (31) edge (315);
  \draw[po] (315) edge (32);
  \draw[rf,out=90,in=0] (11) edge node[above]{\small $\lRF$} (21);
  \draw[rf,out=90,in=180] (41) edge node[above]{\small $\lRF$}  (31);
  \draw[fr,out=180,in=300,looseness=0.3] (32) edge node[pos=.4,below]{\small $\lFR$}  (11);
  \draw[fr,out=0,in=240,looseness=0.3] (22) edge node[pos=.4,below]{\small $\lFR$}  (41);
\end{tikzpicture}}
\]
The execution corresponding to the weak outcome is shown to the right.
For soundness \wrt the promising semantics, \IMM declares this execution to be inconsistent
(which is also natural since it has an SC fence between every two instructions).
It does so due to the $\lPSC$ cycle: each fence reaches the other by a $\lPO\mathbin{;}\lFR\mathbin{;}\lRF\mathbin{;}\lPO \suq \lPSC$ path.
When the SC fences are omitted, since POWER allows the weak outcome, \IMM allows it as well.
\qed
\end{example}

\begin{example}
To illustrate why we make $\lPSC$ part of $\lAR$, rather than a separate acyclicity condition (as in RC11), consider the following program, 
whose annotated outcome is forbidden by the promising semantics.
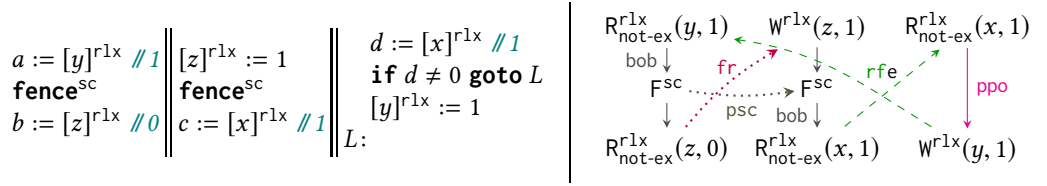
\[
\inarrIII{ 
\readInst{\rlx}{a}{y} \comment{1} \\
\fenceInst{\sco} \\
\readInst{\rlx}{b}{z} \comment{0} \\
}{
\writeInst{\rlx}{z}{1} \\
\fenceInst{\sco} \\
\readInst{\rlx}{c}{x} \comment{1} \\
}{
\quad\readInst{\rlx}{d}{x} \comment{1} \\
\quad\ifGotoInst{d \neq 0}{L} \\
\quad\writeInst{\rlx}{y}{1}  \\
L\colon
}
\quad\vrule\quad
\inarr{\begin{tikzpicture}[yscale=0.8]
  \node (i01) at (0,  0) {$\erlab{\rlx}{y}{1}{\isnotex}$};
  \node (i02) at (0, -1) {$\flab{\sco}$};
  \node (i03) at (0, -2) {$\erlab{\rlx}{z}{0}{\isnotex}$};

  \node (i11) at (2,  0) {$\ewlab{\rlx}{z}{1}{}$};
  \node (i12) at (2, -1) {$\flab{\sco}$};
  \node (i13) at (2, -2) {$\erlab{\rlx}{x}{1}{\isnotex}$};

  \node (i21) at (4,  0) {$\erlab{\rlx}{x}{1}{\isnotex}$};
  \node (i22) at (4, -2) {$\ewlab{\rlx}{y}{1}{}$};

  \draw[po] (i01) edge node[left] {\smaller\smaller$\lBOB$} (i02); 
  \draw[po] (i02) edge (i03);
  \draw[po] (i11) edge (i12); 
  \draw[po] (i12) edge node[left] {\smaller\smaller$\lBOB$}  (i13);
  \draw[ppo] (i21) edge node[right] {\smaller\smaller$\lPPO$} (i22);
  
  \draw[rf] (i22) edge[out=145,in=350] (i01);
  \draw[fr,bend left=10] (i03) edge node[above] {\smaller\smaller$\lFR$ } (i11);
  \draw[rf] (i13) edge node[above] {\smaller\smaller$\lRFE\ \ \ \ $} (i21);
  \draw[sc,bend right=10] (i02) edge node[below] {\smaller\smaller$\lPSC$} (i12);
\end{tikzpicture}}
\]
The execution corresponding to that outcome is shown to the right.
For soundness \wrt the promising semantics, \IMM declares this execution inconsistent,
due to the $\lAR$ cycle. \qed
\end{example}

The final case we add to $\lAR$ is to support the questionable semantics of RMWs in the promising semantics.
The promising semantics requires the ordering between the store of a release RMW and subsequent stores to be preserved,
something that is not generally guaranteed by ARMv8.
For this reason, to be able to compile the promising semantics to \IMM,
and still  keep \IMM as a sound abstraction of ARMv8,
we include the additional ``RMW mode'' in RMW instructions,
which propagates to their induced write events.
Then, we include $[\ewlab{}{}{}{\strong}] \mathbin{;} \lPO \mathbin{;} [\lW]$ in $\lAR$, yielding the following (final) definition:
\begin{align*}
 \lAR &\defeq \lRFE \cup \lBOB \cup \lPPO \cup \lDETOUR \cup \lPSC  \cup [\ewlab{}{}{}{\strong}] \mathbin{;} \lPO \mathbin{;} [\lW]
\end{align*}

\begin{example}
  \label{ex:strong-rmw}
  The following example demonstrates the problem in the intended mapping of the promising semantics to ARMv8. 
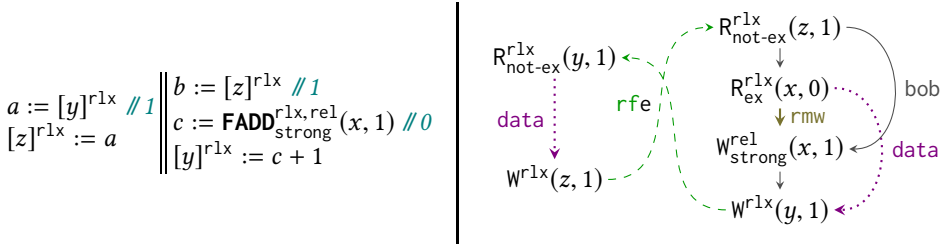
\[
\inarrII{
  \readInst{\rlx}{a}{y} \comment{1} \\
  \writeInst{\rlx}{z}{a}
}{
  \readInst{\rlx}{b}{z} \comment{1} \\
  \incInst{\rlx}{\rel}{c}{x}{1}{\strong} \comment{0} \\
  \writeInst{\rlx}{y}{c+1} \\
}
\quad\vrule\quad
\inarr{\begin{tikzpicture}[yscale=0.8,xscale=1]
  \node (i11) at (0,  -0.5) {$\erlab{\rlx}{y}{1}{\isnotex}$};
  \node (i12) at (0, -2.5) {$\ewlab{\rlx}{z}{1}{}$};
  \node (i21) at (3,  0) {$\erlab{\rlx}{z}{1}{\isnotex}$};
  \node (i22) at (3,  -1) {$\erlab{\rlx}{x}{0}{\isex}$};
  \node (i23) at (3,  -2) {$\ewlab{\rel}{x}{1}{\strong}$};
  \node (i24) at (3,  -3) {$\ewlab{\rlx}{y}{1}{}$};
  \draw[deps] (i11) edge node[left] {\small$\lDATA$} (i12);
  \draw[rmw] (i22) edge node[right] {\small$\lRMW$} (i23);
  \draw[po] (i21) edge (i22);
  \draw[deps,out=0,in=0] (i22) edge node[right] {\small$\lDATA$} (i24);
    \draw[po,out=0,in=0] (i21) edge node[right] {\small$\lBOB$} (i23);
  \draw[po] (i23) edge (i24);
  \draw[rf,out=180,in=0] (i24) edge  (i11);
  \draw[rf,out=0,in=180] (i12) edge node[left] {\small$\lRFE$} (i21);
\end{tikzpicture}}
\]
The promising semantics disallows the annotated behavior (it requires a promise of $y=1$,
but this promise cannot be certified for a future memory that will not allow the atomic 
increment from $0$---see \cref{sec:promise} and \cref{ex:w-strong-iss}).
It is disallowed by \IMM due to the $\lAR$ cycle (from the read of $y$):
$\lPPO\mathbin{;}\lRFE\mathbin{;}\lBOB\mathbin{;}[\ewlab{}{}{}{\strong}] \mathbin{;} \lPO \mathbin{;} [\lW]\mathbin{;}\lRFE$.
Without additional barriers, ARMv8 allows this behavior.
Thus, our mapping of \IMM to ARMv8 places a barrier (``ld fence'') after strong RMWs (see \cref{sec:arm8}).
\qed
\end{example}

\subsection{Consistency}

Putting everything together, \IMM-consistency is defined as follows.

\begin{definition}
\label{def:model}
$G$ is called \emph{\IMM-consistent} if the following hold:
\begin{itemize}
\item $\codom{G.\lRF} = G.\lR$. \labelAxiom{$\lRF$-completeness}{ax:comp}
\item For every location $\loc\in\Loc$, $G.\lCO$ totally orders $G.\ewlab{}{\loc}{}{}$. \labelAxiom{$\lCO$-totality}{ax:total}
\item $G.\lHB\mathbin{;}G.\lECO^?$ is irreflexive. \labelAxiom{coherence}{ax:coh}
\item $G.\lRMW \cap (G.\lFRE\mathbin{;}G.\lCOE) = \emptyset$.  \labelAxiom{atomicity}{ax:at}
\item $G.\lAR$ is acyclic. \labelAxiom{no-thin-air}{ax:nta}
\end{itemize}
\end{definition}

\section{From \IMM to hardware models}
\label{sec:hardware}

In this section, we provide mappings from \IMM to the main hardware architectures and establish their soundness. 
That is, if some behavior is allowed by a target architecture on a target program, then
it is also allowed by \IMM on the source of that program.
Since the models of hardware we consider are declarative, we formulate the soundness results
on the level of execution graphs, keeping the connection to programs only implicit.
Indeed, a mapping of \IMM instructions to real architecture instructions
naturally induces a mapping of \IMM execution graphs to target architecture execution graphs.
Then, it suffices to establish that the consistency of a target execution graph (as defined by the target memory model)
entails the \IMM-consistency of its source execution graph.
This is a common approach for studying declarative models, (see, \eg \cite{c11comp}),
and allows us to avoid orthogonal details of the target architectures' instruction sets.

Next, we study the mapping to POWER (\cref{sec:power}) and ARMv8 (\cref{sec:arm8}).
We note that \IMM can be straightforwardly shown to be weaker than x86-TSO,
and thus the identity mapping (up to different syntax) is a correct compilation scheme from
\IMM to x86-TSO.
The mapping to ARMv7 is closely related to POWER, and it is discussed in \cref{sec:power} as well.
RISC-V~\cite{RISCV,www:risc-herd} is stronger than ARMv8 and therefore soundness of mapping to it from IMM
follows from the corresponding ARMv8 result.

\subsection{From \IMM to POWER}
\label{sec:power}

\begin{figure}[t]
\small
\centering
$\begin{array}{@{}r@{}c@{{}\approx{}}l@{\qquad \qquad \quad}r@{}c@{{}\approx{}}l@{}}
\compile{\readInst{\rlx}{r}{e}} &&  ``\texttt{ld}" & 
\compile{\writeInst{\rlx}{e_1}{e_2}} &&  ``\texttt{st}" \\
\compile{\readInst{\acq}{r}{e}} && ``\texttt{ld;cmp;bc;isync}"  &
\compile{\writeInst{\rel}{e_1}{e_2}} && ``\texttt{lwsync;st}" \\ 
\compile{\fenceInst{\neq \sco}}  && ``\texttt{lwsync}" &
\compile{\fenceInst{\sco}} && ``\texttt{sync}" \\
\compile{\incInst{o_\lR}{o_\lW}{r}{e_1}{e_2}{o_\lU}} &&
\multicolumn{4}{@{}l@{}}{ 
  {\rm wmod}(o_\lW) \doubleplus
  ``\textsf{L:}\texttt{lwarx;stwcx.;bc  }\textsf{L}"
  \doubleplus {\rm rmod}(o_\lR)
}
\\ 
\compile{\casInst{o_\lR}{o_\lW}{r}{e}{e_\lR}{e_\lW}{o_\lU}} &&
\multicolumn{4}{@{}l@{}}{
  {\rm wmod}(o_\lW) \doubleplus
  ``\textsf{L:}\texttt{lwarx;cmp;bc } \textsf{Le}\texttt{;stwcx.;bc  }\textsf{L}\texttt{;}\textsf{Le:}"
  \doubleplus {\rm rmod}(o_\lR)
} \\
\multicolumn{6}{@{}l@{}}{
{\rm wmod}(o_\lW) \defeq o_\lW = \rel \; ? \; ``\texttt{lwsync;}" \; : \; ``" \qquad \qquad \qquad \qquad
{\rm rmod}(o_\lR) \defeq o_\lR = \acq \; ? \; ``\texttt{;isync}" \; : \; ``" 
}
\end{array}$
\caption{Compilation scheme from \IMM to POWER.}
\label{fig:compPower}
\end{figure}

The intended mapping of \IMM to POWER is presented schematically in \cref{fig:compPower}.
It follows the C/C++11 mapping~\cite{www:mappings} (see also~\cite{tutorial_arm_power}):
relaxed reads and writes are compiled down to plain machine loads and stores;
acquire reads are mapped to plain loads followed by a control dependent instruction fence;
release writes are mapped to plain writes preceded by a lightweight fence;
acquire/release/acquire-release fences are mapped to POWER's lightweight fences;
and SC fences are mapped to full fences.
The compilation of RMWs requires a loop which repeatedly uses
POWER's load-reserve/store-conditional instructions until the store-conditional succeeds.
RMWs are accompanied with barriers for acquire/release modes as reads and writes.
CAS instructions proceed to the conditional write only after checking that the loaded value meets the required condition.
Note that \IMM's strong RMWs are compiled  to POWER as normal RMWs.

To simplify our correctness proof, we take advantage of the fact  that 
release writes and release RMWs are compiled down as their relaxed counterparts
with a preceding $\fenceInst{\rel}$.
Thus, we consider the compilation as if it happens in two steps: first, release writes and RMWs are split
to release fences and their relaxed counterparts; and then, the mapping of \cref{fig:compPower} is applied
(for a program without release writes and release RMWs).
Accordingly, we establish 
(\emph{i}) the soundness of the split of release accesses;
and (\emph{ii}) the correctness of the mapping in the absence of release accesses.\footnote{Since \IMM does not have a
primitive that corresponds to POWER's instruction fence, we cannot apply the same trick for acquire reads.}
The first obligation is solely on the side of \IMM, and is formally presented next.

\begin{theorem}
  \label{thm:rel-write-to-fence}
Let $G$ be an \IMM execution graph such that 
$G.\lPO \mathbin{;}[G.\ewlab{\rel}{}{}{}] \suq G.\lPO^?\mathbin{;} [G.\flab{\rel}]\mathbin{;} G.\lPO \cup G.\lRMW$.
Let $G'$ be the \IMM execution graph obtained from $G$ by weakening the access modes of release write events
to a relaxed mode.
Then, \IMM-consistency of $G'$ implies \IMM-consistency of $G$. 
\end{theorem}

Next, we establish the correctness of the mapping (in the absence of release writes) 
with respect to the model of the POWER architecture of~\citet{herding-cats}, which we denote by \POWER.
As \IMM, the \POWER model is declarative, defining allowed outcomes via consistent execution graphs.
Its labels are similar to \IMM's labels (\cref{def:label}) with the following exceptions:
\begin{itemize}
\item Read/write labels have the form 
$\prlab{x}{v}$ and $\pwlab{x}{v}$: they do not include additional modes.
\item There are three fence labels (listed here in increasing strength order):
an ``instruction fence'' ($\flab{\isync}$),
a ``lightweight fence'' ($\flab{\lwsync}$), and a ``full fence'' ($\flab{\sync}$).
\end{itemize}

In turn, \POWER execution graphs are defined as those of \IMM (cf.\ \cref{def:execution}),
except for the CAS dependency, $\lRMWDEP$, which is not present in \POWER executions.
The next definition presents the correspondence between \IMM execution graphs and their mapped \POWER ones
following the compilation scheme in \cref{fig:compPower}.

\begin{definition}
\label{def:cmp_exec_power}
Let $G$ be an \IMM execution graph with whole serial numbers ($\lSN[G.\lE] \suq \N$),
such that $G.\ewlab{\rel}{}{}{}=\emptyset$.
A \POWER execution graph $G_p$ \emph{corresponds} to $G$ if the following hold:
\begin{itemize}
\item 
$G_p.\lE = G.\lE \cup \set{\tup{i,n+0.5} \st \tup{i,n} \in (G.\erlab{\acq}{}{}{} \setminus \dom{G.\lRMW})
\cup \codom{[G.\erlab{\acq}{}{}{}]\mathbin{;}G.\lRMW}}$ 
\\ (new events are added after acquire reads and acquire RMW pairs)
\item  
$G_p.\lLAB = \set{e \mapsto \compile{G.\lLAB(e)} \st e\in G.\lE} \cup \set{ e\mapsto  \flab{\isync} \st e\in G_p.\lE \setminus G.\lE}$ where:
$$\begin{array}{r@{\;}l@{\qquad\quad}r@{\;}l}
\compile{\erlab{o_\lR}{x}{v}{s}}         & \defeq \prlab{x}{v}                     
& \compile{\flab{\acq}} = \compile{\flab{\rel}} = \compile{\flab{\acqrel}} &\defeq \flab{\lwsync} \\
\compile{\ewlab{o_\lW}{x}{v}{o_\lU}} & \defeq \pwlab{x}{v}  & 
\compile{\flab{\sco}}  &\defeq \flab{\sync}
\end{array}$$
\item  
$G.\lRMW = G_p.\lRMW$, $G.\lDATA = G_p.\lDATA$, and $G.\lADDR = G_p.\lADDR$
\\ (the compilation does not change RMW pairs and data/address dependencies)
\item
$G.\lCTRL \suq G_p.\lCTRL$
\\ (the compilation only adds control dependencies)
\item 
$[G.\erlab{\acq}{}{}{}] \mathbin{;} G.\lPO \suq G_p.\lRMW \cup G_p.\lCTRL$
\\ (a control dependency is placed from every acquire read)
\item
$[G.\erlab{}{}{}{\isex}] \mathbin{;}G.\lPO \suq G_p.\lCTRL \cup G_p.\lRMW \cap G_p.\lDATA $
\\ (exclusive reads entail a control dependency to any future event, 
except for their immediate exclusive write successor if arose from an atomic increment)
\item
$G.\lDATA \mathbin{;} [\codom{G.\lRMW}] \mathbin{;} G.\lPO \suq G_p.\lCTRL$
\\ (data dependency to an exclusive write entails a control dependency to any future event)
\item
$G.\lRMWDEP \mathbin{;} G.\lPO \suq G_p.\lCTRL$
\\ (CAS dependency to an exclusive read entails a control dependency to any future event)
\end{itemize}
\end{definition}

Next, we state our theorem that ensures \IMM-consistency if the corresponding \POWER
execution graph is \POWER-consistent.
Due to lack of space, we do not include here the (quite elaborate) definition of \POWER-consistency.
For that definition, we refer the reader to \cite{herding-cats} (\citeapp{sec:power_consistent}{Appendix B}
provides the definition we used in our development).

\begin{theorem}
  \label{thm:imm-to-power}
Let $G$ be an \IMM execution graph with whole serial numbers ($\lSN[G.\lE] \suq \N$),
such that $G.\ewlab{\rel}{}{}{}=\emptyset$,
and let $G_p$ be a \POWER execution graph that corresponds to $G$.
Then, \POWER-consistency of $G_p$ implies \IMM-consistency of $G$. 
\end{theorem}

The ARMv7 model in~\cite{herding-cats} is very similar to the POWER model.
There are only two differences.
First, ARMv7 lacks an analogue for POWER's lightweight fence ($\lwsync$).
Second, ARMv7 has a weaker preserved program order than POWER, 
which in particular does not always include $[G.\erlab{}{}{}{}]; G.\lPO\rst{G.\lLOC}; [G.\ewlab{}{}{}{}]$
(the ${\lPO\rst{\lLOC}}/{\lcc}$ rule is excluded, see \citeapp{sec:power_consistent}). 
In our proofs for POWER, however, we never rely on POWER's $\lPPO$, but rather assume the weaker one of ARMv7.
The compilation schemes to ARMv7 are essentially the
same as those to POWER substituting the corresponding
ARMv7 instructions for the POWER ones: $\texttt{dmb}$ instead of $\texttt{sync}$
and $\texttt{lwsync}$, and $\texttt{isb}$ instead of $\texttt{isync}$. 
Thus, the correctness of compilation to ARMv7 follows directly 
from the correctness of compilation to POWER.

\subsection{From \IMM to ARMv8}
\label{sec:arm8}

\begin{figure}[t]
\small
\centering
$\begin{array}{@{}r@{}c@{}l@{\qquad \qquad \qquad \;\;}r@{}c@{{}\approx{}}l@{}}
\compile{\readInst{\rlx}{r}{e}} & {}\approx{} &  ``\texttt{ldr}" & 
\compile{\writeInst{\rlx}{e_1}{e_2}} &&  ``\texttt{str}" \\
\compile{\readInst{\acq}{r}{e}} & {}\approx{} & ``\texttt{ldar}"  &
\compile{\writeInst{\rel}{e_1}{e_2}} && ``\texttt{stlr}" \\ 
\compile{\fenceInst{\acq}}  & {}\approx{} & ``\texttt{dmb.ld}" &
\compile{\fenceInst{\neq\acq}} && ``\texttt{dmb.sy}" \\
\compile{\incInst{o_\lR}{o_\lW}{r}{e_1}{e_2}{o_\lU}} & {}\approx{} &
\multicolumn{4}{@{}l@{}}{
  ``\textsf{L:}" \doubleplus {\rm ld}(o_\lR) \doubleplus
  {\rm st}(o_\lW) \doubleplus ``\texttt{bc  }\textsf{L}" \doubleplus
  {\rm dmb}(o_\lU)
}\\ 
\compile{\casInst{o_\lR}{o_\lW}{r}{e}{e_\lR}{e_\lW}{o_\lU}} & {}\approx{} &
\multicolumn{4}{@{}l@{}}{  
  ``\textsf{L:}" \doubleplus {\rm ld}(o_\lR) \doubleplus ``\texttt{cmp;bc } \textsf{Le}
  \texttt{;}" \doubleplus {\rm st}(o_\lW) \doubleplus
  ``\text{bc  }\textsf{L}\texttt{;}\textsf{Le:}"
  \doubleplus
  {\rm dmb}(o_\lU)
} 
\end{array}$ \\
$\begin{array}{@{}l@{}l@{\quad \qquad}l@{}}
{\rm ld}(o_\lR) & \defeq o_\lR = \acq \; ? \; ``\texttt{ldaxr;}" \; : \; ``\texttt{ldxr;}" &
{\rm st}(o_\lW) \defeq o_\lW = \rel \; ? \; ``\texttt{stlxr.;}" \; : \; ``\texttt{stxr.;}" \\
{\rm dmb}(o_\lU) \, & 
\multicolumn{2}{@{}l@{}}{\defeq o_\lU = \strong \; ? \; ``\texttt{;dmb.ld}" \; : ``"} 
\end{array}$
\caption{Compilation scheme from \IMM to ARMv8.}
\label{fig:compArm}
\end{figure}

The intended mapping of \IMM to ARMv8 is presented schematically in \cref{fig:compArm}.
It is identical to the mapping to POWER (\cref{fig:compPower}), except for the following:
\begin{itemize}
\item Unlike POWER, ARMv8 has machine instructions for acquire loads (\texttt{ldar}) and release stores (\texttt{stlr}),
which are used instead of placing barriers next to plain loads and stores.
\item ARMv8 has a special \texttt{dmb.ld} barrier that is used for \IMM's acquire fences.
On the other side, it lacks an analogue for \IMM's release fence, for which a full barrier (\texttt{dmb.sy}) 
is used.
\item As noted in \cref{ex:strong-rmw}, the mapping of \IMM's \emph{strong} RMWs requires placing
a \texttt{dmb.ld} barrier after the exclusive write.
\end{itemize}

As a model of the ARMv8 architecture,
we use its recent official declarative model~\cite{ARMv82model} (see also \cite{Pulte-al:POPL18})
which we denote by \ARM.\footnote{We only describe the fragment of the model that is needed for mapping of \IMM,
thus excluding  sequentially consistent reads and \texttt{isb} fences.}
Its labels are given by:
\begin{itemize}
\item \ARM read label: $\erlab{o_\lR}{\loc}{v}{}$ where $\loc\in\Loc$, $v\in\Val$, and $o_\lR\in\set{\rlx,\lQ}$.
\item \ARM write label: $\ewlab{o_\lW}{\loc}{v}{}$ where $\loc\in\Loc$, $v\in\Val$, and $o_\lW\in\set{\rlx,\lL}$.
\item \ARM fence label: $\flab{o_\lF}$ where $o_\lF\in\set{\ld,\full}$.
\end{itemize}
In turn, \ARM's execution graphs are defined as \IMM's ones,
except for the CAS dependency, $\lRMWDEP$, which is not present in \ARM executions.
As we did for POWER, we first interpret the intended compilation on execution graphs:

\begin{definition}
\label{def:cmp_exec_arm}
Let $G$ be an \IMM execution graph with whole serial numbers ($\lSN[G.\lE] \suq \N$).
An \ARM execution graph $G_a$ \emph{corresponds} to $G$ if the following hold (we skip the explanation of conditions that appear in \cref{def:cmp_exec_power}):
\begin{itemize}
\item 
$G_a.\lE = G.\lE \cup \set{\tup{i,n+0.5} \st \tup{i,n} \in G.\ewlab{}{}{}{\strong}}$ 
\\ (new events are added after strong exclusive writes)
\item  
$G_a.\lLAB = \set{ e \mapsto \compile{G.\lLAB(e)} \st e\in G.\lE} \cup \set{ e\mapsto \flab{\ld} \st e\in G_a.\lE \setminus G.\lE}$ where:
$$\begin{array}{r@{\;}l@{\qquad\quad}r@{\;}l}
\compile{\erlab{\rlx}{x}{v}{s}} & \defeq \rlab{\rlx}{x}{v}  &  \compile{\ewlab{\rlx}{x}{v}{o_\lU}} & \defeq \wlab{\rlx}{x}{v} \\
\compile{\erlab{\acq}{x}{v}{s}} & \defeq \rlab{\lQ}{x}{v}  & \compile{\ewlab{\rel}{x}{v}{o_\lU}} & \defeq \wlab{\lL}{x}{v} \\
\compile{\flab{\acq}} & \defeq \flab{\ld}& \compile{\flab{\rel}} =\compile{\flab{\acqrel}} = \compile{\flab{\sco}} & \defeq \flab{\full}
\end{array}$$
\item  
$G.\lRMW = G_a.\lRMW$, $G.\lDATA = G_a.\lDATA$, and $G.\lADDR = G_a.\lADDR$
\item
$G.\lCTRL \suq G_a.\lCTRL$
\item
$[G.\erlab{}{}{}{\isex}] \mathbin{;}G.\lPO \suq G_a.\lCTRL \cup G_a.\lRMW \cap G_a.\lDATA $
\item
$G.\lRMWDEP \mathbin{;} G.\lPO \suq G_a.\lCTRL$
\end{itemize}
\end{definition}

Next, we state our theorem that ensures \IMM-consistency if the corresponding \ARM
execution graph is \ARM-consistent.
Again, due to lack of space, we do not include here the definition of \ARM-consistency.
For that definition, we refer the reader to \cite{ARMv82model,Pulte-al:POPL18} (\citeapp{sec:arm_consistent}{Appendix C}
provides the definition we used in our development).

\begin{theorem}
  \label{thm:imm-to-arm}
Let $G$ be an \IMM execution graph with whole serial numbers ($\lSN[G.\lE] \suq \N$),
and let $G_a$ be an \ARM execution graph that corresponds to $G$.
Then, \ARM-consistency of $G_a$ implies \IMM-consistency of $G$. 
\end{theorem}

\section{From C11 and RC11 to \IMM}
\label{sec:c11}

In this section, we establish the correctness of the mapping from the C11 and RC11 models to \IMM.
Since C11 and RC11 are defined declaratively and \IMM-consistency is very close to (R)C11-consistency,
these results are straightforward.

Incorporating the fixes from \citet{c11comp} and \citet{scfix} to the original C11 model of \citet{Batty-al:POPL11},
and restricting attention to the fragment of C11 that has direct \IMM counterparts (thus, excluding non-atomic and SC accesses),
C11-consistency is defined follows.

\begin{definition}
\label{def:c11}
$G$ is called \emph{C11-consistent} if the following hold:
\begin{itemize}
\item $\codom{G.\lRF} = G.\lR$.
\item For every location $\loc\in\Loc$, $G.\lCO$ totally orders $G.\ewlab{}{\loc}{}{}$.
\item $G.\lHBs\mathbin{;}G.\lECO^?$ is irreflexive.
\item $G.\lRMW \cap (G.\lFRE\mathbin{;}G.\lCOE) = \emptyset$.
\item $[\lF^\sco] \mathbin{;}(\lHBs \cup \lHBs\mathbin{;} \lECO\mathbin{;}\lHBs)\mathbin{;} [\lF^\sco]$ is acyclic.
\end{itemize}
\end{definition}

It is easy to show that \IMM-consistency implies C11-consistency,
and consequently, the identity mapping is a correct compilation from this fragment of C11 to \IMM.
This result can be extended to include non-atomic and SC accesses as follows:
\squishlist
\item Non-atomic accesses provide weaker guarantees than relaxed accesses,
and are not needed for accounting for \IMM's behaviors.
Put differently, one may assume that the compilation from C11 to \IMM
first strengthens all non-atomic accesses to relaxed accesses.
Compilation correctness then follows from the soundness of
this strengthening and our result that excludes non-atomics. 
\item The semantics of SC accesses in C11 was shown to be too strong in~\cite{Manerkar-al:CoRR16,scfix}
to allow the intended compilation to POWER and ARMv7. 
If one applies the fix proposed in~\cite{scfix}, 
then compilation correctness could be established following their reduction,
that showed that it is sound to globally split SC accesses to SC fences and release/acquire accesses on the source level.
This encoding yields the (two) expected compilation schemes for SC loads and stores on x86, ARMv7, and POWER.
On the other hand, handling ARMv8's specific instructions for SC accesses is left for future work.
We note that the usefulness and the ``right semantics'' for SC accesses is still under discussion.
The Promising semantics, for instance, does not have primitive SC accesses at all and implements them using SC fences.
\squishend

In turn, RC11 (ignoring the part related to SC accesses) is obtained by strengthening \cref{def:c11}
with a condition asserting that $G.\lPO \cup G.\lRF$ is acyclic.
To enforce the additional requirement, the mapping of RC11 places a (control) dependency 
or a fence between every relaxed read and subsequent  relaxed write. 
It is then straightforward to define the correspondence between source (RC11) execution graphs and
target (\IMM) ones, and prove that \IMM-consistency of the target graph implies RC11-consistency 
of the source. 
This establishes the correctness of the intended mapping from RC11 without non-atomic accesses to \IMM.
Handling non-atomic accesses, which are intended to be mapped to plain machine accesses with no additional barriers or dependencies (on which \IMM generally allows $\lPO\cup\lRF$-cycles), is left for future work;
while SC accesses can be handled as mentioned above.

\section{From the promising semantics to \IMM: Relaxed fragment}
\label{sec:rlx-prom-compilation}

In the section, we outline the main ideas of the proof of the correctness of compilation from
the promising semantics of \citet{Kang-al:POPL17}, denoted by \Promise, to \IMM.
To assist the reader, we initially restrict attention to programs containing only relaxed read and write accesses.
In \cref{sec:prom-compilation}, we show how to adapt and extend our proof to the full model.

Our goal is to prove that for every outcome of a program $\prog$ (with relaxed accesses only)
under \IMM (\cref{def:allowed}), 
there exists a \Promise trace of $\prog$ terminating with the same outcome.
To do so, we introduce a traversal strategy of \IMM-consistent execution graphs,
and show, by forward simulation argument, that it can be followed by \Promise.
The main challenge in the simulation proof is due to the \emph{certification} requirement of \Promise---after 
every step, the thread that made the transition
has to show that it can run in isolation and fulfill all its so-called promises.
To address this challenge, we break our simulation argument into two parts.
First, we provide a simulation relation, which relates a \Promise thread state with 
a traversal configuration.
Second, after each traversal step, we (i) construct a \emph{certification execution graph} $\Gcert$ and
a new traversal configuration $\TCcert$;
(ii) show that the simulation relation relates $\Gcert$, $\TCcert$, and the current \Promise state; and
(iii) deduce that we can meet the certification condition by traversing $\Gcert$.
(Here, we use the fact that \Promise does not require nested certifications.)

The rest of this section is structured as follows.
In \cref{sec:rlx-promise} we describe the
fragment of \Promise restricted to relaxed accesses.
In \cref{sec:rlx-traversal} we introduce the traversal of \IMM-consistent execution graphs,
which is suitable for the relaxed fragment.
In \cref{sec:rlx-sim-cf-step} we define the simulation relation for \Promise thread steps and the execution graph traversal.
In \cref{sec:rlx-certification} we discuss how we handle certification.
Finally, in \cref{sec:rlx-prom-compl-thm} we state the compilation correctness theorem and provide its proof outline.

\subsection{The promise machine (relaxed fragment)}
\label{sec:rlx-promise}

\Promise is an operational model where threads execute in an interleaved fashion.
The \emph{machine state} is a pair $\PConf=\tup{\gts,\mem}$, where $\gts$
assigns a \emph{thread state} $\lts$ to every thread and
$\mem$ is a (global) \emph{memory}.
The memory consists of a set of \emph{messages} of the form $\msgRlx{\loc}{\val}{t}$
representing all previously executed writes, where
$\loc \in \Loc$ is the target location, $\val \in \Val$ is the stored value,
and $t \in \Q$ is the \emph{timestamp}.
The timestamps totally order the messages to each location
(this order corresponds to $G.\lCO$ in our simulation proof).

The state of each thread contains a \emph{thread view}, 
$\tcom \in \View \defeq \Loc \rightarrow \Q$, which represents the ``knowledge'' of each thread.
The view is used to forbid a thread to read from a (stale) message $\msgRlx{\loc}{\val}{t}$ if
it is aware of a newer one, \ie when $\tcom(\loc)$ is greater than $t$.
Also, it disallows to write a message to the memory with a timestamp not greater than $\tcom(\loc)$.
(Due to lack of space, we refer the reader to \citet{Kang-al:POPL17} for the full definition of thread steps.)

Besides the step-by-step execution of their programs,
threads may non-deterministically \emph{promise} future writes.
This is done by simply adding a message to the memory. 
We refer to the execution of a write instruction whose message was promised before 
as \emph{fulfilling} the promise.

The \emph{thread state} $\lts$ is a triple $\tup{\sigma,\tcom,\lprom}$,
where $\sigma$ is the thread's local state,%
\footnote{The promising semantics is generally formulated over a general labeled state transition system. 
In our development, we instantiate it with the sequential program semantics
that is used in \cref{sec:prog-to-exec} to construct execution graphs.}
$\tcom$ is the thread view, 
and $\lprom$ tracks the set of messages that were promised by the thread and not yet fulfilled.
We write 
$\lts.\lprmem$ to obtain the 
promise set of a thread state $\lts$.
Initially, each thread is in local state
$\lts_0^\tid=\tup{\sigma_0(\prog(\tid)), \lambda \loc.\; 0, \emptyset}$.

To ensure that promises do not make the semantics overly  weak,
each sequence of thread steps in \Promise has to be \emph{certified}:
the thread that took the steps should be able to fulfill all its promises when executed in isolation.
Thus, a \emph{machine step} in \Promise is given by:
\[
\inferrule{
\tup{\gts(\tid),\mem} \astep{}^+ \tup{\lts',\mem'}  \\
\exists \lts''.\; \tup{\lts',\mem'} \astep{}^* \tup{\lts'',\_} \land \lts''.\lprmem = \emptyset
}{\tup{\gts,\mem} \astep{} \tup{\gts[\tid\mapsto \lts'],\mem'}}
\]

Program outcomes under \Promise are defined as follows.

\begin{definition}
\label{def:rlx-prom-outcome}
A function $O:\Loc \to \Val$ is an \emph{outcome} of a program $\prog$ under \Promise if
$\PConf_0(\prog)\astep{}^{*} \tup{\gts, \mem}$
for some $\gts$ and $\mem$ such that 
the thread's local state in $\gts(\tid)$ is terminal for every $\tid \in \Tid$, and
for every $\loc\in\Loc$, there exists a message of the form $\msgRlx{\loc}{O(\loc)}{t} \in \mem$ 
where $t$ is maximal among timestamps of messages to $\loc$ in $\mem$.
Here, $\PConf_0(\prog)$ denotes the initial machine state, 
$\tup{\gts_{\rm init}, \mem_{\rm init}}$,
  where 
  $\gts_{\rm init} = \lambda \tid. \; \lts_0^\tid$, 
and $\mem_{\rm init} = \{\msgRlx{\loc}{0}{0} \mid \loc \in \Loc\}$.
\end{definition}

\begin{example}[Load Buffering]
  \label{ex:lb-data}
Consider the following load buffering behavior under \IMM:
\[
\inarrII{
  e_{11}: \readInst{\rlx}{a}{x} \comment{1} \\
  e_{12}: \writeInst{\rlx}{y}{1} 
}{
  e_{21}: \readInst{\rlx}{b}{y} \comment{1} \\
  e_{22}: \writeInst{\rlx}{x}{b} \\
}
\quad\vrule\quad
\inarr{\begin{tikzpicture}[yscale=0.8,xscale=1]
  \node (i11) at (0,  0) {$e_{11}: \erlab{\rlx}{x}{1}{}$};
  \node (i12) at (0, -1) {$e_{12}: \ewlab{\rlx}{y}{1}{}$};
  \node (i21) at (3,  0) {$e_{21}: \erlab{\rlx}{y}{1}{\isnotex}$};
  \node (i22) at (3, -1) {$e_{22}: \ewlab{\rlx}{x}{1}{\isnotex}$};

  \draw[po] (i11) edge (i12);
  \draw[deps] (i21) edge node[right] {\small$\lDATA$} (i22);
  \draw[rf] (i12) edge node[below] {} (i21);
  \draw[rf] (i22) edge node[above] {\small$\lRF$} (i11);
\end{tikzpicture}}
\]
The \Promise machine obtains this outcome as follows.
Starting with memory $\tup{\msgRlx{x}{0}{0}, \msgRlx{y}{0}{0}}$,
the left thread promises the message $\msgRlx{y}{1}{1}$.
After that, the right thread reads this message and executes its second instruction 
(promises a write and immediately fulfills it),
adding the the message $\msgRlx{x}{1}{1}$ to memory. 
Then, the left thread reads from that message and fulfills its promise.
Each step (including, in particular, the first promise step) could be easily ``certified''
in a thread-local execution.
Note also how the data dependency in the right thread redistrict the execution of the \Promise machine.
Due to the certification requirement, the execution cannot begin by the right thread promising $\msgRlx{x}{1}{1}$,
as it cannot generate this message by running in isolation.
\qed
\end{example}

\subsection{Traversal (relaxed fragment)}
\label{sec:rlx-traversal}

Our goal is to generate a run of \Promise for any given \IMM-consistent initialized execution graph $G$ of a program $\prog$.
To do so, we traverse $G$ with a certain strategy, deciding in each step 
whether to execute the next instruction in the program or promise a future write.
While traversing $G$, we keep track of a \emph{traversal configuration}---a 
pair $\TC=\tup{\CoveredSet, \IssuedSet}$ of subsets of $G.\lE$.
We call the events in $\CoveredSet$ and $\IssuedSet$ \emph{covered} and \emph{issued} respectively.
The covered events correspond to the instructions that were executed by \Promise, 
and the issued events corresponds to messages that were added to the memory (executed or promised stores).

Initially, we take $\TC_0 = \tup{G.\lE \cap \Init, G.\lE \cap \Init}$.
Then, at each traversal step, the covered and/or issued sets are increased,
using one of the following two steps:
$$\inferrule[\textsc{(issue)}]{
    w \in \issuable(G, \CoveredSet, \IssuedSet)
    }{
    G \vdash
    \tup{\CoveredSet, \IssuedSet} \etravStep_{\lTID(w)} \tup{\CoveredSet, \IssuedSet \uplus \{w\}}
} \qquad\qquad
\inferrule[\textsc{(cover)}]{
    e \in \coverable(G, \CoveredSet, \IssuedSet)
}{
    G \vdash 
    \tup{\CoveredSet, \IssuedSet} \etravStep_{\lTID(e)} \tup{\CoveredSet \uplus \{e\}, \IssuedSet}
}$$

The \textsc{(issue)} step adds an event $w$ to the issued set.
It corresponds to a promise step of \Promise.
We require that $w$ is issuable, which says that all the writes of other threads that it depends on 
have already been issued:
\begin{definition}
\label{def:rlx-issuable}
An event $w$ is \emph{issuable} in $G$
and $\tup{\CoveredSet, \IssuedSet}$, denoted $w\in\issuable(G, \CoveredSet, \IssuedSet)$,
if $w\in G.\lW$ and
$\dom{G.\lRFE\mathbin{;} G.\lPPO\mathbin{;}[w]} \suq \IssuedSet$.
\end{definition}

The \textsc{(cover)} step adds an event $e$ to the covered set.
It corresponds to an execution of a program instruction in \Promise.
We require that $e$ is coverable, as defined next.

\begin{definition}
\label{def:rlx-coverable}
An event $e$ is called \emph{coverable} in $G$ 
and $\tup{\CoveredSet, \IssuedSet}$, denoted $e\in \coverable(G, \CoveredSet, \IssuedSet)$, 
if $e \in G.\lE$, $\dom{G.\lPO\mathbin{;}[e]} \suq \CoveredSet$, and either
$(i)$ $e \in G.\lW \cap \IssuedSet$;
or $(ii)$ $e \in G.\lR$ and $\dom{G.\lRF\mathbin{;}[e]} \suq \IssuedSet$.
\end{definition}
The requirements in this definition are straightforward.
First, all $G.\lPO$-previous events have to be covered, \ie previous instructions have to be already executed by \Promise.
Second, if $e$ is a write event, then it has to be already issued;
and if $e$ is a read event, then the write event that $e$ reads from has to be already issued
(the corresponding message has to be available in the memory).

As an example of a traversal, consider the execution from \cref{ex:lb-data}. 
A possible traversal of the execution is the following: 
issue $e_{12}$, cover $e_{21}$, issue $e_{22}$, cover $e_{22}$, cover $e_{11}$, and cover $e_{12}$.

Starting from the initial configuration $\TC_0$,
each traversal step maintains the following invariants:
$(i)$ $\lE \cap \Init \suq \CoveredSet$;
$(ii)$ $\CoveredSet \cap G.\lW \suq \IssuedSet$; 
and $(iii)$ $\IssuedSet \suq \issuable(G, \CoveredSet, \IssuedSet)$ and $\CoveredSet \suq \coverable(G, \CoveredSet, \IssuedSet)$.
When these properties hold, we say that $\tup{\CoveredSet, \IssuedSet}$ 
is a \emph{traversal configuration} of $G$.
The next proposition ensures the existence of a traversal starting from any traversal configuration.
(A proof outline for an extended version of the traversal discussed in \cref{sec:traversal}
is presented in \citeapp{sec:trav-proof}{Appendix F}.)

\begin{proposition}
  \label{prop:rlx-trav-full}
  Let $G$ be an \IMM-consistent execution graph and $\tup{\CoveredSet, \IssuedSet}$
  be a traversal configuration of $G$.
  Then, $G \vdash\tup{\CoveredSet, \IssuedSet} \etravStep^{*} \tup{G.\lE, G.\lW}$.
\end{proposition}

\subsection{Thread step simulation (relaxed fragment)}
\label{sec:rlx-sim-cf-step}

To show that a traversal step of thread $\tid$ can be matched by a \Promise thread step,
we use a simulation relation
$\simthread_{\tid}(G, \TC, \tup{\lts, \mem}, \Tto)$,
where $G$ is an \IMM-consistent initialized full execution of $\prog$;
$\TC=\tup{\CoveredSet, \IssuedSet}$ is a traversal configuration of $G$;
$\lts=\tup{\sigma,\tcom,\lprom}$ is $\tid$'s thread state in \Promise;
 $\mem$ is the memory of \Promise;
and $\Tto : \IssuedSet \rightarrow \Q$ is a function that assigns timestamps to issued writes.
The relation $\simthread_{\tid}(G, \TC, \tup{\lts, \mem}, \Tto)$ holds 
if the following conditions are met (for conciseness we omit the ``$G.$'' prefix):
\begin{enumerate}
\item $\Tto$ agrees with $\lCO$:
\begin{itemize}
\item $\forall w \in \lE \cap \Init. \; \Tto(w) = 0$
\item $\forall \tup{w, w'} \in [\IssuedSet] \mathbin{;}\lCO \mathbin{;} [\IssuedSet]. \; \Tto(w) \le \Tto(w')$
\end{itemize}

\item Non-initialization messages in $\mem$ have counterparts in $\IssuedSet$:
\begin{itemize}
\item $\forall \msgRlx{\loc}{\_}{t} \in \mem. \; t \neq 0 \implies \exists w \in \IssuedSet. \; \lLOC(w) = \loc \land \Tto(w) = t$
\end{itemize}

\item 
Issued events have corresponding messages in memory:
\begin{itemize}
\item $\forall w \in \IssuedSet. \;
      \msgRlx{\lLOC(w)}{\lVAL(w)}{\Tto(w)} \in \mem$
\end{itemize}
\item 
For every promise, there exists a corresponding issued uncovered event $w$:
\begin{itemize}
\item $      \forall \msgRlx{\loc}{\val}{t} \in \lprom. \;
        \exists w \in \lE_{\tid} \cap \IssuedSet \setminus \CoveredSet. \;
     \lLOC(w) = \loc \land \lVAL(w) = \val \land \Tto(w) = t$
\end{itemize}
\item  Every issued uncovered event $w$ of thread $\tid$ has a corresponding promise in $\lprom$.
\begin{itemize}
\item $ \forall w \in \lE_{\tid} \cap \IssuedSet \setminus \CoveredSet.\;
  \msgRlx{\lLOC(w)}{\lVAL(w)}{\Tto(w)} \in \lprom$ 
\end{itemize}

\item The view $\tcom$ is justified by graph paths:
\begin{itemize}
\item $\tcom = \lambda \loc.\; \max\Tto[\lW(\loc) \cap \dom{\lURR_{\rlx} \mathbin{;} [\lE\tidmod{\tid}\cap \CoveredSet]}]$
  where $\lURR_{\rlx}\defeq \lRF^?; \lPO^?$
\end{itemize}

\item 
The thread local state $\sigma$ matches the covered events ($\sigma.\lG.\lE = C \cap \lE\tidmod{\tid}$),
and can always reach the execution graph $G$ 
($\exists \sigma'. \; \sigma \to_{\tid}^* \sigma' \land \sigma'.\lG = G\rst{\tid}$).
\end{enumerate}

\begin{proposition}
\label{prop:rlx-sim-step}
  If $\simthread_{\tid}(G, \TC, \tup{\lts, \mem}, \Tto)$ and 
  $G \vdash \TC \etravStep_{\tid} \TC'$ hold,
  then there exist $\lts'$, $\mem'$, $\Tto'$ such that
$\tup{\lts,\mem} \astep{} \tup{\lts',\mem'}$
and $\simthread_{\tid}(G, \TC', \tup{\lts', \mem'}, \Tto')$ hold.
\end{proposition}

In addition, it is easy to verify that the initial states are related, \ie  
$\simthread_{\tid}(G, \TC_0, \tup{\lts_0^\tid, \mem_{\rm init}},\bot)$ holds
for every $\tid\in\Tid$.

\subsection{Certification (relaxed fragment)}
\label{sec:rlx-certification}

\begin{figure}[t]
\centering
{\footnotesize
$\inarrII{
  \readInst{\rlx}{r_1}{\loc} \comment{1}  \\
  \writeInst{\rlx}{\locy}{r_1}  \\
  \writeInst{\rlx}{\loc}{2}  \\
}{
  \writeInst{\rlx}{\loc}{1}  \\
  \readInst{\rlx}{r_2}{\locy} \comment{1} \\
  \readInst{\rlx}{r_3}{\loc} \comment{2} \\
  \writeInst{\rlx}{\locz}{r_2}  \\
  \writeInst{\rlx}{\loc}{3}  \\
}$}$
\hfill\vrule\hfill
\inarr{\begin{tikzpicture}[scale=0.8, every node/.style={transform shape}]
  \node (i11) at (0, 0) {$\ee{11}\colon \erlab{\rlx}{\loc}{1}{\isnotex}$};
  \node (i12) at (0, -1.2) {$\ee{12}\colon \ewlab{\rlx}{\locy}{1}{}$};
  \node (i13) at (0, -2.4) {$\ee{13}\colon \ewlab{\rlx}{\loc}{2}{}$};

  \node (i21) at (3.5,  0) {$\ee{21}\colon \ewlab{\rlx}{\loc}{1}{}$};
  \node (i22) at (3.5, -1.2) {$\ee{22}\colon \erlab{\rlx}{\locy}{1}{\isnotex}$};
  \node (i23) at (3.5, -2.4) {$\ee{23}\colon \erlab{\rlx}{\loc}{2}{\isnotex}$};
  \node (i24) at (3.5, -3.6) {$\ee{24}\colon \ewlab{\rlx}{\locz}{1}{}$};
  \node (i25) at (3.5, -4.8) {$\ee{25}\colon \ewlab{\rlx}{\loc}{3}{}$};
  \node (hh) at (1.5, -6) {$\inarrC{\text{An execution graph $G$ and}\\ \text{its traversal configuration $\tup{\CoveredSet,\IssuedSet}$}}$};

  \begin{scope}[on background layer]
     \issuedBox{i24};
     \issuedBox{i12};

     \issuedCoveredBox{i21};
  \end{scope}

  \draw[rf] (i21) edge node[above] {\small$\lRFE$} (i11);
  \draw[deps] (i11) edge node[left] {\small$\lDEPS$} (i12);
  \draw[rf] (i12) edge node[above] {\small$\lRFE$} (i22);
  \draw[rf] (i13) edge node[above] {\small$\lRFE$} (i23);

  \draw[deps,out=230,in=130] (i22) edge node[pos=.2,left] {\small$\lDEPS$} (i24);

  \draw[po] (i12) edge (i13);
  \draw[po] (i21) edge (i22);
  \draw[po] (i22) edge (i23);
  \draw[po] (i23) edge (i24);
  \draw[po] (i24) edge (i25);
\end{tikzpicture}}
\hfill\vrule\hfill
\inarr{\begin{tikzpicture}[scale=0.8, every node/.style={transform shape}]
  \node (i12) at (0, -1.2) {$\ee{12}\colon \ewlab{\rlx}{\locy}{1}{}$};

  \node (i21) at (3.5,  0) {$\ee{21}\colon \ewlab{\rlx}{\loc}{1}{}$};
  \node (i22) at (3.5, -1.2) {$\ee{22}\colon \erlab{\rlx}{\locy}{1}{\isnotex}$};
  \node (i23) at (3.5, -2.4) {$\ee{23}\colon \erlab{\rlx}{\loc}{1}{\isnotex}$};
  \node (i24) at (3.5, -3.6) {$\ee{24}\colon \ewlab{\rlx}{\locz}{1}{}$};
  \node (i25) at (3.5, -4.8) {\phantom{$\ee{25}\colon \ewlab{\rlx}{\loc}{3}{}$}};
  \node (hh) at (1.7, -6) {$\inarrC{\text{The  certification graph $\Gcert$ and}\\ \text{its traversal configuration $\tup{\CoveredSetcert,\IssuedSetcert}$}}$};

  \begin{scope}[on background layer]
     \issuedBox{i24};

     \issuedCoveredBox{i12};
     \issuedCoveredBox{i21};
  \end{scope}

  \draw[rf] (i12) edge node[above] {\small$\lRFE$} (i22);
  \draw[rf,out=200,in=180] (i21) edge node[pos=.8,left] {\small$\lRFI$} (i23);

  \draw[deps,out=230,in=130] (i22) edge node[pos=.2,left] {\small$\lDEPS$} (i24);

  \draw[po] (i21) edge (i22);
  \draw[po] (i22) edge (i23);
  \draw[po] (i23) edge (i24);
\end{tikzpicture}}$
\vspace{-10pt}
\caption{A program, its execution graph, and a related certification graph.
Covered events are marked by
{\protect\tikz \protect\draw[coveredStyle] (0,0) rectangle ++(0.35,0.35);}
and issued ones by 
{\protect\tikz \protect\draw[issuedStyle] (0,0) rectangle ++(0.35,0.35);}.}\label{fig:prog-gr-cert}
\end{figure}

To show that a traversal step can be simulated by \Promise, \cref{prop:rlx-sim-step} does not suffice:
the machine step requires the new thread's state to be certified.
To understand how we construct a certification run, consider the example in \cref{fig:prog-gr-cert}.
Suppose that $\simthread_{\tid_2}$ holds for $G,\tup{\CoveredSet, \IssuedSet},\tup{\lts,\mem},\Tto$
(where $\tid_2$ is the identifier of the second thread).
Consider a possible certification run for $\tid_2$.
According to $\simthread_{\tid_2}$, there is one unfulfilled promise of $\tid_2$, 
\ie $\lts.\lprmem = \{\msgRlx{\locz}{1}{\Tto(\ee{24})}\}$.
We also know that $\tid_2$ has executed all instructions
up to the one related to the last covered event $\ee{21}$.
To fulfill the promise, it has to execute the instructions corresponding to $\ee{22}$, $\ee{23}$, and $\ee{24}$.

To construct the certification run, we (inductively) apply a version of \cref{prop:rlx-sim-step} for certification steps,
starting from a sequence of traversal steps of $\tid_2$ that cover $\ee{22}$, $\ee{23}$, and $\ee{24}$.
For $G$ and $\tup{\CoveredSet, \IssuedSet}$, there is no such sequence:
we cannot cover $\ee{23}$ without issuing $\ee{13}$ first 
(which we cannot do since only one thread may run during certification).
Nevertheless, observing that the value read at $\ee{23}$ is immaterial for covering $\ee{24}$,
we may use a different execution graph for this run, namely $\Gcert$ shown in \cref{fig:prog-gr-cert}.
Thus, in $\Gcert$ we redirect $\ee{23}$'s incoming reads-from edge and change its value accordingly.
In contrast, we do not need to change $\ee{22}$'s incoming reads-from edge
because the condition about $G.\lPPO$ in the definition of issuable events ensures that 
$\ee{12}$ must have already been issued.
For $\Gcert$ and $\tup{\CoveredSetcert, \IssuedSetcert}$,
there exists a sequence of traversal steps  that cover $\ee{22}$, $\ee{23}$, and $\ee{24}$.
Since events of other threads have all been made covered in $\CoveredSetcert$,
we know that only $\tid_2$ will take steps in this sequence.

Generally speaking, for a given $\tid\in\Tid$ whose step has to be certified, 
our goal is to construct a certification graph $\Gcert$ 
and a traversal configuration $\TCcert=\tup{\CoveredSetcert, \IssuedSetcert}$ of $\Gcert$ such that
(1) $\Gcert$ is \IMM-consistent (so we can apply \cref{prop:rlx-trav-full} to it)
and (2) we can simulate its traversal in \Promise to obtain the certification run for thread $\tid$.
In particular, the latter requires that $\Gcert\rst{\tid}$ is an execution graph of $\tid$'s program.
In the rest of this section, we present this construction and show how it is used to certify \Promise's steps (\cref{prop:rlx-ext-sim-step}).

First, the events of $\Gcert$ are given by 
$\Gcert.\lE  \defeq \CoveredSet \cup \IssuedSet \cup \dom{G.\lPO\mathbin{;}[\IssuedSet\cap G.\lE\tidmod{\tid}]}$.
They consist of the covered and issued events and all $\lPO$-preceding events of issued events in thread $\tid$.
The $\lCO$ and dependency components of $\Gcert$ are the same as in (restricted) $G$
($\Gcert.\lX = [\Gcert.\lE]\mathbin{;} G.\lX\mathbin{;} [\Gcert.\lE]$ for $\lX \in \set{\lCO,\lADDR,\lDATA,\lCTRL,\lRMWDEP}$).
As we saw on \cref{fig:prog-gr-cert},
we may need to modify the $\lRF$ edges of the certification graph (and, consequentially, labels of events).
In the example, it was required because the source of an $\lRF$ edge was not present in $\Gcert$.
The relation $\Gcert.\lRF$ is defined as follows:
\[\inarr{
\Gcert.\lRF \defeq G.\lRF\mathbin{;} [\D] \cup
  \bigcup_{\loc \in \Loc} ([G.\lW(\loc)] \mathbin{;}\lBVFrlx \mathbin{;}
         [G.\lR(\loc) \cap \Gcert.\lE \setminus  \D]
    \setminus G.\lCO \mathbin{;} G.\lBVFrlx)
\\
\text{~~where}~~
\D=\Gcert.\lE \cap (\CoveredSet \cup \IssuedSet \cup G.\lE\tidmod{\neq\tid} \cup
 \dom{G.\lRFI^? \mathbin{;} G.\lPPO \mathbin{;} [\IssuedSet]})
  \text{~~and} \\
\qquad \quad
G.\lBVFrlx = (G.\lRF \mathbin{;} [\D])^? \mathbin{;} G.\lPO
}\]

The set $D$ represents the \emph{determined} events, whose $\lRF$ edges are preserved.
Intuitively, for a read event $r$ with location $\loc$,
the set $\dom{[G.\ewlab{}{\loc}{}{}]\mathbin{;} G.\lBVFrlx\mathbin{;} [r]}$ consists of writes to $\loc$
that are ``observed'' by $\lTID(r)$ at the moment it ``executes'' $r$.
If $r$ is not determined, we choose the new $\lRF$ edge to $r$ to be from the $\lCO$-latest write in this set.
Thus, in the certification graph, $r$ is not reading a stale value,
and its incoming $\lRF$ edge does not increase the set of ``observed'' writes in thread $\tid$.

The labels (which include the read values) in $\Gcert$ have to be modified as well, to match the new $\lRF$ edges.
To construct of $\Gcert.\lLAB$, we leverage a certain \emph{receptiveness} property of the operational semantics in \cref{fig:prog-to-exec}.
Roughly speaking, we show that if $\tup{\sprog,pc,\Phi,G,\PhiD,S} \rightarrow^+_{\tid} \tup{\sprog,pc',\Phi',G',\PhiD',S'}$,
then for every read $r \in G'.\lE \setminus (G.\lE \cup \dom{G'.\lCTRL})$ and value $\val$, 
there exist $pc''$, $\Phi''$, $G''$, $\PhiD''$, and $S''$ such that
$\tup{\sprog,pc,\Phi,G,\PhiD,S} \rightarrow^+_{\tid} \tup{\sprog,pc'',\Phi'',G'',\PhiD'',S''}$,
$G''.\lVAL(r) = \val$, and $G''$ is identical to $G'$ except (possibly) for values of events that depend on $r$.%
\footnote{The full formulation of the receptiveness property is more elaborate.
  Due to the lack of space, we refer the reader to our Coq development~\url{https://github.com/weakmemory/imm}.}
Applying this property inductively, we construct the labeling function $\Gcert.\lLAB$.

\smallskip
This concludes the construction of $\Gcert$.
Now, we start the traversal from $\TCcert=\tup{\CoveredSetcert,\IssuedSetcert}$
where $\CoveredSetcert\defeq\CoveredSet \cup \Gcert.\lE\tidmod{\neq\tid}$ and $\IssuedSetcert\defeq\IssuedSet$.
Thus, we take all events of other threads to be covered so that the traversal of $\Gcert$ may only include steps of thread $\tid$.
To be able to reuse \cref{prop:rlx-trav-full}, we prove the following proposition.

\begin{proposition}
\label{prop:rlx-cert-cons}
Let $G$ be an \IMM-consistent execution graph, and $\TC=\tup{\CoveredSet,\IssuedSet}$ a traversal configuration of $G$.
Then, $\Gcert$ is \IMM-consistent and $\TCcert$ is a traversal configuration of $\Gcert$.
\end{proposition}

For the full model (see \cref{sec:certification}), we will have to introduce
a slightly modified version of the simulation relation for certification.
For the relaxed fragment that we consider here, however, we use the same relation defined in \cref{sec:rlx-sim-cf-step} 
and prove that it holds for the constructed certification graph:

\begin{proposition}
\label{prop:rlx-cert-graph}
Suppose that $\simthread_{\tid}(G, \TC, \tup{\lts, \mem}, \Tto)$ holds.
Then
$\simthread_{\tid}(\Gcert,\TCcert, \tup{\lts, \mem}, \Tto)$ holds.
\end{proposition}

Putting \cref{prop:rlx-trav-full,prop:rlx-sim-step,prop:rlx-cert-cons,prop:rlx-sim-step,prop:rlx-cert-graph} together,
we derive the following strengthened version of \cref{prop:rlx-sim-step}, 
which additionally states that the new \Promise thread's state is certifiable.

\begin{proposition}
\label{prop:rlx-ext-sim-step}
  If $\simthread_{\tid}(G, \TC, \tup{\lts, \mem}, \Tto)$ and 
  $G \vdash \TC \etravStep_{\tid} \TC'$ hold,
  then there exist $\lts', \mem',\Tto'$ such that
$\tup{\lts,\mem} \astep{}^+ \tup{\lts',\mem'}$ and $\simthread_{\tid}(G, \TC', \tup{\lts', \mem'}, \Tto')$ hold,
and
there exist $\lts'',\mem''$ such that $\tup{\lts',\mem'} \astep{}^* \tup{\lts'',\mem''}$ and
$\lts''.\lprmem = \emptyset$.
\end{proposition}
\begin{proof}[Proof outline]
By \cref{prop:rlx-sim-step}, there exist
  $\lts', \mem'$, and $\Tto'$ such that
  $\tup{\lts,\mem} \astep{}^+ \tup{\lts', \mem'}$ and 
  $\simthread_{\tid}(G, \TC', \tup{\lts', \mem'}, \Tto')$ hold.
  By \cref{prop:rlx-cert-graph}, 
  $\simthread_{\tid}(\Gcert, \TCcert, \tup{\lts', \mem'}, \Tto')$ holds.
  By \cref{prop:rlx-trav-full,prop:rlx-cert-cons}, we have 
  $\Gcert \vdash \TCcert \etravStep_{\tid}^* \tup{\Gcert.\lE, \Gcert.\lW}$.
  We inductively apply \cref{prop:rlx-sim-step} to obtain 
  $\tup{\lts'', \mem''}$ and $\Tto''$
  such that $\tup{\lts', \mem''} \astep{}^* \tup{\lts'', \mem''}$
  and $\simthread_{\tid}(\Gcert, \tup{\Gcert.\lE, \Gcert.\lW}, \tup{\lts'', \mem''}, \Tto'')$ hold.
  From the latter, it follows that $\lts''.\lprmem = \emptyset$.
\end{proof}

\subsection{Compilation correctness theorem (relaxed fragment)}
\label{sec:rlx-prom-compl-thm}

\begin{theorem}
\label{thm:rlx-prom-full}
Let $\prog$ be a program with only relaxed reads and relaxed writes.
Then, every outcome  of $\prog$ under \IMM (\cref{def:outcome})
is also an outcome of $\prog$ under \Promise (\cref{def:rlx-prom-outcome}).
\end{theorem}
\begin{proof}[Proof outline]
We introduce a simulation relation $\simfull$ on traversal configurations and \Promise states:
\[
\simfull(G, \TC, \tup{\gts,\mem}, \Tto) \defeq
  \forall \tid \in \Tid. \; \simthread_{\tid}(G, \TC, \tup{\gts(\tid),\mem}, \Tto)
\]
We show that $\simfull$ holds for an \IMM-consistent execution graph $G$, which has the outcome $O$, of the program $\prog$,
its initial traversal configuration, 
the initial \Promise state $\PConf_0(\prog)$, and the initial timestamp mapping $\Tto=\bot$.
Then, we inductively apply \cref{prop:rlx-ext-sim-step} on a traversal
$G\vdash \tup{G.\lE \cap \Init, G.\lE \cap \Init} \etravStep^* \tup{G.\lE, G.\lW}$, which exists by \cref{prop:rlx-trav-full},
and additionally show that at every step $\simthread_{\tid}$ holds for every thread $\tid$ that did not take the step.
Thus, we obtain a \Promise state $\PConf$ and a timestamp function $\Tto$ such that
$\PConf_0(\prog) \astep{}^* \PConf$
and $\simfull(G, \tup{G.\lE, G.\lW}, \PConf, \Tto)$ hold.
From the latter, it follows that $O$ is an outcome of $\prog$ under \Promise.
\end{proof}

\section{From the promising semantics to \IMM: The general case}
\label{sec:prom-compilation}

In the section, we extend the result of \cref{sec:rlx-prom-compilation} to the full 
\Promise model.
Recall that, due to the limitation of \Promise discussed in \cref{ex:strong-rmw}, 
we assume that all RMWs are ``strong''.

\begin{theorem}
\label{thm:prom-full}
Let $\prog$ be a program in which all RMWs are ``strong''.
Then, every outcome  of $\prog$ under \IMM 
is also an outcome of $\prog$ under \Promise.
\end{theorem}

To prove this theorem, we find it technically convenient
to use a slightly modified version of \IMM, which is (provably) weaker.
In this version, we use the simplified synchronization relation $G.\lSWs$ (see \cref{rem:rs}),
as well as a \emph{total order on SC fences}, $G.\lSC$, which we include as another basic component of execution graphs.
Then, we include $G.\lSC$ in $G.\lAR$  instead of $G.\lPSC$ (see \cref{sec:global}), and require that
$G.\lSC ; G.\lHB ; (G.\lECO; G.\lHB)^?$ is irreflexive (to ensure that $G.\lPSC \suq G.\lSC$).
It is easy to show that the latter modification results in an equivalent model,
while the use of $G.\lSWs$ makes this semantics only weaker than \IMM.
The $G.\lSC$ relation facilitates the construction of a run of \Promise,
as it fully determines the order in which SC fences should be executed.

The rest of this section is structured as follows.
In \cref{sec:promise} we briefly introduce the full \Promise model.
In \cref{sec:traversal} we introduce more elaborate traversal of \IMM execution graphs,
which might be followed by the full \Promise model.
In \cref{sec:sim-cf-step} we define the simulation relation for the full model.
In \cref{sec:certification} we discuss how certification graphs are adapted for the full model.

\subsection{The full promise machine}
\label{sec:promise}

In the full \Promise model, the machine state is a triple $\PConf=\tup{\gts,\gsco,\mem}$.
The additional component $\gsco \in \View $ is a (global) \emph{SC view}.
Messages in the memory are of the form $\msg{\loc}{\val}{f}{t}{\view}$,
where, comparing to the version from \cref{sec:rlx-promise},
(i) a timestamp $t$ is extended to a \emph{timestamp interval} $(f,t] \in \Q \times \Q$
satisfying $f < t$ or $f = t = 0$ (for initialization messages)
and (ii) the additional component $\view \in \View$ is the \emph{message view}.\footnote{
The order $\leq$ on $\Q$ is extended pointwise to order $\Loc \rightarrow \Q$.
$\bot$ and $\sqcup$ denote the natural bottom element and join
operations (pointwise extensions of the
initial timestamp $0$ and the $\max$ operation on timestamps).
$[\loc_1@t_1 \til \loc_n@t_n]$ denotes the function assigning $t_i$ to $\loc_i$ and $0$ to other locations.}
Messages to the same location should have disjoint timestamp intervals, 
and thus the intervals totally order the messages to each location.
The use of intervals allows one to express the fact that two
messages are adjacent (corresponding to $\imm{G.\lCO}$),
which is required to enforce the RMW atomicity condition (\cref{sec:model-atomicity}).

Message views represent the ``knowledge'' carried by the message
that is acquired by threads reading this message (if they use an acquire read or fence).
In turn, the thread view $\tcom$ is now a triple
$\tup{\viewCur, \viewAcq, \viewRel} \in \View \times \View \times (\Loc \rightarrow \View)$,
whose components are called the \emph{current}, \emph{acquire}, and \emph{release} views.
The different thread steps (for the different program instructions)
constrain the three components of the thread view with the timestamps and message views that are included 
in the messages that the thread reads and writes,
as well as with the global SC view $\gsco\in \View$.
These constraints are tailored to precisely enforce 
the coherence and RMW atomicity properties (\cref{sec:model-coherence},\cref{sec:model-atomicity}),  
as well as the global synchronization provided by SC fences.
(Again, we refer the reader to \citet{Kang-al:POPL17} for the full definition of thread steps.)


Apart from promising messages,
our proof utilizes another non-deterministic step of \Promise,
which allows a thread to \emph{split} its promised messages,
\ie to replace its promise $\msg{\loc}{\val}{f}{t}{\view}$
with two promises $\msg{\loc}{\val'}{f}{t'}{\view'}$ and $\msg{\loc}{\val}{t'}{t}{\view}$
provided that $f < t' < t$.

In the full \Promise model, the certification requirement is stronger than the one presented in \cref{sec:rlx-prom-compilation} for the relaxed fragment.
Due to possible interference of other threads before the current thread fulfills its promises,
certification is required for every possible \emph{future memory} and \emph{future SC view}.
Thus, a \emph{machine step} in \Promise is given by:
\[
\inferrule{
\tup{\gts(\tid),\gsco,\mem} \astep{}^+ \tup{\lts',\gsco',\mem'}  \\
\forall \omem\supseteq\mem',\ogsco \geq \gsco'.\; \exists \lts''.\; \tup{\lts',\ogsco,\omem} \astep{}^* \tup{\lts'',\_,\_} \land \lts''.\lprmem = \emptyset
}{\tup{\gts,\gsco,\mem} \astep{} \tup{\gts[\tid\mapsto \lts'],\gsco',\mem'}}
\]



\begin{example}
We revisit the program presented in \cref{ex:rfi-not-preserved}.
To get the intended behavior in \Promise,
thread I starts by promising a message $\msg{z}{1}{1}{2}{[z@2]}$.
It may certify the promise since its fourth instruction does not depend on $a$ and
the thread may read $1$ from $y$ when executing the third instruction in any future memory.
After the promise is added to memory, thread II reads it and writes
$\msg{x}{1}{1}{2}{[x@2]}$ to the memory. 
Then, thread I reads from this message, executes its remaining instructions, and fulfills its promise.
\qed
\end{example}

\begin{remark}
In \Promise, the notion of future memory is broader---a future memory
may be obtained by a sequence of memory modifications including message additions,
message splits and lowering of message views.
In our Coq development, we show that it suffices to consider only future memories that are obtained by adding messages
(\citeapp{sec:futurememory}{Appendix E} outlines the proof of this claim).
\end{remark}

\begin{remark}
What we outline here ignores \Promise's \emph{plain} accesses.
These are weaker than relaxed accesses (they only provide partial coherence),
and are not needed for accounting for \IMM's behaviors.
Put differently, one may assume that the compilation from \Promise to \IMM
first strengthens all plain access modes to relaxed.
The correctness of compilation then follows from the soundness of
this strengthening (which was proved by \citet{Kang-al:POPL17})
and our result that excludes plain accesses. 
\end{remark}

\subsection{Traversal}
\label{sec:traversal}

To support all features of \IMM and \Promise models, we have to complicate the traversal considered
in \cref{sec:rlx-traversal}.
We do it by introducing two new traversal steps (see \cref{fig:traversal-steps})
and modifying the definitions of issuable and coverable events.

The \textsc{(release-cover)} step is introduced because the \Promise model forbids to
promise a release write without fulfilling it immediately.
It adds a release write to both the covered and issued sets in a single step.
Its precondition is simple: all $G.\lPO$-previous events have to be covered.

The \textsc{(rmw-cover)} step reflects that RMWs in \Promise are performed in one atomic step,
even though they are split to two events in \IMM.
Accordingly, when traversing $G$, we require to cover the write part of $\lRMW$ edges
immediately after their read part.
If the write is release, then, again since release writes cannot be promised without immediate fulfillment,
it is issued in the same step.

\begin{figure}[t]
\small
\begin{mathpar}
\inferrule[\textsc{(issue)}]{
    w \in \issuable(G, \CoveredSet, \IssuedSet) \\
    w \nin G.\lW^{\rel}
    }{
    G \vdash
    \tup{\CoveredSet, \IssuedSet} \etravStep_{\lTID(w)} \tup{\CoveredSet, \IssuedSet \uplus \{w\}}
} \and
\inferrule[\textsc{(cover)}]{
    e \in \coverable(G, \CoveredSet, \IssuedSet) \\
    e \nin \dom{G.\lRMW}
}{
    G \vdash 
    \tup{\CoveredSet, \IssuedSet} \etravStep_{\lTID(e)} \tup{\CoveredSet \uplus \{e\}, \IssuedSet}
} \and
 \inferrule[\textsc{(release-cover)}]{
    \dom{G.\lPO\mathbin{;}[w]}\suq \CoveredSet \\
     w \in G.\lW^{\rel}
     }{
     G \vdash
     \tup{\CoveredSet, \IssuedSet} \etravStep_{\lTID(w)} \tup{\CoveredSet \uplus \{w\}, \IssuedSet \uplus \{w\}}
 } \and
\inferrule[\textsc{(rmw-cover)}]{
    r \in \coverable(G, \CoveredSet, \IssuedSet) \\
    \tup{r,w} \in G.\lRMW \\\\
     (w \in \IssuedSet \land \IssuedSet' = \IssuedSet) \lor (w \in G.\lW^{\rel} \land \IssuedSet' =\IssuedSet \uplus \{w\})
}{
    G \vdash 
    \tup{\CoveredSet, \IssuedSet} \etravStep_{\lTID(r)} \tup{\CoveredSet \uplus \{r, w\}, \IssuedSet'}
}
\end{mathpar}
\caption{Traversal steps.}
\label{fig:traversal-steps}
\end{figure}

The full definition of issuable event has additional requirements.
\begin{definition}
\label{def:issuable}
An event $w$ is \emph{issuable} in $G$
and $\tup{\CoveredSet, \IssuedSet}$, denoted $w\in\issuable(G, \CoveredSet, \IssuedSet)$, if $w\in G.\lW$ and the following hold:
      \begin{itemize}
        \item $\dom{([G.\lW^{\rel}]\mathbin{;}G.\lPO\rst{G.\lLOC} \cup [G.\lF]\mathbin{;}G.\lPO)\mathbin{;} [w]} \suq \CoveredSet$
          \labelAxiom{fwbob-cov}{req:fwbob-cov}
        \item $\dom{(G.\lDETOUR \cup G.\lRFE)\mathbin{;} G.\lPPO\mathbin{;}[w]} \suq \IssuedSet$
          \labelAxiom{ppo-iss}{req:ppo-iss}
        \item $\dom{(G.\lDETOUR \cup G.\lRFE)\mathbin{;} [G.\erlab{\acq}{}{}{}]\mathbin{;}G.\lPO\mathbin{;}[w]} \suq \IssuedSet$
          \labelAxiom{acq-iss}{req:acq-iss}
        \item $\dom{[G.\lWstrong]\mathbin{;}G.\lPO\mathbin{;}[w]} \suq \IssuedSet$
          \labelAxiom{w-strong-iss}{req:w-strong-iss}
      \end{itemize}
\end{definition}
The \ref{req:ppo-iss} condition extends the condition from \cref{def:rlx-issuable}.
The \ref{req:fwbob-cov} condition arises from \Promise's restrictions on promises:
a release write cannot be executed if the thread has an unfulfilled promise to the same location,
and a release fence cannot be executed if the thread has any unfulfilled promise.
Accordingly, we require that when $w$ is issued 
$G.\lPO$-previous release writes to the same location and release fences have already been covered.
Note that we actually require this from all $G.\lPO$-previous fences (rather than just release ones).
This is not dictated by \Promise, but simplifies our proofs.
Thus, our proof implies that compilation from \Promise to \IMM remains correct even if 
acquire fences ``block'' promises as release ones.
The other conditions in \cref{def:issuable} are forced by \Promise's certification,
as demonstrated by the following examples.

\begin{figure}[t]\small\centering 
$\inarrII{
  \ee{11}: \writeInst{\rlx}{x}{2}  \\
}{
  \ee{21}: \writeInst{\rlx}{x}{1}  \\
  \ee{22}: \readInst{\rlx}{a}{x} \comment{2} \\
  \ee{23}: \writeInst{\rlx}{y}{a}  \\
}
\quad\vrule\quad
\inarr{\begin{tikzpicture}[yscale=0.8,xscale=1]
  \node (i11) at (0,  -0.5) {$\ee{11}: \ewlab{\rlx}{\loc}{2}{}$};
  \node (i21) at (3,  0) {$\ee{21}: \ewlab{\rlx}{\loc}{1}{}$};
  \node (i22) at (3, -1) {$\ee{22}: \erlab{\rlx}{\loc}{2}{\isnotex}$};
  \node (i23) at (3, -2) {$\ee{23}: \ewlab{\rlx}{\locy}{2}{}$};
  \draw[detour] (i21) edge node[right] {\smaller\smaller$\lDETOUR$} (i22);
  \draw[mo] (i21) edge node[above] {\smaller\smaller$\lCOE$} (i11);
  \draw[rf] (i11) edge node[below] {\smaller\smaller$\lRFE$} (i22);
  \draw[deps] (i22) edge node[right] {\smaller\smaller$\lDEPS$} (i23);
\end{tikzpicture}}$
\caption{Demonstration of the necessity of \ref{req:ppo-iss} in the definition of $\issuable$.}
\label{fig:ex-ppo-iss}
\end{figure}
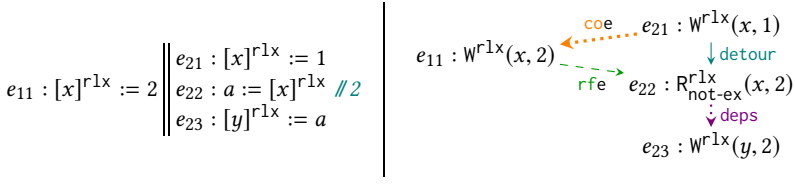

\begin{example}
Consider the program and its execution graph on \cref{fig:ex-ppo-iss}.
To certify a promise of a message that corresponds to $\ee{23}$,
we need to be able to read the value $2$ for $x$ in $\ee{22}$
(as $\ee{23}$ depends on this value).
Thus, the message that corresponds to $\ee{11}$ has to be in memory already, \ie
the event $\ee{11}$ has to be already issued.
This justifies the $G.\lRFE\mathbin{;} G.\lPPO$ part of \ref{req:ppo-iss}.
The justification for the $G.\lDETOUR\mathbin{;} G.\lPPO$  part of \ref{req:ppo-iss} is related to the requirement of certification \emph{for every future memory}.
Indeed, in the same example, it is also required that $\ee{21}$ was issued before $\ee{23}$:
We know that $\ee{23}$ is issued after $\ee{11}$, and thus, 
there is a message of the form $\msg{\loc}{2}{f_{\ee{11}}}{t_{\ee{11}}}{\_}$ in the memory.
Had $\ee{21}$ not been issued before, 
the instruction $\ee{21}$ would have to add a message of the form $\msg{\loc}{1}{f_{\ee{21}}}{t_{\ee{21}}}{\_}$ to the memory
during certification.
Because $\ee{22}$ has to read from $\msg{\loc}{2}{f_{\ee{11}}}{t_{\ee{11}}}{\_}$, the timestamp $t_{\ee{21}}$ has to be smaller than $t_{\ee{11}}$.
However, an arbitrary future memory might not have free timestamps in $(0, f_{\ee{11}}]$.
\qed
\end{example}

\begin{figure}[t]\centering 
\small
$\inarrC{
\inarrIV{
  \ee{11}: \writeInst{\rlx}{x}{3}
}{
  \ee{21}: \writeInst{\rlx}{y}{2} \\
  \ee{22}: \writeInst{\rel}{x}{2}
}{
  \ee{31}: \readInst{\rlx}{a}{x} \comment{2} \\
  \ee{32}: \writeInst{\rel}{z}{2}
}{
  \ee{41}: \readInst{\acq}{b}{z} \comment{2} \\
  \ee{42}: \readInst{\acq}{c}{x} \comment{3} \\
  \ee{43}: \writeInst{\rlx}{y}{1}
} \\
\\
\hline
\\
\inarr{\begin{tikzpicture}[yscale=0.8,xscale=1]
  \node (i11) at (0,  -0.5) {$\ee{11}: \ewlab{\rlx}{\loc}{3}{}$};

  \node (i21) at (3.5,     0) {$\ee{21}: \ewlab{\rlx}{\locy}{2}{}$};
  \node (i22) at (3.5,    -1) {$\ee{22}: \ewlab{\rel}{\loc}{2}{}$};

  \node (i31) at (7,     0) {$\ee{31}: \erlab{\rlx}{\loc}{2}{\isnotex}$};
  \node (i32) at (7,    -1) {$\ee{32}: \ewlab{\rel}{\locz}{2}{}$};

  \node (i41) at (10.5,   0.5) {$\ee{41}: \erlab{\acq}{\locz}{2}{\isnotex}$};
  \node (i42) at (10.5,  -0.5) {$\ee{42}: \erlab{\acq}{\loc}{3}{\isnotex}$};
  \node (i43) at (10.5,  -1.5) {$\ee{43}: \ewlab{\rlx}{\locy}{1}{}$};

  \draw[mo] (i22) edge node[below] {\smaller\smaller$\lCOE$} (i11);
  \draw[po] (i21) edge (i22);
  \draw[rf] (i22) edge node[below] {\smaller\smaller$\lRFE$} (i31);
  \draw[po] (i31) edge (i32);
  \draw[rf] (i32) edge node[below] {\smaller\smaller$\lRFE$} (i41);
  \draw[po] (i41) edge (i42);
  \draw[po] (i42) edge (i43);
  \draw[mo,out=190,in=-40] (i43) edge node[pos=.1,above] {\smaller\smaller$\lCOE$} (i21);
  \draw[rf,out=20,in=160] (i11) edge node[below] {\smaller\smaller$\lRFE$} (i42);

  \begin{scope}[on background layer]
     \coveredBox{i41};

     \draw[coveredStyle] ($(i21)  + (-1.3,0.5)$) rectangle ++(2.6,-2.0);
     \issuedBox{i21};
     \issuedBox{i22};

     \draw[coveredStyle] ($(i31)  + (-1.3,0.5)$) rectangle ++(2.6,-2.0);
     \issuedBox{i32};
  \end{scope}


\end{tikzpicture}}
}
$
\caption{Demonstration of the necessity of \ref{req:acq-iss} in the definition of $\issuable$.
The covered events are marked by
{\protect\tikz \protect\draw[coveredStyle] (0,0) rectangle ++(0.35,0.35);}
and the issued ones by 
{\protect\tikz \protect\draw[issuedStyle] (0,0) rectangle ++(0.35,0.35);}.}
\label{fig:ex-acq-iss}
\end{figure}

\begin{example}
Consider the program and its execution graph on \cref{fig:ex-acq-iss}.
Why does $\ee{43}$ have to be issued after $\ee{11}$, \ie
why to respect a path $[\ee{11}]\mathbin{;} G.\lRFE\mathbin{;} [G.\erlab{\acq}{}{}{}]\mathbin{;}G.\lPO\mathbin{;} [\ee{43}]$?
In the corresponding state of simulation, 
the \Promise memory has messages related to the issued set with timestamps respecting $G.\lCO$.
Without loss of generality, suppose that the memory contains the 
messages $\msg{\locy}{2}{1}{2}{[\locy@2]}$, $\msg{\loc}{2}{1}{2}{[\loc@2,\locy@2]}$,
and $\msg{\locz}{2}{1}{2}{[\loc@2,\locz@2]}$
related to $\ee{21}$, $\ee{22}$, and $\ee{32}$ respectively.
Since the event $\ee{41}$ is covered, the fourth thread has already executed the instruction $\ee{41}$, which is an acquire read.
Thus, its current view is updated to include $[\loc@2,\locz@2]$.
Suppose that $\ee{43}$ is issued. Then, the \Promise machine has to be able to promise
a message $\msg{\locy}{1}{\_}{t_{\ee{43}}}{[y@t_{\ee{43}}]}$ for some $t_{\ee{43}}$. 
The timestamp $t_{\ee{43}}$ has to be less than $2$, which is the timestamp of the message related to $\ee{21}$, since $\tup{\ee{43},\ee{21}} \in G.\lCO$.
Now, consider a certification run of the fourth thread.
In the first step of the run, the thread executes the instruction $\ee{42}$.
It is forced to read from $\msg{\loc}{2}{1}{2}{[\loc@2,\locy@2]}$ since thread's view is equal to $[\loc@2,\locz@2]$.
Because $\ee{42}$ is an acquire read, the thread's current view incorporates the message's view
and becomes $[\loc@2,\locy@2,\locz@2]$. After that, the thread cannot fulfill the promise to the
location $\locy$ with the timestamp $t_{\ee{43}}<2$.
\qed
\end{example}

\begin{example}
  \label{ex:w-strong-iss}
To see why we need \ref{req:w-strong-iss}, revisit the program in \cref{ex:strong-rmw}.
Suppose that we allow to issue $\ewlab{\rlx}{y}{1}{}$ before issuing $\ewlab{\rel}{x}{1}{\strong}$.
Correspondingly, in \Promise, the second thread promises a message $\msg{y}{1}{1}{2}{[y@2]}$ and
has to certify it in any future memory.
Consider a future memory that contains two messages to location $x$: 
an initial one, $\msg{x}{0}{0}{0}{\bot}$, and $\msg{x}{1}{0}{1}{[x@1]}$.
In this state $\incInst{\rlx}{\rel}{c}{x}{1}{\strong}$ has to read from the
non-initial message and assign $1$ to $c$,
since RMWs are required to add messages adjacent to the ones they reads from.
After that, $\writeInst{\rlx}{y}{c+1}$ is no longer able to fulfill
the promise with value $1$. 
\qed
\end{example}


The full definition of coverable event adds (\wrt \cref{def:rlx-coverable}) cases related to fence events:
for an SC fence to be coverable, all $G.\lSC$-previous fence events have to be already covered.

\begin{definition}
An event $e$ is called \emph{coverable} in $G$ 
and $\tup{\CoveredSet, \IssuedSet}$, denoted $e\in \coverable(G, \CoveredSet, \IssuedSet)$, 
if $e \in G.\lE$, $\dom{G.\lPO\mathbin{;}[e]} \suq \CoveredSet$, and either
$(i)$ $e \in G.\lW \cap \IssuedSet$;
$(ii)$ $e \in G.\lR$ and $\dom{G.\lRF\mathbin{;}[e]} \suq \IssuedSet$;
$(iii)$ $e \in G.\lF^{\sqsubset \sco}$;
or $(iv)$ $e \in G.\lF^\sco$ and $\dom{G.\lSC\mathbin{;}[e]} \suq \CoveredSet$.
\end{definition}



By further requiring that traversals configurations $\tup{\CoveredSet, \IssuedSet}$ of
an execution $G$ satisfy $\IssuedSet \cap G.\lW^{\rel} \suq \CoveredSet$ and
$\codom{[\CoveredSet]\mathbin{;}G.\lRMW} \suq \CoveredSet$,
\cref{prop:rlx-trav-full} is extended to the updated definition of the traversal strategy.



\subsection{Thread step simulation}
\label{sec:sim-cf-step}

Next, we refine the simulation relation from \cref{sec:rlx-sim-cf-step}.
The relation $\simthread_{\tid}(G, \TC, \tup{\lts, \gsco, \mem}, \Tfrom, \Tto)$
has an additional parameter $\Tfrom : \IssuedSet \rightarrow \Q$, which is used to assign
lower bounds of a timestamp interval to issued writes ($\Tto$ assigns upper bounds).
We define this relation to hold if the following conditions are met (for conciseness we omit the ``$G.$'' prefix):\footnote{To relate 
the timestamps in the different views to relations in $G$ (items (3),(4),(7)), 
we use essentially the same definitions that were introduced by \citet{Kang-al:POPL17}
when they related the promise-free fragment of \Promise to a declarative model.}
\begin{enumerate}
\item $\Tfrom$ and $\Tto$ agree with $\lCO$ and reflect the requirements on timestamp intervals:
\begin{itemize}
\item $\forall w \in \lE \cap \Init. \; \Tto(w) = \Tfrom(w) = 0$ and $\forall w \in \IssuedSet \setminus \Init. \; \Tfrom(w) < \Tto(w)$
\item $\forall \tup{w, w'} \in [\IssuedSet] \mathbin{;}\lCO \mathbin{;} [\IssuedSet]. \; \Tto(w) \le \Tfrom(w')$ and 
$\forall \tup{w, w'} \in [\IssuedSet] \mathbin{;} \lRF \mathbin{;} \lRMW \mathbin{;} [\IssuedSet]. \; \Tto(w) = \Tfrom(w')$
\end{itemize}

\item Non-initialization messages in $\mem$ have counterparts in $\IssuedSet$:
\begin{itemize}
\item $\forall \msg{\loc}{\_}{f}{t}{\_} \in \mem. \;   t \neq 0 \implies \exists w \in \IssuedSet. \; \lLOC(w) = \loc \land \Tfrom(w) = f \land \Tto(w) = t$
\item $\forall \tup{w,w'} \in [\IssuedSet] \mathbin{;} \lCO \mathbin{;} [\IssuedSet]. \; \Tto(w) = \Tfrom(w') \Rightarrow \tup{w,w'} \in \lRF\mathbin{;}\lRMW$
\end{itemize}

\item The SC view $\gsco$ corresponds to write events that are ``before'' covered SC fences:
\begin{itemize}
\item $\gsco= \lambda \loc.\;  \max \Tto[\lW(\loc) \cap \dom{\lRF^? \mathbin{;} \lHB \mathbin{;} [\CoveredSet \cap \lF^{\sco}]}]$
\end{itemize}
\item 
Issued events have corresponding messages in memory:
\begin{itemize}
\item $\forall w \in \IssuedSet. \;
      \msg{\lLOC(w)}{\lVAL(w)}{\Tfrom(w)}{\Tto(w)}{\msghelper(\Tto, w)} \in \mem$,
      where:
\begin{itemize}
\item $\msghelper(\Tto, w) \defeq (\lambda \loc.\; \max\Tto[\lW(\loc) \cap \dom{\lURR \mathbin{;} \lRELEASE\mathbin{;}[w]}])\sqcup [\lLOC(w)@\Tto(w)] $
\item $\lURR \defeq \lRF^? \mathbin{;} (\lHB \mathbin{;} [\lF^{\sco}])^?\mathbin{;} \lSC^?\mathbin{;} \lHB^?$
\end{itemize}      
\end{itemize}
\item 
For every promise, there exists a corresponding issued uncovered event $w$:
\begin{itemize}
\item $      \forall \msg{\loc}{\val}{f}{t}{\view} \in \lprom. \;
        \exists w \in \lE_{\tid} \cap \IssuedSet \setminus \CoveredSet. \;$
\\ $\hspace*{20pt}\lLOC(w) = \loc \land \lVAL(w) = \val \land \Tfrom(w) = f \land \Tto(w) = t \land \view = \msghelper(\Tto, w) $
\end{itemize}
\item  Every issued uncovered event $w$ of thread $\tid$ has a corresponding promise in $\lprom$.
Its message view includes the singleton view $[\lLOC(w)@\Tto(w)]$
and the thread's release view $\viewRel$ (third component of $\tcom$).
If $w$ is an RMW write, and its read part is reading from an issued write $p$,
the view of the message that corresponds to $p$ is also included in $w$'s message view.
\begin{itemize}
\item $ \forall w \in \lE_{\tid} \cap \IssuedSet \setminus (\CoveredSet \cup \codom{[\IssuedSet] \mathbin{;} \lRF \mathbin{;} \lRMW}).\;$
\\ $\hspace*{30pt} \msg{\lLOC(w)}{\lVAL(w)}{\Tfrom(w)}{\Tto(w)}{[\lLOC(w) @ \Tto(w)] \sqcup \viewRel(\loc) } \in \lprom$ 
\item $ \forall w \in \lE_{\tid} \cap \IssuedSet \setminus \CoveredSet, p\in \IssuedSet.\; \tup{p, w} \in \lRF \mathbin{;} \lRMW \implies$
 \\ $\hspace*{30pt} \msg{\lLOC(w)}{\lVAL(w)}{\Tfrom(w)}{\Tto(w)}{[\lLOC(w) @ \Tto(w)] \sqcup \viewRel(\loc) \sqcup \msghelper(\Tto, p)} \in \lprom$
\end{itemize}

\item The three components $\tup{\viewCur, \viewAcq, \viewRel}$ of $\tcom$ are justified by graph paths:
\begin{itemize}
\item $\viewCur = \lambda \loc.\;  \max\Tto[\lW(\loc) \cap \dom{\lURR \mathbin{;} [\lE\tidmod{\tid}\cap \CoveredSet]}]$
\item $\viewAcq = \lambda \loc.\;  \max\Tto[\lW(\loc) \cap \dom{\lURR \mathbin{;} (\lRELEASE\mathbin{;} \lRF)^? \mathbin{;} [\lE\tidmod{\tid}\cap \CoveredSet]}]$
\item $\viewRel  = \lambda \loc,\locy.\; \max\Tto[\lW(\loc) \cap (\dom{\lURR \mathbin{;} [(\lW^{\rel}(\locy) \cup \flab{\sqsupseteq \rel}) \cap \lE\tidmod{\tid} \cap \CoveredSet]} \cup \lW(\locy) \cap \lE\tidmod{\tid} \cap\CoveredSet)]$
\end{itemize}

\item 
The thread local state $\sigma$ matches the covered events ($\sigma.\lG.\lE = C \cap \lE\tidmod{\tid}$),
and can always reach the execution graph $G$ 
($\exists \sigma'. \; \sigma \to_{\tid}^* \sigma' \land \sigma'.\lG = G\rst{\tid}$).
\end{enumerate}

We also state a version of \cref{prop:rlx-sim-step} for the new relation.
\begin{proposition}
\label{prop:sim-step}
  If $\simthread_{\tid}(G, \TC, \tup{\lts, \gsco, \mem},\Tfrom, \Tto)$ and 
  $G \vdash \TC \etravStep_{\tid} \TC'$ hold,
  then there exist $\lts'$, $\gsco'$, $\mem'$, $\Tfrom'$, $\Tto'$ such that
$\tup{\lts,\gsco,\mem} \astep{}^+ \tup{\lts',\gsco',\mem'}$
and $\simthread_{\tid}(G, \TC', \tup{\lts', \gsco', \mem'},\Tfrom', \Tto')$ hold.
\end{proposition}


\subsection{Certification}
\label{sec:certification}

We move on to the construction of certification graphs.
First, the set of events of  $\Gcert$ is extended:
\begin{align*}
\Gcert.\lE  \defeq ~& \CoveredSet \cup \IssuedSet \cup \dom{G.\lPO\mathbin{;}[\IssuedSet\cap G.\lE\tidmod{\tid}]} ~ \cup \\
 & (\dom{G.\lRMW\mathbin{;} [\IssuedSet \cap G.\lE\tidmod{\neq\tid}]} \setminus \codom{[G.\lE \setminus \codom{G.\lRMW}] \mathbin{;} G.\lRFI}) 
\end{align*}
It additionally contains read parts of issued RMWs in other threads
(excluding those reading locally from a non-RMW write).
They are needed to preserve release sequences to issued writes in $\Gcert$.


The $\lRMW,\lSC$ and dependencies components of $\Gcert$ are the same as in (restricted) $G$
($\Gcert.\lX = [\Gcert.\lE]\mathbin{;} G.\lX\mathbin{;} [\Gcert.\lE]$ for
$\lX \in \set{\lRMW,\lADDR,\lDATA,\lCTRL,\lRMWDEP,\lSC}$) as in \cref{sec:rlx-certification}.
However, $G.\lCO$ edges have to be altered due to the future memory quantification in \Promise certifications.

\begin{example}
  \label{ex:cert-co}
Consider the annotated execution $G$ and its traversal configuration 
($\CoveredSet = \emptyset$ and $\IssuedSet = \{\ee{11},\ee{22}\}$)
shown in the inlined figure.
Suppose that 
$\simthread_{\tid_2}(G, \tup{\CoveredSet, \IssuedSet}, \tup{\tup{\sigma,\tcom,\lprom}, \gsco, \mem}, \Tfrom, \Tto)$ holds for
some $\sigma$, $\tcom$, $\lprom$, $\mem$, $\gsco$, $\Tfrom$ and $\Tto$.
Hence, there are messages of the form $\msg{x}{2}{\Tfrom(\ee{11})}{\Tto(\ee{11})}{\_}$ 
and $\msg{x}{3}{\Tfrom(\ee{22})}{\Tto(\ee{22})}{\_}\}$ in $\mem$ and
$\Tfrom(\ee{11}) < \Tto(\ee{11}) \le \Tfrom(\ee{22}) < \Tto(\ee{22})$.

{\makeatletter
\let\par\@@par
\par\parshape0
\everypar{}\begin{wrapfigure}{r}{0.4\textwidth}\centering
\begin{tikzpicture}[yscale=0.9,xscale=1]
  \node (i11) at (0, -0.6) {$\ee{11}: \ewlab{\rlx}{\loc}{2}{}$};

  \node (i21) at (3,    0) {$\ee{21}: \ewlab{\rlx}{\loc}{1}{}$};
  \node (i22) at (3, -1.2) {$\ee{22}: \ewlab{\rlx}{\loc}{3}{}$};

  \begin{scope}[on background layer]
     \issuedBox{i11};
     \issuedBox{i22};
  \end{scope}

  \draw[mo] (i21) edge node[above] {\smaller\smaller$\lCOE$} (i11);
  \draw[mo] (i21) edge node[right] {\smaller\smaller$\lCOI$} (i22);
  \draw[mo] (i11) edge node[below] {\smaller\smaller$\lCOE$} (i22);
\end{tikzpicture}
\label{fig:co-helper}
 \end{wrapfigure}

During certification, we have to execute the instruction related to $\ee{21}$
and add a corresponding message to $\mem$.
Since certification is required for every future memory $\omem\supseteq\mem$,
it might be the case that here is no free timestamp $t'$ in $\omem$ such that $t' \le \Tfrom(\ee{11})$.
Thus, our chosen timestamps cannot agree with $G.\lCO$.
However, if we place $\ee{21}$ as the immediate predecessor of $\ee{22}$ in $\Gcert.\lCO$,
we may use the splitting feature of \Promise: the promised message 
$\msg{x}{3}{\Tfrom(\ee{22})}{\Tto(\ee{22})}{\_}\}$ can be split into two messages
$\msg{x}{1}{\Tfrom(\ee{22})}{t}{\_}\}$ and $\msg{x}{3}{t}{\Tto(\ee{22})}{\_}\}$
for any $t$ such that $\Tfrom(\ee{22}) < t < \Tto(\ee{22})$.
To do so, we need the non-issued writes of the certified thread
to be immediate predecessors of the issued ones in $\Gcert.\lCO$.
By performing such split, we do not ``allocate'' new timestamp intervals,
which allows us to handle arbitrary future memories.
Note that if we had writes to other locations to perform during the certification,
with no possible promises to split, we would need them to be placed last in 
$\Gcert.\lCO$, so we can relate them to messages whose timestamps are larger than all timestamps in $\omem$.
\qed
  \par}
\end{example}


Following \cref{ex:cert-co}, 
we define $\Gcert.\lCO$ to consist of all pairs $\tup{w,w'}$ such that 
$w,w'\in \Gcert.\lE \cap G.\lW$, 
$G.\lLOC(w)=G.\lLOC(w')$, and either 
$\tup{w,w'} \in ([\IssuedSet] \mathbin{;} G.\lCO \mathbin{;} [\IssuedSet] \cup 
[\IssuedSet] \mathbin{;} G.\lCO \mathbin{;} [\Gcert.\lE\tidmod{\tid}]  \cup [\Gcert.\lE\tidmod{\tid}] \mathbin{;} G.\lCO \mathbin{;}[\Gcert.\lE\tidmod{\tid}])^+$,
or there is no such path, $w\in \IssuedSet$, and $w'\in \Gcert.\lE\tidmod{\tid}\setminus \IssuedSet$.
This construction essentially ``pushes'' the non-issued writes of the certified thread to be as late as possible in 
$\Gcert.\lCO$.
 
%
%

\smallskip

The definition of $\Gcert.\lRF$ is also adjusted to be in accordance with $\Gcert.\lCO$:
\[\inarr{
\Gcert.\lRF \defeq G.\lRF\mathbin{;} [\D] \cup \bigcup_{\loc \in \Loc} ([G.\lW(\loc)] \mathbin{;}G.\lBVF \mathbin{;} [G.\lR(\loc) \cap \Gcert.\lE \setminus  \D] \setminus \Gcert.\lCO \mathbin{;} G.\lBVF)
\\
\text{~~where}~~
\D=\Gcert.\lE \cap (\CoveredSet \cup \IssuedSet \cup G.\lE\tidmod{\neq\tid} \cup
 \dom{G.\lRFI^? \mathbin{;} G.\lPPO \mathbin{;} [\IssuedSet]} \cup
 \codom{G.\lRFE \mathbin{;} [G.\lR^{\acq}]})
 \text{~~and} \\
\qquad \quad
G.\lBVF = (G.\lRF \mathbin{;} [\D])^? \mathbin{;}
         (G.\lHB \mathbin{;} [G.\lF^{\sco}])^?\mathbin{;}
         G.\lSC^?\mathbin{;} G.\lHB
}\]
The set of determined events is extended to include acquire read events which read externally, \ie
the ones potentially engaged in synchronization.

For the certification graph $\Gcert$ presented here, we prove a version of \cref{prop:rlx-cert-cons},
\ie show that the graph is \IMM-consistent and $\TCcert$ is its traversal configuration,
and adapt \cref{prop:rlx-cert-graph} as follows.
\begin{proposition}
\label{prop:cert-graph}
Suppose that $\simthread_{\tid}(G, \TC, \tup{\lts, \gsco, \mem}, \Tfrom, \Tto)$ holds.
Then, for every $\omem \supseteq \mem$ and $\ogsco \ge \gsco$,
$\simthreadcert_{\tid}(\Gcert,\TCcert, \tup{\lts, \ogsco, \omem}, \Tfrom, \Tto)$ holds.
\end{proposition}
Here, $\simthreadcert_\tid$ is a modified simulation relation, which differs to $\simthread_\tid$
in the following parts:
\begin{enumerate}
\item[(2)] 
Since certification begins from an arbitrary future memory, 
we cannot require that \emph{all} messages in memory have counterparts in $\IssuedSet$.
Here, it suffices to assert that 
all RMW writes are issued ($\codom{\Gcert.\lRMW} \suq \IssuedSet$), 
and for every non-issued write either it is last in
$\Gcert.\lCO$ or its immediate successor is in the same thread
($[\Gcert.\lE \setminus \IssuedSet] \mathbin{;} \imm{\Gcert.\lCO} \suq \Gcert.\lPO$).
 The latter allows us to split existing messages to obtain timestamp intervals for non-issued writes
 during certification (see \cref{ex:cert-co}).
\item[(3)]
Since certification begins from arbitrary future SC view,
$\gsco$ may not correspond to $\Gcert$.
Nevertheless, SC fences cannot be executed in the certification run, and we can simply require that all SC fences are covered
($ \Gcert.\lF^{\sco} \suq \CoveredSetcert$).
\end{enumerate}


We also show that a version of \cref{prop:sim-step} holds for $\simthreadcert$.
It allows us to prove a strengthened version \cref{prop:sim-step}, which also concludes that 
new \Promise thread state is certifiable, in a similar way we prove \cref{prop:rlx-ext-sim-step}.

\begin{proposition}
\label{prop:ext-sim-step}
  If $\simthread_{\tid}(G, \TC, \tup{\lts, \gsco, \mem}, \Tfrom, \Tto)$ and 
  $G \vdash \TC \etravStep_{\tid} \TC'$ hold,
  then there exist $\lts', \gsco',\mem',\Tfrom',\Tto'$ such that
$\tup{\lts,\gsco,\mem} \astep{}^+ \tup{\lts',\gsco',\mem'}$ and $\simthread_{\tid}(G, \TC', \tup{\lts', \gsco', \mem'},\Tfrom', \Tto')$ hold,
and for every $\ogsco \geq \gsco',\omem\supseteq\mem'$,
there exist $\lts'',\ogsco',\omem'$ such that $\tup{\lts',\ogsco,\omem} \astep{}^* \tup{\lts'',\ogsco',\omem'}$ and $\lts''.\lprmem = \emptyset$.
\end{proposition}

\section{Related work}
\label{sec:related}

Together with the introduction of the promising semantics, \citet{Kang-al:POPL17} provided a declarative presentation
of the promise-free fragment of the promising model. They established the adequacy of this presentation
using a simulation relation, which resembles the simulation relation that we use in \cref{sec:prom-compilation}.
Nevertheless, since their declarative model captures only the promise-free fragment of \Promise, 
the simulation argument is much simpler,
and no certification condition is required. In particular, their analogue to our traversal strategy would simply
cover the events of the execution graph following $\lPO\cup\lRF$.

To establish the correctness of compilation of the promising semantics to POWER, 
\citet{Kang-al:POPL17} followed the approach of \citet{trns}.  
This approach reduces compilation correctness to POWER to $(i)$ the correctness of compilation
to the POWER model strengthened with $\lPO\cup\lRF$ acyclicity; and $(ii)$ the soundness of local reorderings of memory accesses.
To establish $(i)$, \citet{Kang-al:POPL17} wrongly argued that the strengthened POWER-consistency of mapped 
promise-free execution graphs imply the promise-free consistency of the source execution graphs.
This is not the case due to SC fences, which have relatively strong semantics in the promise-free declarative model
(see \citeapp{sec:POPL17mistake}{Appendix D} for a counter example).
Nevertheless, our proof shows that the compilation claim of \citet{Kang-al:POPL17} is correct.
We note also that, due to the limitations of this approach, \citet{Kang-al:POPL17} only
claimed the correctness of a less efficient compilation scheme to POWER that requires \texttt{lwsync} barriers
after acquire loads rather than (cheaper) control dependent \texttt{isync} barriers.
Finally, this approach cannot work for ARM as it relies on the relative strength of POWER's preserved program order.

\citet{Podkopaev-al:ECOOP17} proved (by paper-and-pencil) the correctness of compilation from the promising semantics to ARMv8. 
Their result handled only a restricted subset of the concurrency features of the promising semantics,
leaving release/acquire accesses, RMWs, and SC fences out of scope.
In addition, as a model of ARMv8, they used an operational model, ARMv8-POP~\cite{arm8-model},  
that was later abandoned by ARM in favor of a stronger different declarative model~\cite{Pulte-al:POPL18}.
Our proof in this paper is mechanized, supports all features of the promising semantics,
and uses the recent declarative model of ARMv8.

\citet{Wickerson-al:POPL17} developed a tool, based on the Alloy solver, that can be used to test the correctness
of compiler mappings. Given the source and target models and the intended compiler mapping, their tool
searches for minimal litmus tests that witness a bug in the mapping.
While their work concerns automatic bug detection, the current work is focused around formal verification
of the intended mappings. In addition, their tool is limited to declarative specifications, and cannot be used
to test the correctness of the compilation of the promising semantics.

%
%
%

Finally, we note that \IMM is weaker than the ARMv8 memory model of \citet{Pulte-al:POPL18}.
In particular, \IMM is not multi-copy atomic (see \cref{ex:iriw}); its release writes provide weaker guarantees
(allowing in particular the so-called 2+2W weak behavior~\cite{sra,tutorial_arm_power}); 
it does not preserve address dependencies between reads 
(allowing in particular the ``big detour'' weak behavior~\cite{Pulte-al:POPL18}); 
and it allows ``write subsumption''~\cite{arm8-model,Pulte-al:POPL18}.
Formally, this is a result of not including $\lFR$ and $\lCO$ in a global acyclicity condition,
but rather having them in a C/C++11-like coherence condition.
While \citet{Pulte-al:POPL18} consider these strengthenings of the ARMv8 model as beneficial
for its simplicity, we do not see \IMM as being much more complicated than the ARMv8 declarative model.
(In particular, \IMM's derived relations are not mutually recursive.)
Whether or not these weaknesses of \IMM in comparison to ARMv8 allow more optimizations and better performance is left for future work.

\section{Concluding remarks}
\label{sec:conclusion}

We introduced a novel intermediate model, called \IMM, 
as a way to bridge the gap between language-level and hardware models and modularize compilation correctness proofs.
On the hardware side, we provided (machine-verified) mappings from \IMM to the main multi-core architectures,
establishing \IMM as a common denominator of existing hardware weak memory models.
On the programming language side, we proved the correctness of compilation from the promising semantics,
as well as from a fragment of (R)C11, to \IMM.

In the future, we plan to extend our proof for verifying the mappings from full (R)C11 to \IMM
as well as to handle infinite executions with a more expressive notion of a program outcome.
We believe that \IMM can be also used to verify 
the implementability of other language-level models mentioned in \cref{sec:intro}.
This might require some modifications of \IMM (in the case it is too weak for certain models)
but these modifications should be easier to implement and check over the existing mechanized proofs.
Similarly, new (and revised) hardware models could be related to (again, a possibly modified version of) \IMM.
Specifically, it would be nice to extend \IMM to support mixed-size accesses~\cite{mixed}
and hardware transactional primitives~\cite{Chong_transactions, Dongol_transactions}.
On a larger scope, we believe that \IMM may provide a basis for extending CompCert~\cite{compcert,compcerttso}
to support modern multi-core architectures beyond x86-TSO.

\begin{acks}
We thank Orestis Melkonian for his help with Coq proof concerning the POWER model
in the context of another project,
and the POPL'19 reviewers for their helpful feedback.
The first author was supported by RFBR (grant number 18-01-00380).
The second author was supported by the Israel Science Foundation
(grant number 5166651), and by Len Blavatnik and the Blavatnik Family foundation.
\end{acks}

\bibliography{main}

\newpage
\appendix

\section{Examples: from programs to execution graphs}
\label{sec:prog2exec_examples}
We provide several examples of sequential programs and their execution graphs,
constructed according to the semantics in \cref{fig:prog-to-exec}.

\begin{example}
The program below has conditional branching.
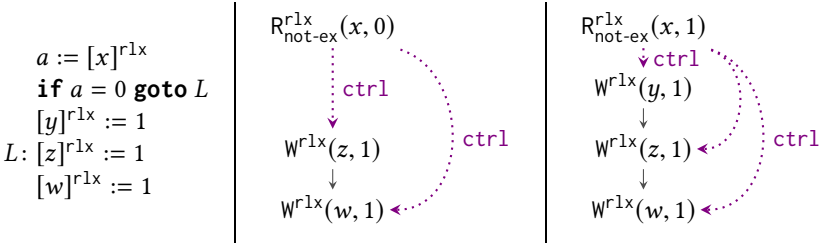
\[
\inarr{
 \phantom{L\colon} \readInst{\rlx}{a}{x} \\
 \phantom{L\colon} \ifGotoInst{a = 0}{L} \\
 \phantom{L\colon} \writeInst{\rlx}{y}{1} \\
 L\colon \writeInst{\rlx}{z}{1} \\
 \phantom{L\colon} \writeInst{\rlx}{w}{1} \\
}
\quad\vrule\quad
\inarr{\begin{tikzpicture}[yscale=0.8,xscale=1]
 \node (i11) at (0,  0) {$\erlab{\rlx}{x}{0}{\isnotex}$};
 \node (i13) at (0, -2) {$\ewlab{\rlx}{z}{1}{}$};
 \node (i14) at (0, -3) {$\ewlab{\rlx}{w}{1}{}$};

 \draw[deps] (i11) edge node[right] {\small$\lCTRL$} (i13);
 \draw[deps,out=-20,in=0] (i11) edge node[right] {\small$\lCTRL$} (i14);
 \draw[po] (i13) edge (i14);
\end{tikzpicture}}
\quad\vrule\quad
\inarr{\begin{tikzpicture}[yscale=0.8,xscale=1]
 \node (i11) at (0,  0) {$\erlab{\rlx}{x}{1}{\isnotex}$};
 \node (i12) at (0, -1) {$\ewlab{\rlx}{y}{1}{}$};
 \node (i13) at (0, -2) {$\ewlab{\rlx}{z}{1}{}$};
 \node (i14) at (0, -3) {$\ewlab{\rlx}{w}{1}{}$};

 \draw[deps] (i11) edge node[right] {\small$\lCTRL$} (i12);
 \draw[deps,out=-20,in=0] (i11) edge node[right] {} (i13);
 \draw[deps,out=-20,in=0] (i11) edge node[right] {\small$\lCTRL$} (i14);
 \draw[po] (i12) edge (i13);
 \draw[po] (i13) edge (i14);
\end{tikzpicture}}
\]
Note that $\lCTRL$ is downward closed (the set $S$ is non-decreasing during the steps of the semantics).
\qed
\end{example}

\begin{example}
The following program has an atomic fetch-and-add instruction,
whose location and added value depend on previous read instructions
(recall that $\Val=\Loc=\N$ and $x,y,z,w$ represent some constants):
\[
\inarr{
  \readInst{\rlx}{a}{x} \comment{z} \\
  \readInst{\rlx}{b}{y} \comment{1} \\
  \incInst{\rlx}{\rlx}{c}{a}{b}{\normal} \comment{2} \\
  \writeInst{\rlx}{w}{1} \\
}
\quad\vrule\quad
\inarr{\begin{tikzpicture}[yscale=0.8,xscale=1]
  \node (i11) at (0,  0) {$\erlab{\rlx}{x}{z}{\isnotex}$};
  \node (i12) at (0, -1) {$\erlab{\rlx}{y}{1}{\isnotex}$};
  \node (i13) at (0, -2) {$\erlab{\rlx}{z}{2}{\isex}$};
  \node (i14) at (0, -3) {$\ewlab{\rlx}{z}{3}{\normal}$};
  \node (i15) at (0, -4) {$\ewlab{\rlx}{w}{1}{}$};

  \draw[po] (i11) edge (i12);
  \draw[po] (i12) edge (i13);
  \draw[deps,out=-160,in=180] (i12) edge node[left] {\small$\lDATA$} (i14);
  \draw[deps,out=-160,in=180] (i13) edge (i14);
  \draw[deps,out=-20,in=0] (i11) edge (i13);
  \draw[deps,out=-20,in=0] (i11) edge node[right] {\small$\lADDR$} (i14);
  \draw[rmw] (i13) edge node[right] {\small$\lRMW$} (i14);
  \draw[po] (i14) edge (i15);
\end{tikzpicture}}
\]
\qed
\end{example}

\section{\POWER-consistency}
\label{sec:power_consistent}

We define \POWER-consistency following~\cite{herding-cats}.
This section is described in the context of a given \POWER execution graph $G_p$,
and the `$G_p.$' prefixes are omitted.

The definition requires the following derived relations
(see~\cite{herding-cats} for further explanations and details):

\begin{align*}
\lSYNC  & \defeq [\lR \cup \lW];\lPO;[\lF^\sync];\lPO;[\lR\cup \lW]
\tag{\emph{sync order}} \\
\lLWSYNC & \defeq [\lR \cup \lW];\lPO;[\lF^\lwsync];\lPO;[\lR \cup \lW] \setminus (\lW\times \lR)
\tag{\emph{lwsync order}} \\
\lFENCE & \defeq \lSYNC \cup \lLWSYNC
\tag{\emph{fence order}} \\
\lHBP & \defeq \lPPOP \cup \lFENCE \cup \lRFE
\tag{\emph{POWER's happens-before}} \\
\lPROP_1 & \defeq [\lW];\lRFE^?;\lFENCE;\lHBP^*;[\lW] \\
\lPROP_2 & \defeq (\lCOE \cup \lFRE)^?;\lRFE^?;(\lFENCE;\lHBP^*)^? ; \lSYNC ; \lHBP^* \\
\lPROP & \defeq \lPROP_1 \cup \lPROP_2
\tag{\emph{propagation relation}}
\end{align*}

In the definition on $\lHBP$, \POWER employs a ``preserved program order'' denoted $\lPPOP$.
The definition of this relation is quite intricate 
and requires several more additional derived relations (its correctness was extensively tested~\cite{herding-cats}):

\begin{align*}
\lCTRLISYNC & \defeq [\lR];\lCTRL;[\lF^\isync];\lPO \tag{\emph{ctrl-isync order}} \\
\lRDW  & \defeq (\lFRE;\lRFE) \cap \lPO  \tag{\emph{read different writes}}\\
\lPPOP & \defeq [\lR];\lii;[\lR]\cup [\lR];\lic;[\lW] 
\tag{\emph{POWER's preserved program order}} \\
\end{align*}
where, $\lii,\lic,\lci,\lcc$ are inductively defined as follows:
$$\begin{array}{c@{\hspace{1em}}c@{\hspace{1em}}c@{\hspace{1em}}c@{\hspace{1em}}c@{\hspace{1em}}c@{\hspace{1em}}c@{\hspace{1em}}c}
\inferrule*{\lADDR}{\lii}&
\inferrule*{\lDATA}{\lii}&
\inferrule*{\lRDW}{\lii}&
\inferrule*{\lRFI}{\lii}& \inferrule*{\lci}{\lii}    &  & \inferrule*{\lic;\lci}{\lii}   & \inferrule*{\lii;\lii}{\lii}      \\
&&& & \inferrule*{\lii}{\lic}      & \inferrule*{\lcc}{\lic}       & \inferrule*{\lic;\lcc}{\lic} & \inferrule*{\lii;\lic}{\lic} \\
&& \inferrule*{\lCTRLISYNC}{\lci} &\inferrule*{\lDETOUR}{\lci}  &&  & \inferrule*{\lci;\lii}{\lci} & \inferrule*{\lcc;\lci}{\lci}  \\
\inferrule*{\lDATA}{\lcc}& \inferrule*{\lCTRL}{\lcc} & \inferrule*{\lADDR;\lPO^?}{\lcc}& 
\inferrule*{\lPO\rst{\lLOC}}{\lcc} & \inferrule*{\lci}{\lcc} &   & \inferrule*{\lci;\lic}{\lcc}  & \inferrule*{\lcc;\lcc}{\lcc} 
\end{array} $$

\begin{definition}
\label{def:power}
A \POWER execution graph $G_p$ is \emph{\POWER-consistent} if the following hold:
\begin{enumerate}
\item $\codom{\lRF} = \lR$. \labelAxiom{$\lRF$-completeness}{ax:P-comp}
\item For every location $\loc\in\Loc$, $\lCO$ totally orders $\ewlab{}{\loc}{}{}$. \labelAxiom{$\lCO$-totality}{ax:P-total}
\item $\lPO\rst{\lLOC} \cup \lRF \cup \lFR \cup \lCO$ is acyclic.
 \labelAxiom{sc-per-loc}{ax:P-sc_per_loc}
\item $\lFRE;\lPROP;\lHBP^*$ is irreflexive.
 \labelAxiom{observation}{ax:P-observation}
\item $\lCO \cup \lPROP$ is acyclic.
\labelAxiom{propagation}{ax:P-propagation}
\item $\lRMW \cap (\lFRE; \lCOE)=\emptyset$.   \labelAxiom{atomicity}{ax:P-atomicity}
\item $\lHBP$ is acyclic. \labelAxiom{power-no-thin-air}{ax:P-thinair}
\end{enumerate}
\end{definition}

\begin{remark}
The model in~\cite{herding-cats} contains an additional constraint: 
$\lCO \cup [\lAT];\lPO;[\lAT]$ should be acyclic
(where $\lAT=\dom{\lRMW}\cup\codom{\lRMW}$). 
Since none of our proofs requires this property, we excluded it from \cref{def:power}.
\end{remark}

\section{\ARM-consistency}
\label{sec:arm_consistent}

We define \ARM-consistency following~\cite{ARMv82model}.
This section is described in the context of a given \ARM execution graph $G_a$,
and the `$G_a.$' prefixes are omitted.

The definition requires the following derived relations
(see~\cite{Pulte-al:POPL18} for further explanations and details):

\begin{align*}
\lOBS &\defeq  \lRFE \cup \lFRE \cup \lCOE  \tag{\emph{observed-by}} \\
\lDOB &\defeq \inarr{(\lADDR \cup \lDATA); \lRFI^? \cup 
	(\lCTRL \cup \lDATA); [\lW]; \lCOI^? \cup 
	\lADDR; \lPO; [\lW] } \tag{\emph{dependency-ordered-before}} \\
\lAOB &\defeq \lRMW \cup [\lW^\isex]; \lRFI; [\lR^{\lQ}] \tag{\emph{atomic-ordered-before}} \\
\lBOB &\defeq \inarr{\lPO; [\lDMBSY]; \lPO  \cup 
	{}[\lR]; \lPO; [\lDMBLD]; \lPO  \cup 
	{}[\lR^{\lQ}]; \lPO \cup 
	\lPO; [\lW^\lL]; \lCOI^?  \tag{\emph{barrier-ordered-before}} }
\end{align*}

\begin{definition}
\label{def:arm}
An \ARM execution graph $G_a$ is called \ARM-consistent if the following hold:
\begin{itemize}
\item $\codom{\lRF} = \lR$. \labelAxiom{$\lRF$-completeness}{ax:A-comp}
\item For every location $\loc\in\Loc$, $\lCO$ totally orders $\ewlab{}{\loc}{}{}$. \labelAxiom{$\lCO$-totality}{ax:A-total}
\item $\lPO\rst{\lLOC} \cup \lRF \cup \lFR \cup \lCO$ is acyclic.
 \labelAxiom{sc-per-loc}{ax:A-sc_per_loc_arm}
\item $\lOBS \cup \lDOB \cup \lAOB \cup \lBOB$ is acyclic. \labelAxiom{external}{ax:external}
\item $\lRMW \cap (\lFRE; \lCOE)=\emptyset$.   \labelAxiom{atomicity}{ax:A-atomicity}
\end{itemize}
\end{definition}

\section{Mistake in Kang et al. (2017)'s compilation to POWER correctness proof}
\label{sec:POPL17mistake}

The following execution graph is not consistent in the promise-free declarative model of \cite{Kang-al:POPL17}.
Nevertheless, its mapping to \POWER (obtained by simply replacing $\flab{\sco}$ with $\flab{\sync}$) is \POWER-consistent
and $\lPO\cup\lRF$ is acyclic (so it is Strong-\POWER-consistent).
Note that, using promises, the promising semantics allows this behavior.

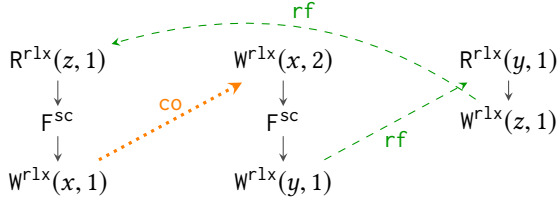
\[
\inarr{\begin{tikzpicture}[yscale=0.8,xscale=1]
  \node (i11) at (0,   1) {$\erlab{\rlx}{z}{1}{}$};
  \node (i12) at (0,   0) {$\flab{\sco}$};
  \node (i13) at (0, -1) {$\ewlab{\rlx}{x}{1}{}$};

  \node (i21) at (3,   1) {$\ewlab{\rlx}{x}{2}{}$};
  \node (i22) at (3,   0) {$\flab{\sco}$};
  \node (i23) at (3, -1) {$\ewlab{\rlx}{y}{1}{}$};

  \node (i31) at (6,   1) {$\erlab{\rlx}{y}{1}{}$};
  \node (i32) at (6,   0) {$\ewlab{\rlx}{z}{1}{}$};
  
  \draw[po] (i11) edge (i12);
  \draw[po] (i12) edge (i13);
  \draw[po] (i21) edge (i22);
  \draw[po] (i22) edge (i23);
  \draw[po] (i31) edge (i32);

  \draw[rf,bend right=30] (i32) edge node[above] {\small$\lRF$} (i11);
  \draw[mo] (i13) edge node[above] {\small$\lCO$} (i21);
  \draw[rf] (i23) edge node[below] {\small$\lRF$} (i31);
\end{tikzpicture}}
\]


\section{Future memory simplification}
\label{sec:futurememory}

\begin{proposition}
  \label{prop:future-mem}
  Let $\tup{\lts,\gsco,\mem}$ be a thread configuration, $\omem$---a future memory (as defined in \cite{Kang-al:POPL17}) to $\mem$ \wrt $\lts.\lprmem$,
  and $\ogsco$---a view such that $\ogsco \ge \gsco$. 
  Then, there exist $\omem'$ and $\ogsco'$ such that $\omem'\supseteq \mem$, $\ogsco' \ge \ogsco$, 
  and the following statement holds.
  If there exist $\lts',\mem'$ and $\gsco'$ such that
  $\tup{\lts,\ogsco',\omem'} \astep{}^* \tup{\lts',\gsco',\mem'}$
  and $\lts'.\lprmem = \emptyset$, then there exist
  $\lts'', \mem''$ and $\gsco''$ such that 
  $\tup{\lts,\ogsco,\omem} \astep{}^* \tup{\lts'',\gsco'',\mem''}$ and $\lts''.\lprmem = \emptyset$ hold.
\end{proposition}
\begin{proof}[Proof outline]
  First, we inductively construct $\omem'$ from $\mem \to^* \omem$ by ignoring modifications, which are not appends of messages.
  Also, we may have to enlarge views of some appended messages to preserve their closeness in $\omem'$ since
  some of them in $\omem$ may point to messages obtained from split modifications. For the same reason,
  we update $\ogsco$ to $\ogsco'$.
  Thus, we know that $\omem'\supseteq \mem$ and $\omem$ and $\omem'$ satisfy the predicate $\futuresimmem$:
  \[\inarr{\futuresimmem(\omem,\omem') \defeq \\
  \quad
  \inarr{(\forall \msg{\loc}{\val}{f'}{t}{\view'} \in \omem'. \\
    \quad \exists f \ge f', \view \le \view'. \; \msg{\loc}{\val}{f}{t}{\view} \in \omem) \land {} \\
  (\forall \msg{\loc}{\val}{f}{t}{\view} \in \omem. \\
    \quad \exists f' \le f, t' \ge t. \; \msg{\loc}{\_}{f'}{t'}{\_} \in \omem').}
  }\]
  Having $\omem'$ and $\ogsco'$, we fix $\lts', \gsco', \mem'$ such that $\tup{\lts,\ogsco',\omem'} \astep{}^* \tup{\lts',\gsco',\mem'}$
  and $\lts'.\lprmem = \emptyset$.

  To prove that the main statement, we do simulation of the \emph{target} execution
  $\tup{\lts,\ogsco',\omem'} \astep{}^* \tup{\lts',\gsco',\mem'}$
  in a \emph{source} machine, which starts from $\tup{\lts,\ogsco,\omem}$.
  To do so, we use the following simulation relation:
  \[\inarr{
  \fsim(
          \tup{\tup{\sigmaT, \tup{\viewCurT,\viewAcqT,\viewRelT},\lpromT}, \gscot,\memt},
\tup{\tup{\sigmaS, \tup{\viewCurS,\viewAcqS,\viewRelS},\lpromS}, \gscos,\mems}
) \defeq \\
  \quad \inarr{
    \sigmaS = \sigmaT \land \lpromS = \lpromT \land {} \\
    \viewCurS \le \viewCurT \land
    \viewAcqS \le \viewAcqT \land
    (\forall \loc. \; \viewRelS(\loc) \le \viewRelT(\loc)) \land {} \\
    \gscos \le \gscot \land \futuresimmem(\mems,\memt).
  }}\]
  It holds for the initial state of the simulation, \ie $\fsim(\tup{\lts,\ogsco',\omem'},\tup{\lts,\ogsco,\omem})$ holds.
  The induction step holds as, from $\fsim$, it follows that the source machine has less restrictions.
\end{proof}


\section{On existence of traversal}
\label{sec:trav-proof}
All results described in this section are mechanized in Coq.

We use a \emph{small traversal step} and the notion of \emph{partially coherent traversal configuration} to prove
the extended of \cref{prop:rlx-trav-full} discussed in \cref{sec:traversal}.
First, we show that for a partial traversal configuration $\tup{\CoveredSet,\IssuedSet}$ of $G$ such that $C \neq G.\lE$ there exists
a small traversal step to a new partial traversal configuration (\cref{prop:trav-step}).
Second, we prove that for a traversal configuration $\tup{\CoveredSet,\IssuedSet}$ if there exists a small traversal step from it,
then there exists a (normal) traversal step from it (\cref{prop:trav-to-etrav}).
Using that, we prove the extension of \cref{prop:rlx-trav-full} for an execution graph $G$ and its traversal configuration
$\tup{\CoveredSet, \IssuedSet}$ by an induction on $|G.\lE \setminus C| + |G.\lW \setminus I|$
applying \cref{prop:trav-step} and \cref{prop:trav-to-etrav}.

\begin{definition}
A pair $\tup{\CoveredSet, \IssuedSet}$ is a \emph{partial traversal configuration} of an execution $G$,
denoted $\travConfigP(G, \tup{\CoveredSet, \IssuedSet})$, if
  $\lE \cap \Init \suq \CoveredSet$,
  $\CoveredSet \suq \coverable(G, \CoveredSet, \IssuedSet)$, and
  $\IssuedSet \suq \issuable(G, \CoveredSet, \IssuedSet)$ hold.
\end{definition}

An operational semantics of a so-called \emph{small traversal step}, denoted $\travConfigStep$, has two rules.
One of them adds an event to covered, another one---to issued (here $\coverable$ and $\issuable$ are defined as
in \cref{sec:traversal}):
\begin{mathpar}
\inferrule*{
    a \in \coverable(G, \CoveredSet, \IssuedSet)
}{
    G \vdash 
    \tup{\CoveredSet, \IssuedSet} \travConfigStep \tup{\CoveredSet \uplus \{a\}, \IssuedSet}
} \and
\inferrule*{
    w \in \issuable(G, \CoveredSet, \IssuedSet)
    }{
    G \vdash
    \tup{\CoveredSet, \IssuedSet} \travConfigStep \tup{\CoveredSet, \IssuedSet \uplus \{w\}}
}
\end{mathpar}
It is obvious that $G \vdash \TC \etravStep \TC'$ implies $G \vdash \TC \travConfigStep^{+} \TC'$ for any $G, \TC, \TC'$.

\begin{proposition}
  \label{prop:trav-step}
  Let $G$ be an \IMM-consistent execution and $\tup{\CoveredSet, \IssuedSet}$ be its partial
  traversal configuration.
  If $\CoveredSet \neq G.\lE$, then there exist $\CoveredSet'$ and $\IssuedSet'$ such that
  $G \vdash \tup{\CoveredSet, \IssuedSet} \travConfigStep \tup{\CoveredSet', \IssuedSet'}$.
\end{proposition}
\begin{proof}
  Let's denote a set of threads, which have non-covered events, by $U$, \ie $U \defeq \{\tid \mid G\rst{\tid} \not \suq \CoveredSet\}$.
  For each thread $\tid \in U$, there exists an event, which we 
  denote $n_{\tid}$, such that $\dom{G.\lPO; [n_{\tid}]} \suq \CoveredSet$ and $n_{\tid} \nin \CoveredSet$.

  Consider the case then there exists a thread $\tid \in U$ such that
  $n_{\tid} \in \coverable(G, \CoveredSet, \IssuedSet)$.
  Then, the statement is proven since $G \vdash \tup{\CoveredSet, \IssuedSet} \travConfigStep \tup{\CoveredSet \uplus \{n_{\tid}\}, \IssuedSet}$ holds.
  
  Now, consider the case then $n_{\tid} \nin \coverable(G, \CoveredSet, \IssuedSet)$ for each thread $\tid \in U$.
  If there exists a $\tid \in U$ such that $n_{\tid} \in G.\lW$,
  we know that $n_{\tid} \nin \IssuedSet$ since it is not coverable.
  From definition of $n_{\tid}$, it follows that $n_{\tid} \in \issuable(G, \CoveredSet, \IssuedSet)$ holds, and 
  the statement is proven since $G \vdash \tup{\CoveredSet, \IssuedSet} \travConfigStep \tup{\CoveredSet, \IssuedSet \uplus \{n_{\tid}\}}$ holds.

  In other case, $N \defeq \{n_{\tid} \mid \tid \in U\} \suq G.\lR \cup G.\lF^{\sco}$.
  For each $r \in N \cap G.\lR$, we know that  $G.\lRF^{-1}(r) \nin \IssuedSet$, and
  for each $f \in N \cap G.\lF^\sco$, there exists $f' \in \dom{G.\lSC; [f]} \setminus \CoveredSet$.
  For this situation, we show that there exists a write event, which is issuable.

  Let's show that there is at least one read event in $N$.
  Suppose that there is no read event, then $N \suq \lF^\sco$.
  Let's pick a fence event $f'$ from $\lF^\sco \setminus \CoveredSet$, which is minimal according to $G.\lSC$ order.
  Since it is not in $N$ according to the previous paragraph,
  there is an event $f \in N$ such that $\tup{f, f'} \in G.\lPO$. That means $\tup{f, f'} \in G.\lBOB$
  and there is a $G.\lAR$-cycle since $\tup{f', f} \in G.\lSC$. It contradicts \IMM-consistency of $G$.

  Thus, there is at least one read $r \in N$. We know that the read is not coverable. It means
  that $G.\lW \not \suq \IssuedSet$ and there is a write event, which is not promised yet, \ie $G.\lRF^{-1}(r) \nin \IssuedSet$.
  Let's pick a write event $w \in G.\lW \setminus \IssuedSet$ such that it is $\lAR^+$-minimal among $G.\lW \setminus \IssuedSet$,
  \ie $\nexists w' \in G.\lW \setminus \IssuedSet. \; \lAR^+(w', w)$.
  In the remainder of the proof, we show that $w$ is issuable, and
  $G \vdash \tup{\CoveredSet, \IssuedSet} \travConfigStep \tup{\CoveredSet, \IssuedSet \uplus \{w\}}$ holds consequently.
  
  There are two options: either $w$ is $G.\lPO$-preceded by a fence event from $N$,
  or $w$ is $G.\lPO$-preceded by a read event from $N$.
  Consider the cases:
\begin{itemize}
\item There exist $f \in N \cap G.\lF^{\sco}$ and $f' \in G.\lF^{\sco}$
  such that  $G.\lPO(f, w)$, $G.\lSC(f', f)$, and $f' \nin G.\lE \setminus \CoveredSet$.
  Without loss of generality, we may
  assume that $f'$ is a $\lSC$-minimal event, which is not covered.
  From the definition of $N$, it follows that there exists $r \in N \cap \lR$, such
  that $G.\lPO(r, f')$.
  We also know that $G.\lRF^{-1}(r) = G.\lRFE^{-1}(r) \nin \IssuedSet$.
  It means that $\tup{G.\lRF^{-1}(r), w} \in G.\lRFE; G.\lPO; [G.\lF^\sco]; G.\lSC; [G.\lF^\sco]; G.\lPO \suq G.\lRFE; G.\lBOB; G.\lSC; G.\lBOB \suq G.\lAR^+$.
  It contradicts $\lAR^+$-minimality of $w$.
\item There exists $r \in N \cap \lR$,
  such that $G.\lPO(r, w)$, $G.\lRF^{-1}(r) \nin \IssuedSet$.
  Since $\CoveredSet \cap G.\lW \suq \IssuedSet$ and $\CoveredSet$ is prefix-closed, $G.\lRF^{-1}(r) = G.\lRFE^{-1}(r)$.
  \begin{itemize}
    \item[\ref{req:fwbob-cov}:] Let $e$ s.t. $\tup{e, w} \in ([G.\lW^{\rel}];G.\lPO\rst{G.\lLOC} \cup [G.\lF];G.\lPO) \suq G.\lFWBOB$ and $e \nin \CoveredSet$.
      Since $G.\lPO^{?}(r, e)$ and $w \in G.\lW$, we know that $\tup{r, w} \in G.\lPO^?;G.\lFWBOB\suq\lFWBOB^{+}\suq G.\lAR^{+}$.
      It follows that $\tup{G.\lRFE^{-1}(r), w} \in G.\lRFE; G.\lBOB^{+}; [G.\lW] \suq \lAR^+$. It means $G.\lRFE^{-1}(r) \in \IssuedSet$.
      It contradicts that $r$ cannot be covered.
    \item[\ref{req:ppo-iss}:]
    \item[\ref{req:acq-iss}:]
      Let $r' \in G.\lR$ be an event such that $\tup{r', w} \in G.\lPPO \cup [G.\erlab{\acq}{}{}{}];G.\lPO$. If $G.\lRFE^{-1}(r') \neq \bot$,
      then $\tup{G.\lRFE^{-1}(r'), w} \in G.\lRFE; [G.\lR]; (G.\lPPO \cup [G.\erlab{\acq}{}{}{}];G.\lPO); [G.\lW] \suq \lAR^+$.
      It means $G.\lRFE^{-1}(r') \in \IssuedSet$.

      Let $w', r'$ be events such that  $\tup{w', r'} \in G.\lDETOUR$ and
      $\tup{r', w} \in G.\lPPO \cup [G.\erlab{\acq}{}{}{}];G.\lPO$,
      then $\tup{w', w} \in G.\lDETOUR; [G.\lR];(G.\lPPO \cup [G.\erlab{\acq}{}{}{}];G.\lPO); [G.\lW] \suq G.\lDETOUR; G.\lAR^+ \suq G.\lAR^+$.
      It means $w' \in \IssuedSet$.


    \item[\ref{req:w-strong-iss}:] Let $w'$ be an event such that $\tup{w', w} \in [G.\lWstrong];G.\lPO$.
      We know that $w' \in \IssuedSet$ since $\tup{w', w} \in G.\lAR^+$.
      \qed



  \end{itemize}
\end{itemize}  
\end{proof}

\begin{proposition}
  \label{prop:trav-to-etrav}
  Let $G$ be an \IMM-consistent execution, $\tup{\CoveredSet, \IssuedSet}$---its traversal configuration,
  Then, if there exist $\CoveredSet'$ and $\IssuedSet'$ such that
  $G \vdash \tup{\CoveredSet, \IssuedSet} \travConfigStep \tup{\CoveredSet', \IssuedSet'}$,
  then there exist $\CoveredSet''$ and $\IssuedSet''$ such that
  $G \vdash \tup{\CoveredSet, \IssuedSet} \etravStep \tup{\CoveredSet'', \IssuedSet''}$.
\end{proposition}
\begin{proof}
  Consider cases.
  If $\CoveredSet' = \CoveredSet \uplus \{e\}$ for some $e$, there are two cases to consider.
  \begin{itemize}
    \item $e \nin \dom{G.\lRMW}$: Then $G \vdash \tup{\CoveredSet, \IssuedSet} \etravStep \tup{\CoveredSet \uplus \{e\}, \IssuedSet}$ holds.
    \item $\exists w. \; \tup{e, w} \in G.\lRMW$:
      Then $G \vdash \tup{\CoveredSet, \IssuedSet} \etravStep \tup{\CoveredSet \uplus \{e\}, \IssuedSet'}$ holds,
      where either $w \in \IssuedSet$ and $\IssuedSet' = \IssuedSet$, or $w \in \lW^{\rel}$ and $\IssuedSet' =\IssuedSet \uplus \{w\}$.
  \end{itemize}
  If $\IssuedSet' = \IssuedSet \uplus \{e\}$ for some $e$, then
  $G \vdash \tup{\CoveredSet, \IssuedSet} \etravStep \tup{\CoveredSet', \IssuedSet \uplus \{e\}}$ holds,
  where either $w \nin G.\lW^\rel$ and $\CoveredSet' = \CoveredSet$, or $w \in G.\lW^{\rel}$ and $\CoveredSet' =\CoveredSet \uplus \{w\}$.
\end{proof}

\end{document}